\documentclass[draftcls,12pt,onecolumn]{IEEEtran}
\usepackage{amsmath, amsfonts,amssymb, latexsym,epsfig,color,rotating,
times,float,subfigure,algorithm,verbatim,afterpage}
\usepackage{epstopdf,amsthm}

\newcommand{\bver}{\bar{\ver}}
\newcommand{\ver}{\nu}

\newcommand{\ys}{{y^*}}

\renewcommand{\r}{r}
\newcommand{\Epzero}{\E^\mu_{\pi_0}}

\newcommand{\bpi}{\bar{\pi}}

\newcommand{\q}{{q}}
\newcommand{\D}{\mathcal{\bar{P}}}
\newcommand{\hy}{\eta}

\newcommand{\Pb}{\bar{\mathcal{P}}}
\newcommand{\f}{\mathbf{f}}
\renewcommand{\i}{\mathbf{i}}
\renewcommand{\j}{\mathbf{j}}
\newcommand{\Cb}{\bar{C}}

\newcommand{\Vb}{\bar{V}}
\newcommand{\pizero}{\pi_0}
\newcommand{\Stop}{\mathcal{S}}

\newcommand{\uJ}{\underline{J}}
\newcommand{\umu}{\underline{\mu}}
\newcommand{\uV}{\underline{V}}
\newcommand{\uQ}{\underline{Q}}

\renewcommand{\H}{\mathcal{H}}
\newcommand{\G}{\mathcal{G}}
\newcommand{\Ep}{\E^\mu_{\pi_0}}

\newcommand{\cb}{\bar{c}_a}

\newcommand{\F}{\mathcal{F}}

\newcommand{\Bs}{R^\pi} 
\newcommand{\Bsl}{R^l}
\newcommand{\Bsi}{R}

\newcommand{\T}{T^\pi}
\newcommand{\Tk}{T^{\pi_{k-1}}}

\newcommand{\sigs}{\sigma}

\newcommand{\sigp}{\sigma}
\newcommand{\priv}{\pi^P}

\newcommand{\discount}{\rho}

\newcommand{\glX}{\geq_{L_X}}

\newcommand{\Mu}{\boldsymbol{\mu}}
\newcommand{\U}{\mathbb{U}}
\newcommand{\A}{\mathbb{A}}
\newcommand{\X}{\mathbb{X}}
\newcommand{\Y}{\mathbb{Y}}
\newcommand{\ta}{\tilde{a}}

\newcommand{\ca}{c_a}

\renewcommand{\S}{\mathbf{S}}
\newcommand{\ones}{\mathbf{1}}

\newcommand{\tpi}{\tilde{\pi}}

\renewcommand{\l}{\mathcal{L}}

\newcommand{\p}{\prime}

\renewcommand{\l}{\mathcal{L}}

\newcommand{\gtp}{\underset{\text{\tiny TP2}}{\geq}}

\newcommand{\gl}{\geq_{L_1}}

\newcommand{\glx}{\geq_{L_X}}

\newcommand{\E}                 {\Bbb{E}}
\renewcommand{\P}                 {\Bbb{P}}
\newcommand{\reals}{\mathbb{R}}
\newcommand{\I}{\Pi(X)}

\renewcommand{\u}{\{1,2\}}
\newcommand{\argmin}{\operatorname{argmin}}
\newcommand{\ole}{\stackrel{\triangle}{=}}

\newtheorem{theorem}{Theorem}
\newtheorem{corollary}{Corollary}

\newtheorem{result}{Result}
\newtheorem{lemma}{Lemma}
\newtheorem{definition}{Definition}

\newcommand{\gr}{\geq_r}
\newcommand{\lr}{\leq_r}
\newcommand{\gs}{\geq_s}

\newcommand{\bp}{{\bar{\pi}}}

\newcommand{\beq}{\begin{equation}}
\newcommand{\eeq}{\end{equation}}
\renewcommand{\qed} {{$\hfill\blacksquare$}}

\makeatletter
\newtheorem*{rep@theorem}{\rep@title}
\newcommand{\newreptheorem}[2]{%
\newenvironment{rep#1}[1]{%
 \def\rep@title{#2 \ref{##1}}%
 \begin{rep@theorem}}%
 {\end{rep@theorem}}}
\makeatother

\newreptheorem{theorem}{Theorem}
\newreptheorem{lemma}{Lemma}

\begin{document}

\title{Quickest Detection with  Social Learning: Interaction of  local and global decision makers}
\author{Vikram Krishnamurthy  {\em Fellow, IEEE} 
\thanks{This work was partially supported by NSERC.}
\thanks{V. Krishnamurthy is
 with the Department of Electrical and Computer
Engineering, University of British Columbia, Vancouver, V6T 1Z4, Canada. 
(email:  vikramk@ece.ubc.ca).}}

\maketitle

\begin{abstract} We consider how local and global decision policies interact in stopping time problems such as quickest time
change detection.
Individual agents make myopic local decisions via
social learning, that is,
each agent records a private  observation of a noisy
underlying state process,
 selfishly optimizes
its local utility
and then broadcasts its local decision. Given these local decisions,
how can a global decision maker achieve quickest time  change detection when the underlying  state changes  according to  a phase-type distribution? The paper presents four results.
First, using Blackwell dominance of measures, it is shown that
 the optimal cost incurred in social learning based quickest detection is always larger than that of classical quickest detection.
Second, it is shown that in general
 the optimal  decision policy for social learning based quickest detection is characterized by 
multiple thresholds within the space
of Bayesian distributions. Third, using lattice programming and stochastic dominance,
sufficient conditions are given for the optimal decision policy to consist
of a single linear hyperplane, or, more generally, a threshold curve. Estimation of the
 optimal linear approximation to this threshold curve is formulated as a simulation-based stochastic optimization problem. 
Finally, the paper   shows that in multi-agent sensor management  with quickest detection, where each
 agent views the world according to its
 prior, the optimal policy has a similar  structure to  social learning.

\end{abstract}


\section{Introduction} \label{sec:intro}

Classical Bayesian quickest time detection  \cite{Shi63,Shi78}
involves detecting a geometrically distributed change time by  optimizing the tradeoff
between false alarm frequency and  delay penalty. 
The literature is vast, with applications in biomedical signal processing, machinery
monitoring and finance \cite{PH08,BN93,Mou86,Shi78}, see also  \cite{PPA89} for team detection, and
\cite{Yak97,YKP99}.
Classical quickest detection  can be formulated
as the following sequential protocol involving   a countable number of agents: Suppose
each agent acts once in a pre-determined sequential order indexed by $k=1,2,\ldots$.
Agent $k$ receives an observation of the underlying state at time $k$ and computes the posterior probability that the state has changed.
It then reveals this posterior probability to subsequent agents. This process repeats until  a stopping time
at which the global decision maker
 announces a change. It is well known \cite{Shi63,Shi78} that the optimal policy  to  declare a change has a threshold (monotone) structure: 
if the posterior probability (belief state)  exceeds a  threshold, then a change is announced; otherwise agents continue making observations.

\subsection{Context} 
Motivated by understanding how local decisions  affect global decision-making in multi-agent systems, this paper considers
a  generalization of the above classical quickest detection setup.  Given   local decisions from agents that are performing social learning,  how can a 
global decision maker achieve quickest time change detection?  In other words, how can a stochastic control problem
(stopping time problem)
be solved to make global decisions based on local decisions of agents? We consider phase-type distributed change times
and interaction between local and global decision-makers as outlined in the following two examples:

{\em  Example 1. Social Learning based Quickest-time detection}: Suppose that a multi-agent system  performs
social learning \cite{Cha04}\footnote{Another way of viewing the  social learning model
is that there are finite number of agents that act repeatedly in some pre-defined order. If each
 agent picks its local decision using the current public belief, then the setup is identical to the social learning
 setup.
We also refer reader to \cite{AO10,AOT10} for several recent results in social learning
over several types of network adjacency matrices.} to estimate an underlying state as follows: Just as in the
classical quickest detection protocol above, agents  act sequentially in a pre-determined order.
 However, instead of revealing its posterior distribution of change, each agent reveals
 its local decision to subsequent agents. The  agent chooses its
local 
decision  by  optimizing a local utility function (which depends on the public belief of the state and its  local observation).  Subsequent agents update their public belief based on these local
decisions (in a Bayesian setting), and the sequential procedure continues.
Given these local decisions, how can such a multi-agent system detect a change in the underlying state
and make a global decision to stop?

{\em Example 2: Quickest-time detection with adaptive sensing}:  Consider a multi-sensor system where each adaptive sensor is equipped with
a local sensor manager (controller). The multi-sensor system acts sequentially
as follows: Based on the existing belief of the underlying state, the local  sensor-manager chooses (adapts) the 
sensor mode e.g., low resolution or high resolution. The sensor then views the world based on this
mode. 
Given the  belief states and local sensor-manager decisions, how can such a multi-agent system achieve quickest time change detection?\footnote{The information flow patterns of Example 1 and 2 are similar. 
In Example 1, the sequence of events is
$\text{prior} \rightarrow \text{observation} \rightarrow \text{local decision} \rightarrow \text{posterior}$.
In Example 2, the sequence of events is
$\text{prior} \rightarrow \text{local decision} \rightarrow \text{observation} \rightarrow \text{posterior}$.}
Quickest detection  with such sensor management is of importance in automated tracking and surveillance systems
\cite{ASS02,AT96,CZK07}. In such cases, if 
individual agents or  cluster heads  are polled sequentially (e.g. round-robin fashion) then
the resulting dynamics are very similar to the social learning setup.

Classical quickest 
detection   
 is a trivial case of the above examples where  agents reveal their local observation (instead of local decision) to subsequent agents.
The above examples are non-trivial generalizations 
due to the interaction of  the local and global decision makers\footnote{A signal processing interpretation of social learning is as follows. Instead of using the posterior distribution to achieve quickest time detection, the
decision maker (or individual agents) computes the maximum aposteriori (MAP) estimate of the underlying state at each time instant. Given these hard decision MAP state estimates (local decisions), how can the global decision maker achieve quickest  change detection?}.  In both examples, the local decision determines the belief state which determines
the global decision (stop or continue) which determines the local decision at the next time instant and so on.
This interaction of local and global decision-making 
leads to discontinuous dynamics for the posterior probabilities (belief state) and unusual behavior as  outlined below.
We will show that the optimal decision policy has multiple thresholds and the stopping regions are non-convex.

Fig.\ref{fig:redgreen}(a) gives a  visual description of the optimal policy of social learning based quickest detection.
It illustrates a 
{\em triple threshold policy} for  geometric distributed change time.
Complete details of this numerical example are given in Sec.\ref{sec:numerical}.
The horizontal axis $\pi(2)$ is the posterior probability
of no change.
The vertical axis denotes the optimal decision:   $u=1$ denotes stop and declare change, while
 $u=2$ denotes continue.
The multi-threshold behavior of Fig.\ref{fig:redgreen}(a) is unusual: if it  is optimal to declare a change for a particular posterior probability, it may not be optimal to declare a change  when the posterior probability of a change is larger!
 Thus, the global decision (stop or continue) is a non-monotone function of the posterior probability obtained from local decisions.
Fig.\ref{fig:redgreen}(b) shows the associated value function obtained via stochastic dynamic programming.
Unlike standard sequential detection problems where the value function is concave, the figure shows 
that the value function is non-concave and
discontinuous.
To summarize, Fig.\ref{fig:redgreen}  shows
that  social learning based quickest detection results in fundamentally different decision policies compared to classical quickest time detection (which has a single threshold). Thus making global decisions (stop or continue) based on local decisions (from social learning) is non-trivial.

\begin{figure}\centering
\mbox{\subfigure[Optimal global decision policy $\mu^*(\pi)$]
{\epsfig{figure=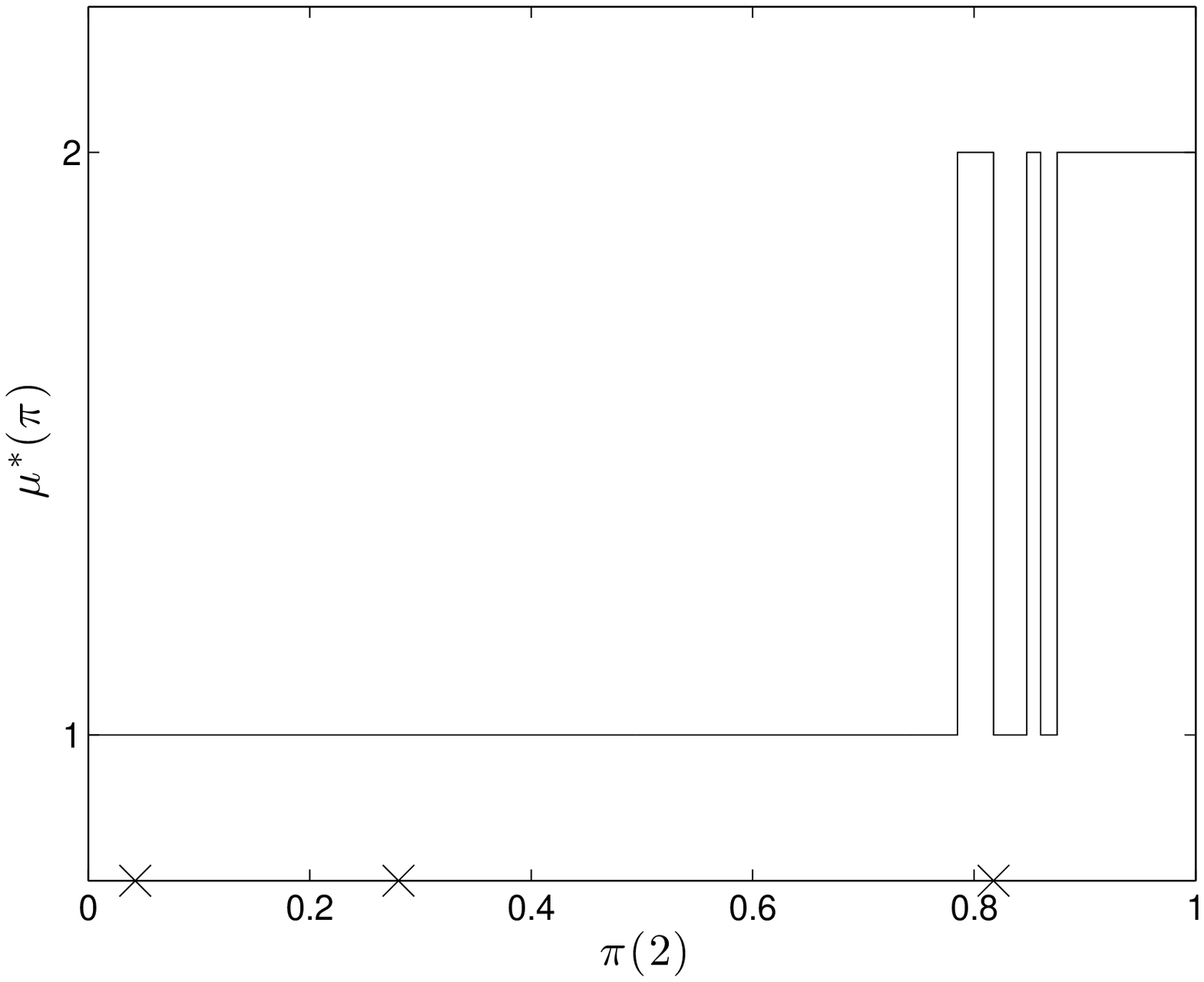,width=0.45\linewidth}} \quad
\subfigure[Value function $V(\pi)$ for  global decision policy]
{\epsfig{figure=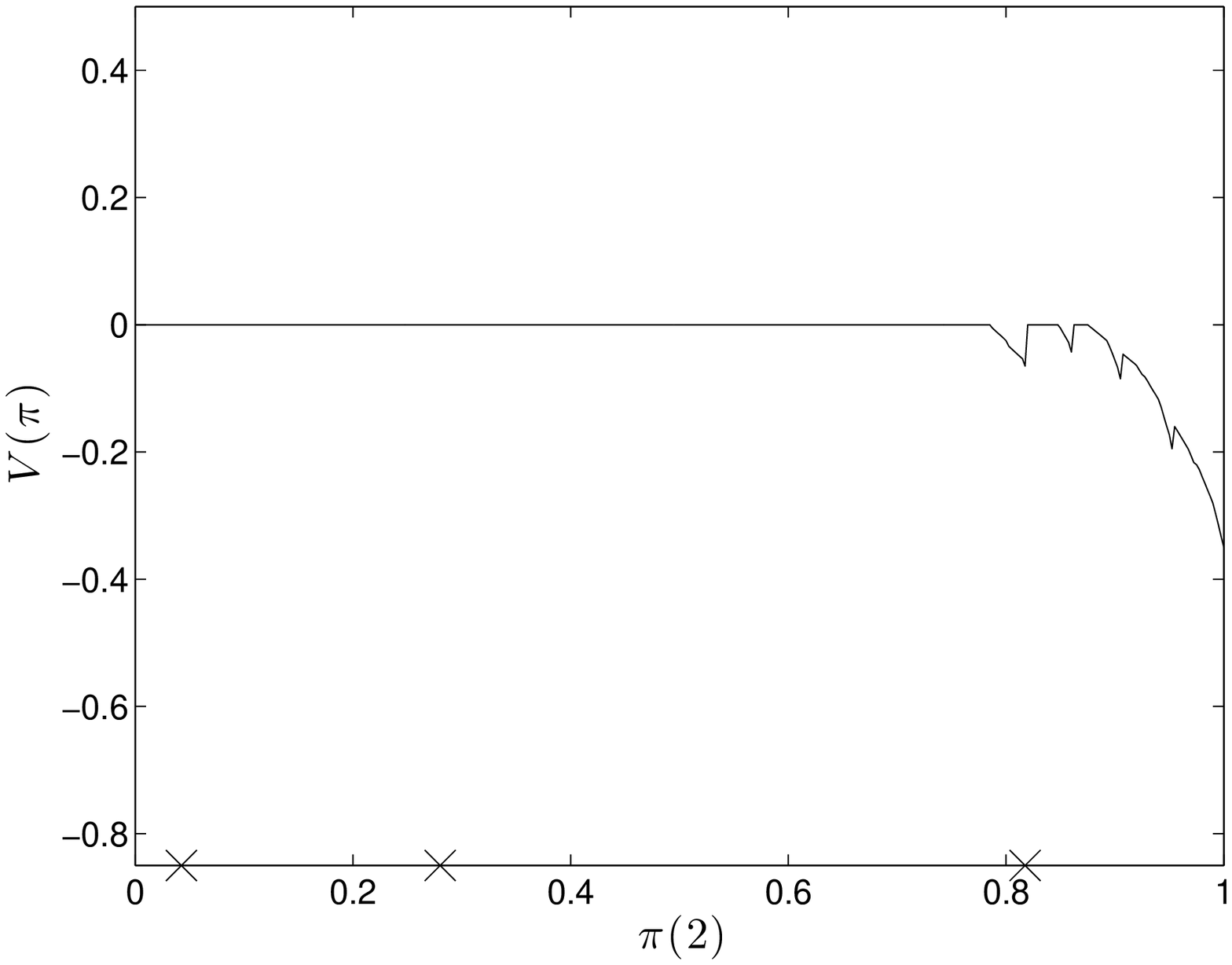,width=0.45\linewidth}}  }
\caption{Optimal decision policy for social learning based quickest time change detection 
for geometric distributed  change time, see
Example 1 of Sec.\ref{sec:numerical} for details.  The optimal policy $\mu^*(\pi)$ is characterized by a triple threshold. The value function $V(\pi)$ is non-concave and discontinuous.}
\label{fig:redgreen}
\end{figure}

\subsection{Motivation and Related Works}  
{\em Social Learning}:
In the last decade, social learning has been studied widely in economics to model the behavior of financial markets, 
crowds and social networks, see \cite{AO10,AOT10,Cha04,SS00,LADO07} and numerous references therein.
The social  learning framework
is similar to Hellman's and Cover's seminal papers \cite{CH70,HC70} which analyze learning with limited memory.
\cite[Chapters 3 and 4]{Cha04}
gives an excellent exposition of social learning.
An important result in social learning \cite{Ban92,BHW92} is that if the underlying state is a random variable and the observation
and local decision spaces are finite, then
agents eventually herd and end up making  the same local decision  irrespective of their observation.  Such information cascades
have been used in \cite{Cha04} to model sequences of financial trades, crashes and booms, and  auctions.
There is
strong motivation to understand the interaction of local and global decision makers in social learning.
Global decision making with social learning has recently been studied by several economists; for example \cite{CG94,BV02,Cha04,SS97,Kri11}  describe how information externalities affect global and local decision making
in social learning.
 The current paper
can be viewed as addressing a related problem: if individual agents make (simple) decisions by optimizing
a local utility, how can the global system achieve the (complex) task of detecting a change.
In a non-Bayesian setting such problems of designing sophisticated global behavior given simple local behavior
have also been studied in game-theoretic learning \cite{HM00,Har05,KMY08} involving correlated equilibria.

{\em PH-distributed change time}:  This paper deals with quickest detection for PH-distributed change times.
 PH-distributions are used widely 
 in queuing theory \cite{Neu89} and include geometric distributions as a special case.
The optimal detection of a PH-distributed change point is useful since the family of all PH-distributions forms a dense subset for the set of all distributions,
i.e., for any given distribution function $F$ such that $F(0) = 0$, one can find a sequence of PH-distributions	
$\{F_n , n	\geq	1\}$	 to	uniformly	approximate	$F$	over $[0, \infty)$;	see	\cite{Neu89}. Therefore there is strong
motivation to analyze quickest detection
with PH-distributed change times and social learning. Quickest time change detection for  PH-distributed change times is analyzed in \cite{Kri11}. The current paper generalizes these results to include social learning. A systematic investigation of the statistical properties of PH-distributions can be found in \cite{Neu89}.

\subsection{Main Results and Organization}
This paper deals with characterizing the structure of the global quickest-time change detection policy in  multi-agent
systems where individual agents make local myopic decisions when performing social learning.
The main results and organization of the paper are as follows:\\
1. {\em Multi-agent Protocol}: Sec.\ref{sec:prob} presents the multi-agent social learning protocol. The quickest time detection
problem is formulated.  We also point out in (\ref{eq:ksd}) the difference between the social learning model
and the classical Kolmogorov-Shiryaev  model for quickest change detection.
\\
2. {\em Dynamic Programming Formulation and Dominance of Classical Detection}:
In Sec.\ref{sec:dpf},  the optimal stopping policy is characterized in terms of stochastic dynamic programming.
It is shown that the value function is in general non-concave. Also  Theorem \ref{thm:blackwell} uses Blackwell ordering of measures to show that the optimal cost incurred in social learning based  quickest detection
is always
larger than classical quickest detection. Although such a result might appear intuitive (decision making using social learning is based on less information than classical quickest detection), the proof
is nontrivial. One needs to show that the expected cost of the entire trajectory of a stochastic
dynamical system (driven by the social learning protocol) is larger than that of classical quickest detection. 
 \\
3. {\em Main assumptions and Multi-threshold Policies}: Sec.\ref{sec:threshold} starts
with the main assumptions required to analyze the structure of the optimal quickest detection policy.
These assumptions allow us to  decompose the belief
space into polytopes (Theorem \ref{lem:polytopes1}).  On each of these polytopes, the conditional probability of a local decision given the underlying
state and posterior distribution is a constant.

The main result of Sec.\ref{sec:threshold} is to
characterize quickest time change detection
policies when the probability of change, denoted $\epsilon$, is small.  When the probability of change equal to zero, Theorem~\ref{thm:ep0}   characterizes explicitly
the multi-threshold structure of the optimal decision policy and 
non-concave behavior of the value function for sequential detection of a fixed state.
Then 
Corollary \ref{cor:qdslow}  shows that the optimal quickest-time detection policy  for change probability $\epsilon$,
yields a cost that is within $O(\epsilon)$ 
of the optimal cost for zero change probability.
An important ingredient in the proof of this
result is characterization of fixed points of the social learning filter update (Lemma \ref{lem:fixed}) which
also characterizes regions where the agents form information cascades in social learning. \\
4. 
{\em Phase-type Distributed Change Times}: 
The next main result is to 
 is to characterize the optimal policy of the global decision maker to achieve quickest time detection when the change time  has a   {\em phase-type} (PH) distribution and
 individual agents are performing social learning.
As mentioned above, PH-distributions can approximate arbitrary distributions and so are widely used in discrete-event systems.

A PH-distributed change time can be modelled
as a multi-state Markov chain with an absorbing state, see  \cite{Kri11} and also \cite{Neu89} for a systematic description.
(For a 2-state Markov chain, the PH-distribution
specializes to the geometric distribution).
So  for quickest time detection with PH-distributed change time, the
belief states  (Bayesian posterior)  lie in  a multidimensional simplex of probability
mass functions. 

 {\em Under what conditions will there  exist   a threshold stopping  policy for quickest detection with PH-distributed change time and
 social learning?}  Under what conditions for the geometric change time case does the optimal policy coincide with the classical
 Kolmogorov-Shiryaev model?

To answer these questions, the main results of Sec.\ref{sec:ph} are as follows:\\
(i) Theorem \ref{thm:cpi} gives sufficient conditions under which the optimal decision policy
for the global decision maker is myopic and characterized by a linear threshold hyperplane
in the multidimensional simplex. For the geometric case, this results yields an identical threshold to the Kolmogorov-Shiryaev model. 
\\
(ii) Theorem \ref{thm:1} gives sufficient conditions so that the optimal decision policy is characterized by a single switching curve
in the multidimensional simplex. The result uses  lattice programming   \cite{Top98} and structural results involving
monotone likelihood ratio stochastic orders \cite{Rie91,KW09}, and a novel modification of it.
 The result is useful because it implies that the global decision to stop
can be implemented efficiently at each agent. Each agent simply needs to compare its belief state with respect to the
threshold curve (in terms of a monotone likelihood ratio partial order on the space of posterior distributions).
Theorem \ref{thm:dep} gives sufficient
conditions on the optimal linear approximation to this curve that preserves the monotone likelihood ratio increasing
structure of the optimal decision policy. This linear approximation  can be estimated
via simulation based stochastic optimization.
\\
5. {\em Multi-agent Quickest Time Detection with active sensing}: Sec.\ref{sec:sensor}  considers  multi-agent quickest time detection  outlined in Example~2 above.
We show that  the optimal policy is similar to that in social learning based quickest detection.

\section{Social Learning Model and Protocol for Quickest Time Detection}
\label{sec:prob}

In this section, the multi-agent social learning model is presented in Sec.\ref{sec:prot1}. This constitutes the {\em local} decision-making
framework for estimating an underlying state.
Then Sec.\ref{sec:qdform} formulates the  costs incurred  by the {\em global} decision maker in
quickest time detection.  Sec.\ref{sec:quick}
presents the global quickest time detection objective.  
Finally, Sec.\ref{sec:summary} summarizes
 the entire social learning quickest detection model.

 
\subsection{The  Multi-agent Social Learning Model} \label{sec:prot1}

Consider a countably infinite number of agents\footnote{As mentioned earlier, the same setup holds if a finite number of agents are  polled repeatedly  in some pre-defined order, providing each agent picks
its local decision based on the most recent public belief.} performing social learning  to estimate an underlying state process $x$.
Each agent acts once  in a predetermined sequential order indexed by $k=1,2,\ldots$. The index $k$ can also be viewed
as the discrete time instant when agent $k$ acts.

Let $y_k \in \Y =  \{1,2,\ldots,Y\}$ denote the local (private) observation of agent $k$ and $a_k \in \A =  \{1,2,,\ldots, A\}$ denote the local
decision agent $k$ takes. Define the sigma algebras:
\begin{align}  \H_k & \quad 
\sigma\text{-algebra generated by } (a_1,\ldots,a_{k-1},y_k),  \nonumber\\
\G_k & \quad
\sigma\text{-algebra generated by } (a_1,\ldots,a_{k-1},a_k).   \label{eq:sigg} \end{align}

The social learning model \cite{BHW92,Cha04} comprises of the following  ingredients:

1. {\em Absorbing-state Markov chain and Phase-Type Distribution Change Times}: The state $x_k$ represents the underlying process
that changes at time $\tau^0$.
We  model the change point  $\tau^0$  by a   
{\em  phase type (PH) distribution}. 
As mentioned in Sec.\ref{sec:intro}    PH-distributions form a dense subset for the set of all distributions
	\cite{Neu89} and so can be used to  approximate   change times with arbitrary distribution. This is done by
	constructing a multi-state  Markov chain as follows:
Assume the underlying state  $x_k$ evolves as a Markov chain on
the finite state space $\X = \{1,\ldots,X\}$. Here state `1' is an absorbing state
and denotes the state after the jump change.  The states $2,\ldots,X$ can be viewed as  a single composite state that $x$ resides in before the jump.

The initial distribution is
$ \pi_0 = (\pi_0(i), i \in \X)$, $  \pi_0(i) = 
P(x_0 = i)$. We are only interested in the case where the change occurs after a least one measurement, 
so  assume $\pi_0(1) = 0$.
 So the
transition probability matrix $P$ is of the form
\beq \label{eq:phmatrix}
P = \begin{bmatrix}  1 & 0 \\ \underline{P}_{(X-1)\times 1} & \bar{P}_{(X-1)\times (X-1)} \end{bmatrix}
\eeq
Let the ``change time" $\tau^0$ denote the time at which $x_k$ enters the absorbing state 1,
i.e., 
\beq \tau^0 = \inf\{k: x_k = 1\} . \label{eq:tau0} \eeq
The distribution of the change  time $\tau^0$ is  equivalent to the
distribution of the absorption time to state 1 and is given by
\beq \label{eq:nu}
 \nu_0 = \pi_0(1), \quad \nu_k = \bar{\pi}_0^\p \bar{P}^{k-1} \underline{P}, \quad k\geq 1 \eeq
  where $\bar{\pi}_0 = [\pi_0(2),\ldots,\pi_0(X)]^\p$.
So by appropriately choosing the pair $(\pi_0,P)$ 
and  state space dimension $X$,
 one can approximate any given discrete distribution on $[0, \infty)$ by the distribution $\{\nu_k, k \geq 0\}$; see 
\cite[240-243]{Neu89}. To ensure that $\tau^0$ is finite, we assume states $2,3,\ldots,X$ are transient.
  In the special case when $x$ is a 2-state Markov chain,
 the change time $\tau^0$ is geometrically distributed.


2. {\em Local Observation}:
Agent's $k$ local  (private) observation $y_k  \in \Y=\{1,\ldots,Y\}$ 
 is obtained
from the observation likelihood distribution 
\beq B_{xy}= P(y_k=y|x_k=x) . \label{eq:B} \eeq  The states $2, 3,.\ldots, X$ are fictitious and are  defined to generate  the PH-distributed change time $\tau^0$.
So   states  $2, 3,.\ldots, X$ are indistinguishable in terms of the observation $y$.
That is, $P(y|2) = P(y|3) = \cdots = P(y|X)$ for all $y\in \Y$.

{\em 3. Private belief}: Using local observation $y_k$, agent $k$ updates its private belief  $\priv_k$
defined as 
\beq \priv_k = 
 (\priv_k(i), \;i \in \X), \quad\priv_k(i) = \E\{I(x_k =i)| \H_{k}\} = P(x_k = i| a_1,\ldots,a_{k-1},y_k),\text{ initialized by   $\pi_0$. } \label{eq:mudef}
\eeq
Thus the private belief is the posterior distribution of the underlying state given the past local decisions and current observation.
It   is computed by agent $k$ according to the following Hidden Markov Model (HMM)  filter:
\begin{align}  \label{eq:privu} \priv_k &= T(\pi_{k-1},y_k),  \text{ where }
T(\pi,y) = 
\frac{B_y P^\p \pi}{\sigp(\pi,y)},
\; \sigp(\pi,y) = \mathbf{1}^\p B_y   P^\p \pi. 
\\
B_y  &= \text{diag}(B_{1y},\ldots, B_{Xy})  \quad  (X\times X \text{ diagonal matrix for each $y \in \Y$) } \nonumber
 \end{align} 
Also $\pi_{k-1} $ denotes the public belief available at time $k-1$ (defined in Step 5 below).

4. {\em Agent's local decision}: Agent  $k$ then makes local decision $a_k\in \A = \{1,2,,\ldots, A\}$ to  minimize myopically its expected cost.
To formulate this, let $c(i,a)$ denote the non-negative cost incurred if the agent picks local decision $a$ when the underlying state is $x=i$.
Denote the local decision $X$-dimensional cost vector 
\beq
\ca = \begin{bmatrix} c(1,a) & c(2,a)  & \cdots & c(X,a) . \end{bmatrix}  . \label{eq:ca} \eeq
Then agent $k$ chooses  local decision $a_k$ greedily to minimize its expected cost:
\beq 
a_k = a(\pi_{k-1},y_k) = \arg\min_{a \in \A} \E\{c(x,a)|\mathcal{H}_k\}  =\arg\min_{a\in \A} \{c_a^\p\priv_k\}
\label{eq:step2} \eeq
 In quickest change detection, since states $2,3,\ldots,X$ are indistinguishable in terms of observation $y$, we assume that 
$c(2,a) = c(3,a) = \cdots =c(X,a)$  for each $a\in \A$.
 
 5.  {\em Social learning Public Belief}:
 Finally agent $k$ broadcasts its local decision $a_k$. Subsequent agents $\bar{k} > k$ use decision $a_k$ to update their
 public belief of the underlying state $x_k$ as follows:
Define   the public belief  $\pi_k$  as the posterior distribution
 of the state $x$ given all local decisions taken up to time $k$. 
\beq \label{eq:pidef}
\pi_k =  \E\{x_k | \mathcal{G}_{k}\} = (\pi_k(i), \;i \in \X), \quad
\pi_k(i)  = P(x = i|a_1,\ldots a_k), \quad \text{initialized by   $\pi_0$. }
\eeq
Then agents $\bar{k} > k$
 update their public belief according to the following ``social learning Bayesian filter":
\beq \pi_k = \Tk(\pi_{k-1},a_k), \text{ where } \T(\pi,a) = 
 \frac{\Bs_a P^\p \pi}{\sigs(\pi,a)},
\; \sigs(\pi,a) = \mathbf{1}_X^\p \Bs_a  P^\p \pi
\label{eq:piupdate} \eeq
We 
use the notation $\T(\cdot)$ to point out that the above Bayesian update map depends explicitly on the belief state $\pi$.
(For notational simplicity we have chosen not to use the superscript $\pi$ for $\sigma(\pi,a)$).
This is a key  difference compared to  the HMM filter (\ref{eq:privu}) where the Bayesian update map $T(\cdot)$ does not depend explicitly
on 
belief state  $\pi$.
In (\ref{eq:piupdate}),  $\Bs_a$ denotes the diagonal matrix $\Bs_a =  \text{diag}(\Bs_{i,a},\, i\in \X )$ where
\beq \Bs_{i,a} =  P(a_k=a|x_k=i,\pi_{k-1}=\pi) \label{eq:bs}\eeq
 denotes the conditional probability that agent
$k$ chose local decision  $a$ 
given state $i$.  We call $\Bs_{i,a}$ as the {\em local decision likelihood probabilities} in analogy to  observation likelihood probabilities 
$B_{iy}$ (\ref{eq:B}) in classical filtering.

Clearly  observing the local decision $a_k$ taken by agent $k$ yields  information about its
local observation $y_k$. That is, $a_k$ serves as a surrogate observation of the underlying state $x_k$.
The following lemma summarizes how subsequent agents use $a_k$ to  compute  the local decision likelihood probabilities $\Bs_{ia}$ in the social
learning filter.
The proof  is straightforward and omitted.

\begin{lemma}  \label{lem:disc} The local decision likelihood probability matrix $\Bs$ in the social learning Bayesian filter (\ref{eq:piupdate}) is computed as 
\beq\Bs =  B M^\pi \text{ where } 
M^\pi_{y,a} \ole 
P(a|y,\pi)=
\prod _{\ta \in \A - \{a\}}I(c_a^\p B_{y} P^\p \pi < c_{\ta}^\p B_{y} P^\p\pi). \label{eq:aprob}  \eeq
Here $\Bs$ is a  $\Y \times \A$ matrix,
 $B,B_y$ are the private observation probabilities defined in (\ref{eq:B}), (\ref{eq:privu}), $c_a,c_{\ta}$ are the local cost vectors defined in (\ref{eq:ca}), and  $I(\cdot)$ denotes the indicator function. \qed
\end{lemma}
\vspace{1cm}
The main implication of Lemma \ref{lem:disc} is that the social learning Bayesian filter (\ref{eq:piupdate})  is discontinuous in the belief state
$\pi$,  due to the presence of indicator functions in (\ref{eq:aprob}).
The likelihood probabilities $\Bs$ in (\ref{eq:bs}) are an explicit function of the belief state $\pi$ -- this is stark contrast
to the  standard quickest detection problems where the observation distribution is not an explicit
function of the posterior distribution.

{\em  Summary}:   A key aspect of  the information pattern in  the above social learning protocol
 is that  agent $k$ does not have access to the private belief state $\priv_{k-1}$ or private
observations of previous agents. Instead each agent  $k$ only has access to the local decisions taken by previous agents together with its
own current private observation $y_k$.  
The fact that  the likelihood probabilities $\Bs$ is an explicit function of the public belief state $\pi$ (see  (\ref{eq:aprob})) is an important aspect of social learning that is not present in classical
sequential detection problems. It makes the Bayesian update of the public belief discontinuous with $\pi$ and makes our proofs substantially harder than standard concavity arguments in classical quickest detection problems.

 {\em Belief State Space}: Before proceeding with the quickest time detection formulation, we briefly describe the space in which
 the public belief $\pi$ defined in (\ref{eq:pidef}) lives.  The public belief 
 belongs to the
unit $X-1$ dimensional simplex denoted as
\begin{align}
\I \ole \left\{\pi \in \reals^{X}: \mathbf{1}_{X}^{\p} \pi = 1,
\quad 
0 \leq \pi(i) \leq 1 \text{ for all } i \in \X \right\} .\label{eq:Pi}
\end{align}
So for geometric-distributed change times, the belief state space $\Pi(2)$ is the interval $[0,1]$.
For PH-distributed change times, the belief space $\I$ is a multi-dimensional simplex.
For  example, $\Pi(3)$ is a two-dimensional unit simplex (equilateral triangle); $\Pi(4)$
is a tetrahedron, etc.
The vertices of the unit simplex $\I$ are the unit $X$-dimensional vectors  $ e_1,\ldots,e_X $, where
\beq \label{eq:ei}
\text{ $e_i$ denotes  the unit  vector with $1$ in the $i$th position,
$i \in \X$.}  \eeq  Of course the private belief $\priv$ (\ref{eq:mudef}) also lives in $\I$.

\subsection{Quickest Time Detection: Costs Incurred by Global Decision Maker} \label{sec:qdform}

With the above social learning based local decision framework, we now formulate the quickest time detection problem
faced by the global decision maker.  At each time $k$, 
given the public belief $\pi_k$, let $u_k$ denote
 the  global decision taken: 
\beq \label{eq:actionpolicy}
 u_k = \mu(\pi_k)  \in \{1 \text{ (announce change and stop)} ,2 \text{ (continue) }\}. \eeq
 Thus the global decision  $u_k$  is  $\G_k$  measurable, where $\G_k$ is defined in (\ref{eq:sigg}).
In (\ref{eq:actionpolicy}),  the policy $\mu$  belongs to the class of stationary decision
policies denoted $\Mu$.
Below we formulate the costs incurred when taking these global decisions $u_k$.

(i) {\em  Cost of announcing change and stopping}:  
If global decision $u_k=1$  is chosen, then the social learning protocol of Sec.\ref{sec:prot1} terminates. If  $u_k=1$ is chosen before  the change point $\tau^0$, then a false
alarm penalty is incurred.
The false alarm event  $\cup_{i\geq 2}   \{x_k =  i\} \cap \{u_k = 1\} = \{x_k\neq 1\} \cap \{u_k = 1\}$ represents the event
that a change is announced before the change happens at time $\tau^0$.
To evaluate the {\em false alarm penalty}, 
 let  $f_i I(x_k=i,u_k=1)$ denote the cost of a false alarm in state $i$, $i \in \X$, where
$f_i \geq 0$. Of course,
$f_1 = 0$ since a false alarm is only incurred if the stop action is picked in states $2,\ldots, X$.  The expected false alarm penalty is 
\beq \Cb(\pi_{k},u_k=1) =  \sum_{i \in \X} f_i \E\{I(x_k=i,u_k=1)|\G_k\} = \f^\p \pi_k , \quad
\text{ where }
  \f = (f_1,\ldots,f_X)^\p, \; f_1 = 0. \label{eq:cp1}\eeq
The false alarm  vector $\f$ is chosen with increasing elements so that states  further
from  state~1 incur larger  penalties. (Obviously $f_i \geq 0$ since $f_1 = 0$).

(ii) {\em Delay cost of continuing}: If global decision $u_k=2$ is taken then the social learning protocol of Sec.\ref{sec:prot1} continues to
time $k+1$.
A delay cost is incurred when the event  $ \{x_{k} = 1, u_k = 2\}$ occurs,
i.e., no change is declared at time $k$, even though the state has changed at time $k$.
The expected delay cost is
 \beq  \Cb(\pi_{k},u_k=2) = d\,\E\{I(x_{k} = 1, u_k=2) | \G_k\}
= d e_1^\p \pi_k  \label{eq:cp2} \eeq
where $d > 0$ denotes the delay cost and $e_1$ is defined in (\ref{eq:ei}).

\noindent  {\em Remarks}: (i)  Recall that the public belief state $\pi$ depends on the local decisions $a$. Also the choice of global decision
$u$ determines when the local decision process terminates. This links
the local and global decision makers.  \\
(ii) The above costs (\ref{eq:cp1}), (\ref{eq:cp2})  should be viewed as an example only.
The results of this paper also apply to more general stopping time problems with minor modifications if the global decisions  $u_k$ are $\H_k$ measurable (instead of $\G_k$
measurable), where $\H_k$ and $\G_k$  are defined in (\ref{eq:sigg}).
More generally, $\Cb(\pi,u)$ can also include the local decision cost incurred in social learning, see
remark at the end of Sec.\ref{sec:tcurve}.

\subsection{Quickest Time Detection Objective}\label{sec:quick}
Let $(\Omega,\mathcal{F})$ be the underlying measurable space where $\Omega = (\X \times \U \times \Y)^\infty$ is the product space, which is endowed with the product topology and $\mathcal{F}$  is the corresponding product sigma-algebra. For any $\pi_0\in \I$,  and policy $\mu \in \Mu$,
 there exists a (unique) probability measure $\P^\mu_{\pi_0}$ on  $(\Omega, \mathcal{F})$, 
   see 
 \cite{HL96} for details. Let $\Epzero$  denote the expectation with respect to the measure  $\P^\mu_{\pi_0}$.
  
 Let $\tau$ denote  a stopping time adapted to the sequence of 
 $\sigma$-algebras $\G_k,k\geq 1$, see (\ref{eq:sigg}).
 That is, with $u_k$ determined by decision policy (\ref{eq:actionpolicy}),  
\beq
\tau = \{ \inf k:  u_k = 1\} .   \label{eq:tauu}\eeq
 For each initial distribution $\pi_0 \in \I$, and policy
 $\mu$, the following cost  is associated:
\beq \label{eq:csdef}
J_\mu(\pizero) = \Ep\{\sum_{k=1}^{\tau-1} \discount^{k-1} \Cb(\pi_{k},u_k=2)
+  \rho^{\tau-1} \Cb(\pi_{\tau},u_\tau = 1)
 \}.
\eeq
Here $\discount \in [0,1]$ denotes an economic discount factor.
Since  $\Cb(\pi,1)$, $\Cb(\pi,2)$ are non-negative and bounded for all $\pi \in \I$, stopping is guaranteed
in finite time, i.e., $\tau$ is finite with probability 1 for any $\discount  \in [0,1]$ (including $\rho = 1$). 

\noindent {\bf Kolmogorov--Shiryaev criterion}:  Suppose   $\X = \{1,2\}$ implying
that the change time $\tau^0$ is geometrically distributed. Choose the false alarm vector $\f= f_2 e_2=
[0, f_2]^\p$ where $f_2$
is a positive constant,
 delay cost (\ref{eq:cp2}), and discount factor $\rho =1$. Then the quickest time objective (\ref{eq:csdef})
assumes the classical Kolmogorov--Shiryaev
criterion for detection of disorder \cite{Shi63}:
 \beq J_\mu(\pizero) =   d \Ep\{(\tau - \tau^0)^+\} +  f_2 \,\P^\mu_{\pi_0}(\tau < \tau^0) .
\label{eq:ksd} \eeq
However, unlike classical quickest detection, the posterior (public belief) $\pi$ has  discontinuous dynamics given by the social learning Bayesian
filter (\ref{eq:piupdate}). (Recall from (\ref{eq:piupdate}), (\ref{eq:aprob}) that the dynamics of public
belief $\pi$ depend on the local decision costs $\ca$). \qed

The goal of the global decision maker is to determine the change time $\tau^0$   with minimal cost, that is, compute the optimal global
decision policy $\mu^* \in \Mu$ to minimize (\ref{eq:csdef}),
where $$J_{\mu^*}(\pizero) = \inf_{\mu \in \Mu} J_\mu(\pizero). $$ The existence of an optimal stationary policy $\mu^*$ follows from \cite[Prop.1.3, Chapter 3]{Ber00b}.

\subsection{Summary} \label{sec:summary}
In summary, the social learning based quickest detection problem with PH-distributed change time is specified by the model 
\beq \label{eq:model}
(P, B, c, C, \rho, \X,\Y,\A,\mathbf{u} ) \eeq
where $P$ is the transition probability matrix (\ref{eq:phmatrix}), $B$ is the private observation matrix (\ref{eq:B}), 
$c$  are the  local decision costs  (\ref{eq:ca}), $C$  defined in (\ref{eq:costdef})  is  the transformed global decision cost vector for quickest detection (in terms of false alarm $\f$ (\ref{eq:cp1})
and delay penalty $d$ (\ref{eq:cp2})), and $\rho \in [0,1]$ is the discount factor (\ref{eq:csdef}).
Also
 $\X$ is the state space,   $\Y$ is the private observation space,    $\A$ is the local decision space and  $\mathbf{u} = \{1 \text{ (stop) }, 2 \text{ (continue) } \}$ is the global
decision space.

\section{Stochastic Dynamic Programming Formulation and Dominance of Classical Quickest Detection}\label{sec:dpf}
Sec.\ref{sec:dp} formulates
the optimal decision policy for social learning based quickest detection as the solution of a stochastic dynamic programming problem.
 Sec.\ref{sec:nontrivial}  describes
why social learning based quickest detection is a non-trivial extension of the standard quickest detection problem.
 Finally, Sec.\ref{sec:black} presents our first structural result -- it uses Blackwell dominance of measures to show that optimal cost incurred in quickest time detection
with social learning is always larger than that with classical quickest detection.

\subsection{Stochastic Dynamic Programming Formulation} \label{sec:dp}
Given the stopping time problem (\ref{eq:csdef}), it is well known \cite{Lov87} that the optimal policy $\mu^*(\pi)$ can be expressed
as the solution of a stochastic dynamic programming problem in terms of the belief state $\pi$. Our characterization of the structure of the
optimal policy $\mu^*(\pi)$ will be based
on analyzing the structure of this dynamic programming problem.

The optimal stationary policy $\mu^*: \I \rightarrow \u$ and associated value function
 $\Vb(\pi)$ of the stopping time problem (\ref{eq:csdef}) 
are the solution of 
 ``Bellman's dynamic programming  equation'' 
\begin{align} \label{eq:dp_initial}
\mu^*(\pi)&= \arg\min\{ \Cb(\pi,1), \;\Cb(\pi,2)
+ \discount \sum_{a \in \A}  \Vb\left( \T(\pi ,a) \right) \sigp(\pi,a)\} , \quad J_{\mu^*}(\pi_0) = \Vb(\pi_0) \\
 \Vb(\pi) &= \min \{ \Cb(\pi,1),\; \Cb(\pi,2)
+ \discount \sum_{a \in \A}  \Vb\left( \T(\pi,a) \right) \sigp(\pi,a)\}.
  \nonumber
\end{align}
Here the global decision maker's costs $\Cb(\pi,u)$ are defined in (\ref{eq:cp1}), (\ref{eq:cp2}), $\T$ is the public belief Bayesian update
(\ref{eq:piupdate}), and the measure $\sigma(\pi,a)$ is defined in (\ref{eq:piupdate}).

For our subsequent analysis, it is convenient  to rewrite Bellman's equation as follows.
Define  the transformed value function and global decision costs $V(\pi)$, $C(\pi,1)$ and $C(\pi,2)$ as follows:
\begin{align}
V(\pi) &= \Vb(\pi) -   \f^\p \pi,
\quad C(\pi,1) =  0, \quad
C(\pi,2) = \Cb(\pi,2) -  \f^\p \pi + \rho \f^\p P^\p \pi = C^\p \pi  \label{eq:costdef} \\
\text{ where }  C &\ole d e_1 - (I - \rho P) \f  \text{   with elements denoted as $C_j$, $j=1,\ldots,X$}. \nonumber
 \end{align}
Then clearly $V(\pi)$ satisfies Bellman's dynamic programming  equation 
\begin{align} \label{eq:dp_alg}
\mu^*(\pi)&= \arg\min_{u \in \U} Q(\pi,u) , \;J_{\mu^*}(\pi_0) =V(\pi_0) , \quad
V(\pi) = \min_{u \in \{1,2\}} Q(\pi,u), \\
 \text{ where } & Q(\pi,2) =  C(\pi,2)
+ \discount \sum_{a \in \A}  V\left( \T(\pi ,a) \right) \sigp(\pi,a),\quad
Q(\pi,1) =   C(\pi,1) = 0
   \nonumber
\end{align} 
The above transformation\footnote{This transformation is used in \cite[pp.389]{HS84} to deal with stopping time problems.  As a result of this transformation, the initial condition of the value iteration algorithm is modified,
see (\ref{eq:vi}).} is convenient since the
transformed stopping cost $C(\pi,1) =0$ and  $C(\pi,2) = C^\p \pi$ in (\ref{eq:costdef}) captures all the costs involved in quickest detection. Of course,
the optimal policy $\mu^*(\pi)$ and hence stopping set $\Stop$ remain unchanged with this coordinate transformation.
The goal for the global decision-maker is to determine the optimal stopping set denoted $\Stop$. That is, $\Stop$
is  the set of public belief states $\pi$ for
which it is optimal to declare a change and stop:
\begin{align} 
\Stop &=  \{\pi \in \I : \mu^*(\pi) = 1\} = \{\pi \in \I:  C(\pi,1) \leq C(\pi,2)
+ \discount \sum_{a \in \A}  V\left( \T(\pi ,a) \right) \sigp(\pi,a) \}  \nonumber \\
&=
\{\pi \in \I:  \Cb(\pi,1) \leq \Cb(\pi,2)
+ \discount \sum_{a \in \A}  \Vb\left( \T(\pi ,a) \right) \sigp(\pi,a) \}.
 \label{eq:stopset}
 \end{align}

\subsubsection*{Value Iteration Algorithm}

Let $k=1,2,\ldots,$ denote iteration number (the fact that we used
$k$ previously to denote time  should not result in confusion).
The   value iteration
algorithm is  a fixed point iteration of 
Bellman's equation (\ref{eq:dp_alg}) and proceeds as follows: $V_0(\pi) = -\Cb(\pi,1)$ and 
\begin{align}\nonumber
V_{k+1}(\pi) &= \min_{u \in \u} Q_{k+1}(\pi,u), \quad
\mu^*_{k+1}(\pi)= \argmin_{u \in \u} Q_{k+1}(\pi,u) \quad
\pi \in \I,\\
\text{ where } Q_{k+1}(\pi,2) &=  C(\pi,2) 
+ \discount \sum_{a \in \A}  V_k\left( T(\pi,a) \right) \sigma(\pi,a),
\;
Q_{k+1}(\pi,1) = C(\pi,1) = 0.  \label{eq:vi}
\end{align}
Let $\mathcal{B}(X)$ denote the set of bounded real-valued functions on $\I$.
Since $C(\pi,1)$, $C(\pi,2)$, $\pi \in \I$, are  bounded, 
the value iteration algorithm (\ref{eq:vi}) will generate a sequence of lower semi-continuous value functions
$\{V_k\} \subset \mathcal{B}(X)$ that will converge pointwise
as $k\rightarrow \infty$ to $V(\pi) \in \mathcal{B}(X)$, the solution of Bellman's equation,
see \cite[Prop.1.3, Chap 3, Vol.2]{Ber00b}

Since the belief state space $\I$ in (\ref{eq:Pi}) is a unit simplex,  the value iteration algorithm  (\ref{eq:vi})
does not 
yield a practical solution methodology for computing stopping set $\Stop$  since 
$V_k(\pi) $ needs to be evaluated on the continuum $\pi \in \I$.
Although Bellman's equation and the value iteration algorithm is not useful from a computational point of view, in subsequent
sections, we exploit its
structure 
 to characterize  the stopping set $\Stop$ in (\ref{eq:stopset}).
We then exploit this structure to devise stochastic gradient  algorithms
for approximating the optimal  policy $\mu^*$ and thus determining the stopping set $\Stop$.

\subsection{Why Social Learning based Quickest Detection  is non-trivial}\label{sec:nontrivial}
Let us illustrate why social learning based quickest detection results in a non-trivial behavior.
We will show in Sec.\ref{sec:threshold} that the belief space $\I$ can be decomposed into
$Y+1$ polytopes denoted $\mathcal{P}_1,\ldots,\mathcal{P}_{Y+1}$ such that on each of these polytopes $\mathcal{P}_l$,
the belief state update $\T(\pi,a) = T^l(\pi,a)$.
Consider the value iteration algorithm (\ref{eq:vi}) which is used as a basis for mathematical induction to prove properties associated with
Bellman's equation (\ref{eq:dp_alg}).  It can be expressed as\footnote{Note that from (\ref{eq:vi}),
$V_k(\pi)$ is positively homogeneous, that is, for any $\alpha > 0$, $V_k(\alpha \pi) = \alpha V_k(\pi)$.
So choosing $\alpha = \sigma(\pi,a)$ which is the denominator term of $\T$ in  (\ref{eq:piupdate}) yields the expression in the second equality
of (\ref{eq:vinc}).}
\begin{align} V_{k+1}(\pi) &= \min\{C^\p \pi + \rho \sum_a \sum_{l=1}^{Y+1}V_k(T^l(\pi,a)) \sigma(\pi,a) I(\pi\in  \mathcal{P}_{l}) , 0 \}  \nonumber\\
&= \min\{C^\p \pi + \rho \sum_a \sum_{l=1}^{Y+1}V_k(\Bsl_a P^\p \pi) I(\pi\in  \mathcal{P}_{l}) , 0 \}
\label{eq:vinc} \end{align}
It should be clear from (\ref{eq:vinc}) that  if $V_k(\pi)$ is assumed to be concave on $\I$, $V_{k+1}(\pi)$ is not necessarily
concave on $\I$. In fact, even if $V_k(\pi)$ is assumed to be concave
in just one of the polytopes, say polytope $\mathcal{P}_l$,   then $V_{k+1}(\pi)$ is not necessarily
concave on  $\mathcal{P}_l$, since $T^l(\pi,a)$ in (\ref{eq:vinc}) may map two distinct belief states in polytope $\mathcal{P}_l$ to
two different polytopes. As will be shown in numerical examples,  in general $V(\pi)$ will be discontinuous and non--concave.

Classical quickest detection problems are special instances of  partially observed Markov  decision process (POMDP)
stopping time problems \cite{Kri11}. In POMDPs, the belief state update $\T$ is not an explicit function of belief
state $\pi$ since  the observation probabilities  are not an explicit function of $\pi$.  For such  POMDP stopping
time problems 
the value iteration algorithm reads\footnote{We use the notation $\uV(\pi)$ to denote the value function of the classical stopping problem. This will be defined formally in Sec.\ref{sec:black} where we will show $\uV(\pi) \leq V(\pi)$, i.e., quickest detection
with social learning always incurs a higher optimal cost than classical quickest detection.}
$$\uV_{k+1}(\pi) = \min\{C^\p \pi + \rho \sum_y \sum_{l=1}^{Y+1}\uV_k(B_y P^\p \pi)  , 0 \}$$
and is to be compared with (\ref{eq:vinc}).
Since the composition of a concave function with a linear function preserves concavity, it is easily 
seen that if $\uV_k(\pi)$ is piecewise linear and concave, then so is $\uV_{k+1}(\pi)$.
So by mathematical induction on the value iteration algorithm, and since the sequence $\{\uV_k(\pi)\}$ converges pointwise 
(actually uniformly for POMDPs) to $\uV(\pi)$,  the value function $\uV(\pi)$ is concave and the stopping set $\Stop$ is a convex
(and therefore connected) set \cite{Lov87a}.  The key difference in the above
social learning quickest detection formulation is that the local decision likelihoods $\Bs$  (\ref{eq:aprob})  and therefore social learning
filter $\T$ are explicit and discontinuous functions of 
$\pi$. This results in a possibly
 non-concave value
function $V(\pi)$ making determining $\Stop$ non-trivial.

\subsection{Quickest Time Detection with Social Learning is More Expensive}\label{sec:black}
This section presents our first main result. We prove that
quickest detection with social learning is always more expensive than classical quickest detection.
In social learning, agents have access to local decisions of previous agents instead of the 
actual observations. Thus one would expect intuitively that this information loss results in less
efficient quickest time change detection compared to classical quickest detection.
Here we confirm this intuition.
 The main idea is to use Blackwell dominance of observation measures.

\subsubsection{Notation}
First define the optimal policy and cost in classical quickest time detection.
Similar to (\ref{eq:dp_alg}), the optimal policy $\umu^*(\pi)$ and cost $\uV(\pi)$ incurred in classical quickest detection, satisfies the following Bellman's equation:
\begin{align} \label{eq:dp_algc}
\umu^*(\pi)&= \arg\min_{u \in \U} \uQ(\pi,u) , \; \uJ_{\mu^*}(\pi_0) =  \uV(\pi_0), \;  \uV(\pi) = \min_{u \in \{1,2\}} \uQ(\pi,u),\\
 \text{ where }  \uQ(\pi,2) &=  C(\pi,2)
+ \discount \sum_{y \in \Y}  \uV\left( T(\pi ,y) \right) \sigp(\pi,y),\quad
\uQ(\pi,1) =   C(\pi,1) = 0
   \nonumber
\end{align} 
Recall $T(\pi,y)$ is the Hidden Markov Model Bayesian filter defined in (\ref{eq:privu}). Thus the only difference between the classical
and social learning quickest detection problems is the update of the belief state, namely (\ref{eq:privu}) in the classical setup versus
(\ref{eq:piupdate}) in the social learning formulation.

\subsubsection{Main Result}
The following theorem says that if the initial belief state is chosen from any of the polytopes $\mathcal{P}_{\ys}, \ldots
 \mathcal{P}_{Y+1}$,  the optimal detection policy with social learning incurs a higher cost than classical quickest detection.

\begin{theorem}\label{thm:blackwell} Consider the social learning quickest time detection problem $(P, B, c, C, \rho)$ in (\ref{eq:model}) and associated value function $V(\pi)$ in (\ref{eq:dp_alg}). Consider also the classical quickest detection
problem with value function $\uV(\pi)$ in (\ref{eq:dp_algc}).
 Then for any initial
 belief state $\pi \in \I$, the optimal cost incurred by classical quickest detection
is smaller than that of quickest detection with social learning. That is, $\uV(\pi) \leq V(\pi)$.
\end{theorem}

 Since the theorem holds for the case $A=Y=2$ (equal number of local decision choices and observation symbols),
 a naive explanation that information is lost due to using fewer symbols in $\A$ compared to $\Y$ is not true.
 
  The proof of Theorem \ref{thm:blackwell}  is given  in Appendix \ref{sec:pblackwell}.  Recall from (\ref{eq:aprob})  that $\Bs = B M^\pi$ where 
 $B$ and $M^\pi$ are stochastic matrices.  Thus 
observation $y$ with conditional distribution specified by $B$ is 
said to be more  informative than (Blackwell dominates) observation $a$ with conditional distribution $\Bs$,
see \cite{Rie91}. The main idea in the proof is that under the assumptions of Theorem~\ref{thm:blackwell},
the value function $\uV(\pi)$ is concave for $\pi \in \I$. Then the result is established
 using Jensen's inequality together
with Blackwell dominance  on the  Bellman's equation. value iteration algorithm proves the result.
 
The first instance of a similar proof using Blackwell dominance for POMDPs
 was given in \cite{WD80}, see also \cite{Rie91}, where it was used to show optimality of certain myopic policies.
Our use of Blackwell dominance in Theorem \ref{thm:blackwell} is somewhat different since we are using it 
to compare the value functions of two different 
dynamic programming problems.
A useful consequence of  Theorem \ref{thm:blackwell} is that performance analysis of standard quickest detection problems \cite{TV05} readily
applies  to form a lower bound for the cost incurred in social learning based quickest detection.

\section{Assumptions and Quickest Detection with Small Change Probabilities} \label{sec:threshold}
 This section comprises of two parts.\\
(i)  Sec.\ref{sec:assumptions} lists the main  assumptions (A1), (A2), (S) which result in a natural partition   of
belief space  $\I$ into $Y+1$ convex polytopes with decision likelihoods  $\Bs$ (defined in (\ref{eq:aprob})) being a constant (with respect to $\pi)$ on
each polytope (Theorem \ref{lem:polytopes1}).  These polytopes play an important role in specifying the global quickest detection policy in the rest of the paper.\\
(ii) Sec.\ref{sec:convex} considers quickest time change detection with geometric distributed change time and gives explicit
conditions for the optimal policy to have a  double threshold. In particular, Theorem~\ref{thm:ep0} and Corollary \ref{cor:qdslow}  show that the optimal quickest-time detection policy  for change probability $\epsilon$,
yields a cost that is within $O(\epsilon)$ of the optimal cost for sequential detection of a constant state.

\subsection{Polytope Structure and Main Assumptions}  \label{sec:assumptions}


Since the public belief state $\pi \in  \I$ is continuum (see (\ref{eq:Pi})), as a first step in characterizing the optimal 
policy $\mu^*(\pi)$, we need to understand the structure of the decision likelihood probabilities $\Bs$ defined in (\ref{eq:aprob}).
Even though the belief state $\pi \in \I$ is continuum, it turns out that there are only $2^Y-1$
  possible local decision likelihood probability matrices $\Bs$.
Let $\mathcal{Q}_l$, $l=1,\ldots,2^Y-1$  denote the elements of the power set of $\Y$ (excluding, of course, the empty set).
Define the following  $2^Y-1$ convex polytopes $\mathcal{\Pb}_l$, $l=1,2,\dots, 2^Y-1$:
\beq \mathcal{\Pb}_l = \left\{ \pi \in \I :  \begin{cases} (c_1-c_2) ^\p B_y P^\p \pi < 0 &  y \in \mathcal{Q}_l \\
                                        (c_1 -c_2)^\p B_y P^\p \pi \geq 0 & y \in \Y - \mathcal{Q}_l  \end{cases} \right\} \label{eq:expopoly}\eeq
                                        Recall the local cost vectors $c_a$ are defined in (\ref{eq:ca}).
Then  from (\ref{eq:aprob}) it follows that $M^\pi$ and hence $\Bs$ is a constant on each polytope $\mathcal{Q}_l$. 
Specifically,
for rows $ y \in \mathcal{Q}_l$,
$M^\pi_{y1} = 1$  and for rows  $  y \in 
\Y - \mathcal{Q}_l$, $ M^\pi_{y2} = 1$.

 Although in general there are $2^{Y}-1$ possible 
 $\Bs$ matrices,
we now show that by introducing assumptions (A1), (A2) and (S)  below,  there are only
 $Y+1$ distinct local decision likelihood matrices $\Bs$. This forms an important preliminary step for characterizing the optimal global
 decision policy.

Recalling the notation in Sec.\ref{sec:prot1},  we list the following  assumptions.
\begin{itemize}

\item[(A1)] The  observation distribution
$B_{xy} = p(y|x)$ is TP2 (see Definition \ref{def:tp2} in Appendix \ref{sec:mlrdef}), i.e., all second order minors of matrix $B$ are non-negative.

\item[(A2)] The transition probability matrix $P$ is TP2. (All second order minors of $P$ are non-negative).

\item[(A3)]  The elements of vector $C$ in (\ref{eq:costdef}) are decreasing.
A sufficient condition is that for $j \geq i$ and $i \geq 2$ the false alarm vector $\f$ and delay penalty $d$ satisfy
$f_i \geq \max\{1, \rho \f^\p P^\p e_i - d\}$  and
$f_j - f_i \geq \rho \f^\p P^\p(e_j-e_i) $.

\item[(S)] The local decision cost vector $\ca$ in (\ref{eq:ca}) is submodular. That is,
the elements $c(i,a)$  satisfy
 $c(1,2) > c(1,1)$ and $c(2,2)< c(2,1)$. (Recall from Sec.\ref{sec:prot1} that $c(2,a) = c(3,a) = \cdots =c(X,a)$ in quickest detection
 problems with PH-distributed change time).

 \end{itemize}

{\em  Discussion of Assumptions}:\\
{\em Assumption (A1)}:  The requirement that $P(y|x)$ is TP2 with respect to states $\{1,2\}$ and $y \in \Y$
holds for numerous examples,
see  Karlin's classic book \cite{Kar68} and also \cite{KR80}.   Examples include quantized Gaussians, quantized exponential distributions,
 Binomial, Poisson, etc.
For example consider quantized Gaussians.
Suppose   $ B_{iy}=P(y|x=i) = \frac{\bar{b}_{iy}}{\sum_{y=1}^Y \bar{b}_{iy}}$  where 
$ \bar{b}_{iy} = \frac{1}{\sqrt{2 \pi \Sigma}}
\exp \biggl( - \frac{1}{2} \frac{(y - g_i)^2}{ \Sigma} \biggr) $, $\Sigma > 0$, 
and  $g_1 < g_2$. Then 
 (A1) holds.

{\em Assumption (A2)} always holds trivially for $X=2$. For $X>2$, see \cite{Gan60,Kij97} for numerous examples.
Consider the  tridiagonal transition probability matrix $P$ with
 $p_{ij} = 0$ for $ j\geq i+2$ and $j \leq i-2$. As shown in
 \cite[pp.99--100]{Gan60}, a necessary and sufficient condition for
 tridiagonal $P$ to be TP2 is that $p_{i,i} p_{i+1,i+1} \geq p_{i,i+1}
 p_{i+1,i} $. 
Such a diagonally dominant tridiagonal matrix 
satisfies 
Assumption (A2).

 Assumption (A3) is   a sufficient
condition for $C(\pi,2)$ to be decreasing in $\pi$ with respect to the monotone likelihood ratio order. We will use (A3) in Sec.\ref{sec:ph} to
obtain sufficient conditions for a threshold policy.
Assumption (A3) always holds for the 
geometric distributed change times  ($X=2$). For PH-distributed change times ($X > 2$), Assumption (A3) can be viewed as design constraints the decision maker needs
to take into account so that  quickest detection with PH-distributed change times has a threshold policy
\cite{Kri11}.
Feasible values for the elements of $\f$ are straightforwardly obtained using a LP solver such as {\tt linprog} in Matlab.

{\em Assumption (S)} is only required for the problem to be non-trivial.  If (S) does not hold and $c(i,1) < c(i,2)$ for $i=1,2$, then local decision $a=1$  will always
dominate  decision $a=2$ and the problem reduces
to a standard quickest detection problem where the observed local decision $a=1$ yields no information about the state.
Assumption  (S) implies $c(x,2) -c(x,1) $ is decreasing in $x \in \{1,2\}$, i.e., 
the local cost $c(x,a)$ is submodular
 which implies the zero crossing condition that is  important
in the proof of  Theorem~\ref{lem:polytopes1}.

The following theorem is an abbreviated version of Theorem \ref{lem:polytopes1} presented in Appendix \ref{sec:polytopes1}. It will be used in the rest of the paper as a natural partition of the belief state space
$\I$.
Recall that transition probability $P$, observation probability matrix $B_y$ and  local cost vector $\ca$  are  defined in
(\ref{eq:phmatrix}), 
 (\ref{eq:privu}), (\ref{eq:ca}) respectively.

\begin{theorem} \label{lem:polytopes1}
Under (A1),  (A2), (S), 
the belief state space $\I$ defined in (\ref{eq:Pi}) can be partitioned into at most $Y+1$ non-empty polytopes denoted $\mathcal{P}_1,\ldots,\mathcal{P}_{Y+1}$ where
\begin{align} \label{eq:reduced}
\mathcal{P}_1 &= \{\pi \in \I: (c_1 - c_2)^\p B_1 P^\p \pi \geq 0 \}  \\
\mathcal{P}_l &= \{\pi \in \I: (c_1 - c_2)^\p B_{l-1} P^\p \pi < 0\; \cap\;
(c_1 - c_2)^\p B_{l} P^\p \pi \geq 0 \},\; l = 2,\ldots,Y  \nonumber \\
\mathcal{P}_{Y+1} & = \{\pi \in \I: (c_1 - c_2)^\p B_Y P^\p \pi < 0 \} \nonumber \end{align}
On each such polytope, the local decision likelihood matrix $\Bs$ defined in (\ref{eq:aprob})  is a constant with respect to
belief state $\pi$. \qed
\end{theorem}

As a consequence of Theorem \ref{lem:polytopes1} and (\ref{eq:aprob}), there are only $Y+1$ possible
decision likelihood matrices $\Bs$, one per  polytope $ \mathcal{P}_l$,  $l=1,\ldots, Y+1$. We will
denote these decision likelihood matrices as
\beq \Bsl = \Bs = B M^l = B M^\pi ,  \pi \in \mathcal{P}_l, l=1,\ldots, Y+1 .
\label{eq:bsl}
\eeq

 {\em Example}:
To give some insight into the structure of  decision likelihood matrix $\Bs$, suppose  $X=2$ (state space), $Y=3$ (observation space), $A=2$ (local decision space). Then assuming (A1), (A2), (S), by Theorem~\ref{lem:polytopes1} there are up to $Y+1=4$ convex polytopes. The matrices
$M^l$ defined in (\ref{eq:aprob}), (\ref{eq:bsl})  are
\beq M^4 = \begin{bmatrix}  
1 & 0 \\ 1& 0 \\ 1 & 0
\end{bmatrix},\;
 M^3 = \begin{bmatrix}  
1 & 0 \\ 1& 0 \\ 0 & 1
\end{bmatrix},\;
 M^2 = \begin{bmatrix}  
1 & 0 \\ 0& 1 \\ 0 & 1
\end{bmatrix},\;
 M^1 = \begin{bmatrix}  
0 & 1 \\ 0& 1 \\ 0 & 1
\end{bmatrix} . \label{eq:Mexample}\eeq
 Then from (\ref{eq:bsl}) the 4 possible  decision likelihood matrices $\Bsl$ are
 \beq \label{eq:ex31}
\Bsi^1 = \begin{bmatrix} 0 & 1 \\ 0 & 1 \end{bmatrix},\;
\Bsi^2 = \begin{bmatrix} B_{11}  & B_{12} + B_{13} \\ 
					B_{21} & B_{22}+B_{23} \end{bmatrix},\;
\Bsi^3 = \begin{bmatrix}  B_{11} + B_{12} & B_{13} \\ 
					   B_{21} + B_{22} & B_{23} \end{bmatrix}, \;
\Bsi^4 = \begin{bmatrix} 1 & 0 \\ 1 & 0 \end{bmatrix}	.				   
					\eeq
The detailed version of Theorem \ref{lem:polytopes1} in Appendix \ref{sec:polytopes1} guarantees that each of these matrices is TP2. 
Fig.\ref{fig:expoly} illustrates  these  polytopes and hyperplanes $\eta_y$ defined below.

\begin{figure}
\mbox{\subfigure[$X=3, Y=4, A=2$]
{\epsfig{figure=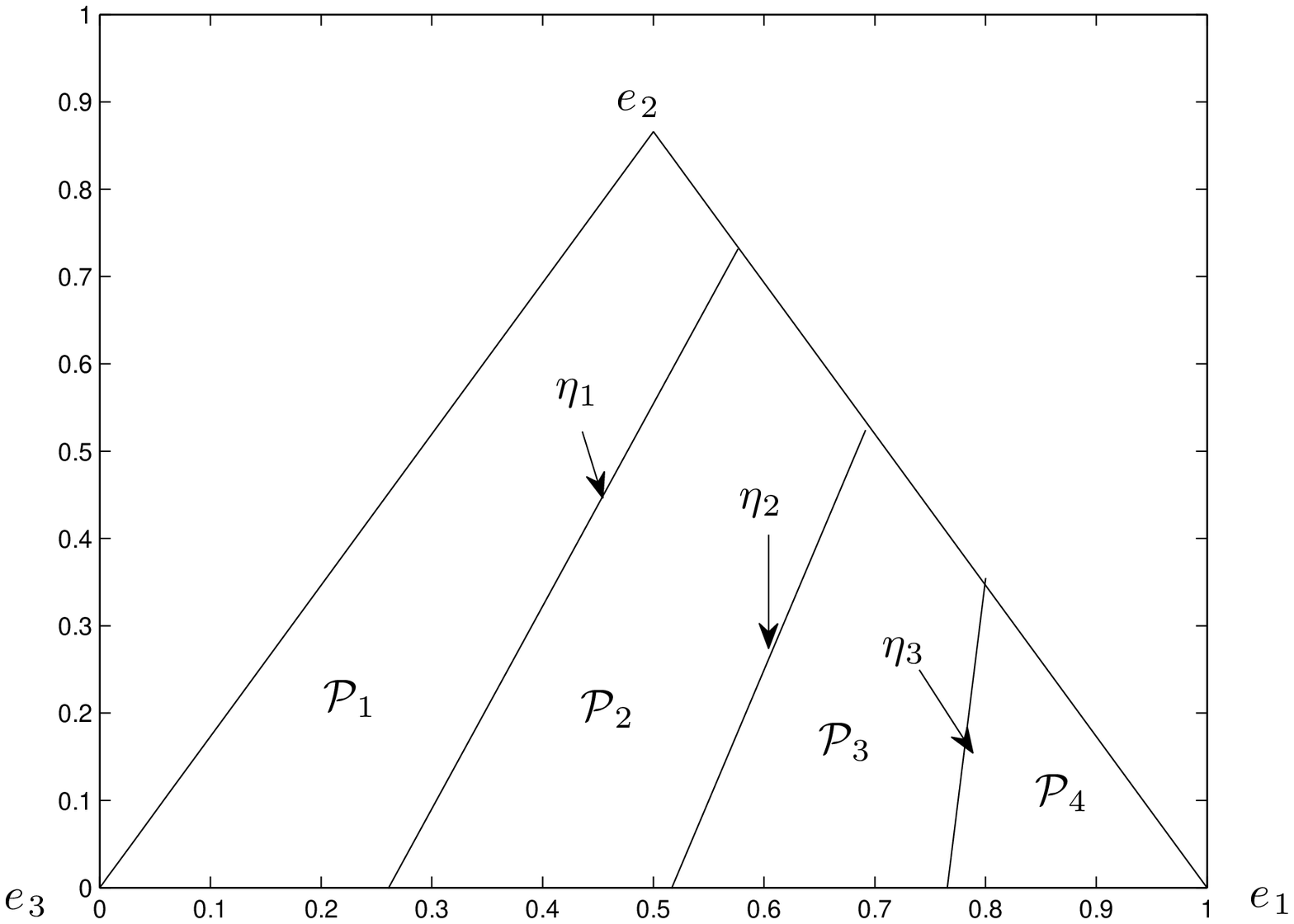,width=0.45\linewidth}} \quad
\subfigure[$X=2, Y=4, A=2$.]
{\epsfig{figure=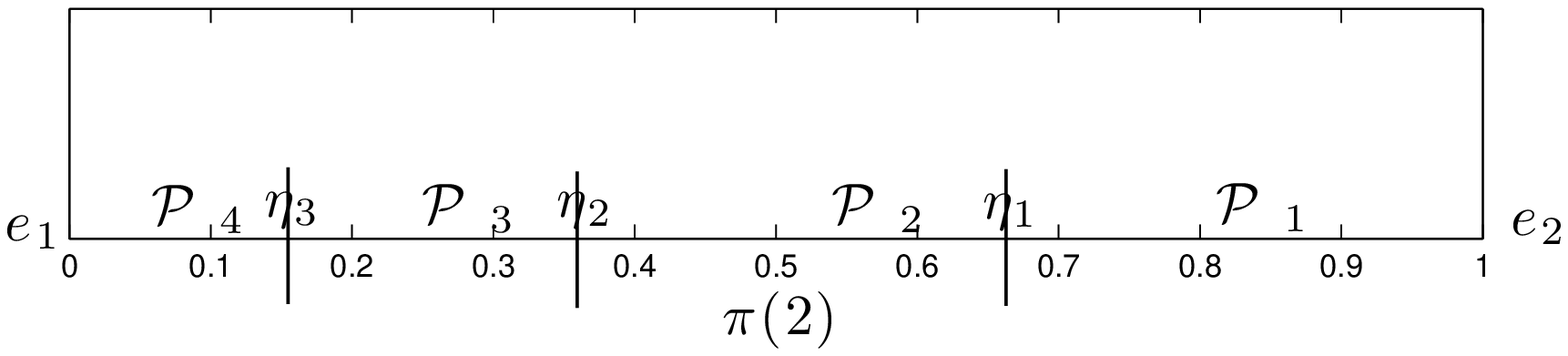,width=0.45\linewidth}}  }
\caption{Illustration of polytopes $\mathcal{P}_1,\mathcal{P}_2,\mathcal{P}_3,\mathcal{P}_4$ defined in (\ref{eq:reduced}) and hyperplanes
 $\eta_1,\eta_2,\eta_3$ defined in (\ref{eq:hy}) for $Y=4, A=2$. Theorem \ref{lem:polytopes1} ensures that the hyperplanes do not intersect within the simplex
$ \I$ and on each polytope, the local decision likelihoods $\Bs$ are a constant.
In the figure, $e_2,e_3 \in \mathcal{P}_1$ -- Assumption (PH)(ii) in Sec.\ref{sec:ph} ensures this.
} \label{fig:expoly}
\end{figure}

Let us give some intuition behind  Theorem \ref{lem:polytopes1}. Define the following $Y$  hyperplanes  that are subsets of $\I$:
\beq \label{eq:hy}
\hy_y = \{\pi \in \I: (c_1 - c_2)^\p B_{y} P^\p \pi =  0\} , \quad  y=1,\ldots,Y.
\eeq
The main intuition of the above theorem  is that (A1), (A2),  (S) imply that $(c_1-c_2)^\p B_y P^\p \pi$ satisfies a single crossing condition \cite{Ami05} with respect to $a,y$, see Definition \ref{def:scc}
in Appendix \ref{sec:mlrdef}.
This means
that the set of belief states satisfy the following subset property:
\beq \{\pi: (c_1-c_2)^\p B_y P^\p \pi \geq 0 \} \subseteq  \{\pi: (c_1-c_2)^\p B_{y+1} P^\p \pi \geq 0 \}.  \label{eq:scmean} \eeq
This  implies that the hyperplanes $\hy_y$, $y\in \Y$,  do not intersect within the simplex $\I$.  It is nice that 
straightforward conditions such as (A1), (A2), (S) ensure this. 
Otherwise dealing with intersecting hyperplanes in a multi-dimensional simplex can be a real headache. Theorem \ref{lem:polytopes1}(iv) 
in Appendix \ref{sec:polytopes1} shows that each hyperplane
$\hy_y$ partitions $\I$ such that vertices $e_1,e_2,\ldots,e_{i_y}$ lie on one
side and $e_{i_y+1},\ldots,e_X$ lie on the other side. In Sec.\ref{sec:ph}, we will introduce Assumption (PH)(ii) which ensures
that $e_2,\ldots,e_X$ always lie in polytope $\mathcal{P}_{1}$ as illustrated in Fig.\ref{fig:expoly}.


\subsection{Multi-Threshold Structure of Social Learning based Quickest  Detection}\label{sec:convex}
The main result (Theorem \ref{thm:ep0} and Corollary \ref{cor:qdslow}) below gives sufficient conditions under which social learning based quickest detection has a double
threshold policy.
Consider the model
$ (P, B, c, C, \rho)$ in (\ref{eq:model})
 with  geometric change time:
\beq \X = \Y= \A = \{1,2\}, \quad P= \begin{bmatrix} 1 & 0 \\ \epsilon & 1- \epsilon 
\end{bmatrix}
\label{eq:qdslow},\eeq 
with  $\f = \f = (0,\; f_2)^\p$ false-alarm vector in (\ref{eq:cp1}) and delay cost (\ref{eq:cp2}).
Here the change probability $\epsilon \ll 1 $ is a small non-negative scalar. 
So the change time $\tau^0$ is geometrically distributed with
$\E\{\tau^0\} = 1/\epsilon$.

The analysis in this subsection proceeds as follows:\\
{\em Step 1}: For $\epsilon = 0$, the problem becomes a simple sequential detection problem for state $1$ -- we  
explicitly 
characterize the multi-threshold behavior of the optimal decision policy in Theorem \ref{thm:ep0} below.\\
{\em Step 2}:
It is then shown that for small $\epsilon$, the optimal value function is within $O(\epsilon)$ of the value function
for the case of zero change probability (Corollary \ref{cor:qdslow}).  So, the optimal policy computed for zero change probability yields 
performance that is close to that of the optimal quickest detection policy for small $\epsilon$.


\subsubsection{Step 1: Sequential Detection of State 1}
In line with above plan,
consider  the sequential  detection problem for state 1 with social learning formulated  in Sec.\ref{sec:prob} with
\beq \X = \Y= \A = \{1,2\}, \quad P=I . \label{eq:special}\eeq 
The state $x$ is a random variable chosen at $k=0$ with distribution $\pi_0$ and remains constant
for $k>0$.
 The goal is to detect and announce state $1$ if $x_0 = 1$ based on noisy observations.
   The global decision $u_k = \mu(\pi_{k}) \in \{1 \text{ (stop) } ,2 \text{ (continue)}\} $ is a function of  the  public belief $\pi_{k}$.
   The optimal policy $\mu^*(\pi)$ that optimizes (\ref{eq:csdef}) satisfies Bellman's equation (\ref{eq:dp_alg}).

 The 2-dimensional  belief state $\pi = [1 - \pi(2),\pi(2)] $ is parametrized by the scalar $\pi(2) \in  [0,1]$, i.e., $\I $ is  the interval $[0,1]$.
Each hyperplane $\hy_y$ (\ref{eq:hy}) now is a point  on the interval $[0,1]$;  let  the 2-dimensional vector
$[1-\hy_y(2),\hy_y(2)]$ denote the belief 
state corresponding to  $\hy_y$. The  polytopes $\mathcal{P}_{1}$,  $\mathcal{P}_{2}$, $\mathcal{P}_{3}$
in Theorem \ref{lem:polytopes1} are now intervals which are subsets of $[0,1]$.
 If (A1) and (S)  hold, then 
 $\mathcal{P}_{3} = [0,\hy_2(2))$,  $\mathcal{P}_2 = [\hy_2(2),\hy_1(2))$, $\mathcal{P}_1 = [\hy_1(2),1]$.

   To handle  the discontinuity in the social learning filter (\ref{eq:piupdate}), 
we start   with the following lemma that characterizes useful structural properties of the social learning filter. First define the belief state
\beq \q = T^{\hy_1}(\hy_1,1). \label{eq:qstate} \eeq
 
 \begin{lemma}\label{lem:fixed}  Consider the social learning filter (\ref{eq:piupdate}) and assume
(A1), (S) hold.  Then: \\
(i)  $\q =T^{\hy_1}(\hy_1,1) = T^{\hy_2}(\hy_2,2)$.\\
(ii) If $B$ is symmetric,  then $\hy_1$ and $\hy_2$ are fixed points of the composite Bayesian map:
\beq \hy_1 = T^q(T^{\hy_1}(\hy_1,1),2), \quad \hy_2 = T^q((T^{\hy_2}(\hy_2,2),1)
\label{eq:symb} \eeq
 \qed
  \end{lemma}
  
      The implication of the  above lemma is that $\I$ can be partitioned into 4 intervals, namely $[e_1,\hy_2)$, $[\hy_2,\q)$, $[q, \hy_1)$ and $[\hy_1,e_2]$. Fig.\ref{fig:pub} illustrates these regions and the dynamics specified in Lemma \ref{lem:fixed}.
   The main result below characterizes the structure of the optimal global decision policy $\mu^*(\pi)$ on these 4 intervals.  The theorem also characterizes information cascades \cite{Cha04} (more colloquially ``herding")  which is a  salient feature of social learning.

\begin{figure}
\epsfig{figure=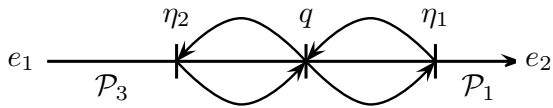,width=0.5\linewidth}
\caption{Structure of social learning filter under the assumptions of  Lemma \ref{lem:fixed} and symmetric $B$. Right (left) arrows represent evolution of the public belief
when $a=2$ ($a=1$). As can be seen, $\eta_1$, $\eta_2$ are fixed points of the composite maps in (\ref{eq:symb}).} \label{fig:pub}
\end{figure}

   \begin{theorem}  \label{thm:ep0} Consider the sequential detection problem with parameters (\ref{eq:special}).  Suppose agents
   make local decisions via social learning.
   Assume (A1), (S) hold. (Note (A2) holds trivially since $P=I$). The optimal global decision policy $\mu^*(\pi)$ has the following properties:\\
   (i) For  $\pi \in \mathcal{P}_1 \cup \mathcal{P}_3$, the global decision policy has a threshold structure:
\beq \mu^*(\pi) = \begin{cases} 2 & \text{ if } \pi(2) > \pi^*(2)
   \\ 1 & \text{ otherwise }  \end{cases}  \quad
   \text{ where } \pi^*(2) = \frac{d}{f_2 (1-\rho) + d} \label{eq:explicit}
  \eeq
   Also for $\pi \in \mathcal{P}_1 \cup \mathcal{P}_3$, the value function (\ref{eq:dp_alg})  
   is $V(\pi) = \min\{0, C(\pi,2)/(1-\rho)\}$ where $C(\pi,2)$ is defined in (\ref{eq:costdef}).\\
   (ii)
       The intervals $\mathcal{P}_1$ and  $\mathcal{P}_{3} $ are ``information cascades'' \cite{Cha04}. That is, if $\pi_k \in \mathcal{P}_1
 \cup \mathcal{P}_{3}$, then $\pi_{k+1} = \pi_k$ and social learning ceases.
   \\
   (iii) If $B$ is symmetric, then for $\pi \in \mathcal{P}_2$, the global decision policy has the following structure:\\
   (a) For  $\pi \in [\hy_2(2),\q(2))$, $V(\pi)$ is concave and there is at most one interval where $\mu^*(\pi) = 1$. \\
 (b) For   $\pi \in [q(2),\hy_1(2))$, $V(\pi)$ is concave and there is at most one interval where $\mu^*(\pi) = 1$.   
\qed
   \end{theorem}
   
\vspace{1cm}
The  implication of Part (iii) of the above theorem is that the stopping set $\Stop$ comprises of at most three intervals.
One of these intervals is $(\pi^*(2),1)$, with the threshold $\pi^*(2)$ defined in (\ref{eq:explicit}).
 The second claim of the theorem follows, since if public belief $\pi \in \mathcal{P}_1$, then the optimal local decision is $a=2$ irrespective of the observation $y$.
Similarly, if $\pi \in \mathcal{P}_3$, then the optimal local decision is $a=1$ irrespective of the observation $y$. Therefore when the public
belief is in  $\mathcal{P}_1
 \cup \mathcal{P}_{3}$, the local decision of an agent reveals no information about its local observation to subsequent agents.

\subsubsection{Step 2: Quickest Time Detection bound for small $\epsilon$}
Given the characterization in Theorem \ref{thm:ep0} of the optimal policy for $\epsilon= 0$, we now consider the quickest
 change detection problem for small $\epsilon$ specified in (\ref{eq:qdslow}). It is convenient
to introduce the following $\epsilon$ dependent notation.

Let $V_{\mu_\epsilon^*}(\pi)$ denote the cost incurred by the optimal policy $\mu_\epsilon^*$ with transition
matrix 
$P^\epsilon = 
\begin{bmatrix} 1 & 0 \\ \epsilon & 1- \epsilon \end{bmatrix} 
$.
We use the notation  $\mathcal{P}^\epsilon_l$ to denote the explicit dependence of the 3 intervals $\mathcal{P}_1$,
$\mathcal{P}_2$, $\mathcal{P}_3$, defined 
 in (\ref{eq:reduced}). For $\epsilon = 0$, we denote these intervals as  $\mathcal{P}^0_l$.
The following result bounds the difference between  $V_{\mu_0^*}(\pi)$ and $V_{\mu_\epsilon^*}(\pi)$.
Note that $\mu_0^*(\pi)$ is characterized in Theorem \ref{thm:ep0} and $P^0 = I$ (identity matrix).

Recall from (\ref{eq:dp_initial})  that $\Vb(\pi)$  is the actual optimal expected cost associated with optimal decision
policy $\mu^*(\pi)$. As mentioned below  (\ref{eq:dp_alg}), the transformed value function $V(\pi)$ is more convenient to deal with to prove
the existence of optimal threshold policies and  the optimal policy  remains invariant to 
the transformation from  $\Vb(\pi)$  to $V(\pi)$.

\begin{corollary} \label{cor:qdslow}
Consider the social learning based quickest detection model  $ (P, B, c, C, \rho)$ in (\ref{eq:model}) with probability of change specified
in (\ref{eq:qdslow}). Then, for initial belief $\pi \in \mathcal{P}_l^\epsilon \cap \mathcal{P}_l^0$, $l=1,2,3$, the optimal policy  $\mu_0^*$  (characterized in Theorem (\ref{thm:ep0}))
incurs a total global  cost $V_{\mu_0^*}(\pi)$ that constitutes an $O(\epsilon)$ upper-bound to the optimal global cost 
$\Vb_{\mu_\epsilon^*}(\pi)$
incurred in the quickest detection
problem. More specifically, for $\pi \in \mathcal{P}_l^\epsilon \cap \mathcal{P}_l^0$, $l=1,2,3$,
\beq
\Vb_{\mu_0^*}(\pi) - \Vb_{\mu_\epsilon^*}(\pi)  \leq \frac{4 \rho \epsilon}{(1-\rho)^2} \max (d,f_2) . 
\text{ \qed }  \label{eq:bound}
\eeq
\end{corollary}
{\em Discussion}: The  implication of (\ref{eq:bound})  is that the simple policy $\mu^*_0(\pi)$ of Theorem \ref{thm:ep0} is near optimal
for quickest time detection with social learning when $\epsilon$ is small. Note that (\ref{eq:bound}) compares  the optimal costs in
regions $\pi \in \mathcal{P}_l^\epsilon \cap \mathcal{P}_l^0$, $l=1,2,3$, so we are omitting intervals where the models have different local decision likelihood probabilities $\Bs$. The regions we are omitting
 are $O(\epsilon)$ in size.  
In each region $\pi \in \mathcal{P}_l^\epsilon \cap \mathcal{P}_l^0$,
the only difference between the quickest detection model and the simplified model is the
transition matrix ($P^\epsilon$ vs $P^0$). This allows us to give a  tight bound in the sense
that for $\epsilon = 0$, the optimal costs $\Vb_{\mu_0^*}(\pi)$ and $\Vb_{\mu_\epsilon^*}(\pi) $ coincide.
Of course, (\ref{eq:bound})  requires the discount factor
$\rho < 1$. We refer the reader to \cite{TV05} for an alternative and more general approach.

The proof of Corollary \ref{cor:qdslow} follows from Theorem 2 of \cite{RIM09}. In terms of our notation,
Theorem 2 of \cite{RIM09} shows
that for a POMDP with piecewise linear value function at each iteration of the value-iteration algorithm, for
$\pi \in \mathcal{P}_l^\epsilon \cap \mathcal{P}_l^0$, 
\beq
\Vb_{\mu_0^*}(\pi) \leq \Vb_{\mu_\epsilon^*}(\pi)  + \frac{2 \rho }{(1-\rho)^2} 
\|\Cb(\pi,u)\|_\infty
\sup_{i} \| [P^\epsilon- P^0]_{ij} \,\Bsl_a \|_1 \label{eq:bound2}
\eeq
where the $\|\cdot \|_1$ induced matrix norm is with respect to the $(j,a)$ elements. Since from Theorem~\ref{thm:ep0}, the value function is piecewise linear, (\ref{eq:bound2}) applies.
From the structure of $P^\epsilon$ in (\ref{eq:qdslow}) and since $P^0 = I$, clearly 
$$\sup_{i} \| [P^\epsilon -P^0]_{ij} \,\Bsl_a \|_1 = \epsilon \max(B_{11}+B_{21},\, B_{12}+B_{22}) \leq 2 \epsilon.$$
Also $\|\Cb(\pi,u)\|_\infty = \max(d,f_2)$.
Substituting these in (\ref{eq:bound2}) yields the bound  (\ref{eq:bound}).

\begin{figure}[h]\centering
\mbox{\subfigure[Optimal global decision policies $\mu^*_0(\pi)$ and $\mu^*_\epsilon(\pi)$]
{\epsfig{figure=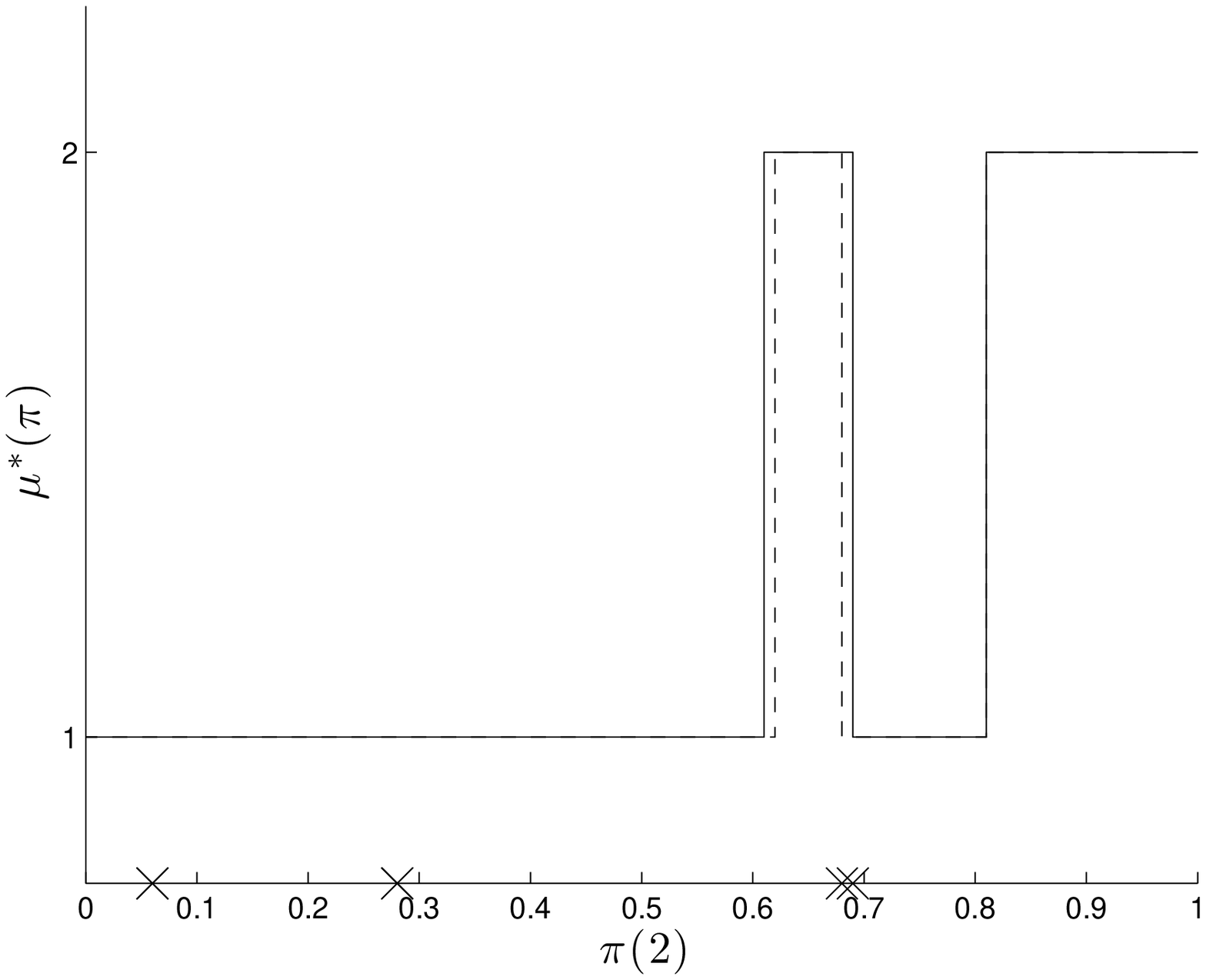,width=0.45\linewidth}} \quad
\subfigure[Value functions  for  global decision policy]
{\epsfig{figure=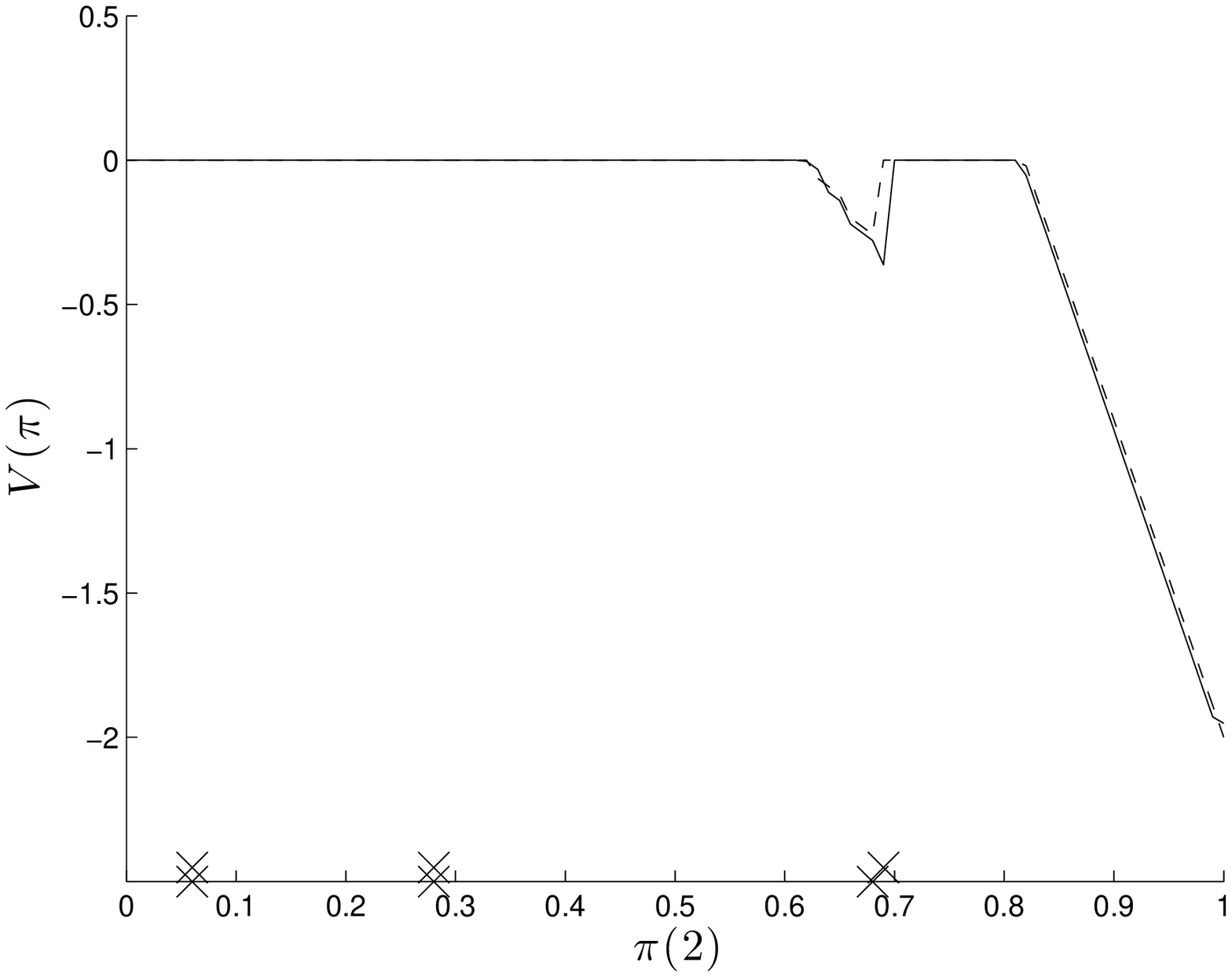,width=0.45\linewidth}}  }
\caption{Optimal decision policy for quickest time change detection  with social learning
for geometric distributed  change time with small probability of change.
In each sub-figure, the graph with solid lines  is for $\epsilon =0.005$ and the graph with
 broken lines is for $\epsilon = 0$.
The policies and optimal costs are very close for $\epsilon = 0.005$ and $\epsilon= 0$. Equation (\ref{eq:bound})
gives a bound for the difference in the optimal costs.
In both cases, the optimal policies are a double threshold and the value functions are  non-concave and discontinuous.}
\label{fig:bound}
\end{figure}

\subsubsection{Numerical Example}
Consider the social learning quickest detection model $ (P, B, c, C, \rho)$ with $\X = \Y = \A = \{1,2\}$,  
\beq B = \begin{bmatrix}  0.85  & 0.15  \\  0.15 & 0.85 \end{bmatrix}, \quad
c = \begin{bmatrix} 
1 & 2 \\ -1 & -3.57
 \end{bmatrix} , \quad \rho =0.8,\quad  d=1.8,\quad f_2 = 2. \label{eq:smalleps}
\eeq

Fig.\ref{fig:bound} shows the optimal policies $\mu_0^*$ (Theorem \ref{thm:ep0})
 and $\mu_\epsilon^*$ (optimal quickest detection policy) together with
optimal costs $V_{\mu_0^*}(\pi)$  and  $V_{\mu_\epsilon^*}(\pi)$ for change probability $\epsilon = 0.005$.
 As can be seen the quickest detection optimal policy and costs 
are very
close to the costs and policies specified by Theorem \ref{thm:ep0}.
 For $\epsilon = 0.002$ the policies
$\mu_0^*$ and $\mu_\epsilon^*$  are almost identical and cannot be distinguished in Fig.\ref{fig:bound}. 
 The policies and optimal costs were obtained by running the value iteration algorithm
for horizon 500 with $\I = [0,1]$ discretized to a grid of 100 points.

\section{Quickest Time Detection for Geometric and PH-distributed Change Time} \label{sec:ph}

The previous section illustrated
 the multi-threshold behavior of social learning based quickest time change detection.
{\em  What sufficient conditions on the social learning model lead to single threshold behavior}?  This section gives
such conditions for PH-distributed change
times $\tau^0$ modelled by a $X\geq 2$-state Markov chain. 
For geometric change times (i.e., $X=2)$ these conditions yield a threshold that  is identical to
 the classical  Kolmogorov--Shiryaev criterion~(\ref{eq:ksd}).

This section comprises of the following results. \\
(i) Sec.\ref{sec:explicit} gives sufficient conditions for
the optimal global decision policy $\mu^*$ to be myopic and characterized by a linear hyperplane threshold.\\
(ii) Sec.\ref{sec:tcurve} gives less restrictive conditions under which the optimal policy is increasing
with respect to the monotone likelihood ratio (MLR)   order and is characterized by a  single threshold
curve. 
Recall that for PH-distributed change time,  the belief space $\I$ is a multi-dimensional simplex.
To order posterior distributions on this simplex, the MLR stochastic order (which is a partial order)
will be used since it is preserved under conditional expectations.
The results involve analysis of the structure of the social learning Bayesian filter
together with lattice programming. All definitions of these orders and consequences are given in the
Appendix.
\\(iii)  Sec.\ref{sec:extensions} describes how sufficient conditions can be given for multiple-threshold policies.\\
(iv)  Finally, Sec.\ref{sec:linear} characterizes the optimal
linear approximation to the MLR increasing policy. It then formulates estimation of the optimal linear approximation
to the threshold curve as a stochastic optimization problem.

\subsubsection*{Assumption (PH)}

Recall fictitious states $2,\ldots,X$ (corresponding to belief states $e_2,\ldots,e_X$)  are used
to model the PH-distribution in (\ref{eq:nu}). It therefore makes sense to constrain the model parameters so that the global decision policy $\mu^*(\pi)$
at the belief states $e_2,\ldots,e_X$ are identical (and similarly for the local decisions taken in social learning).
Throughout this section, when considering PH-distributed change times,  we make the following assumption.
\begin{itemize}
 \item[(PH)] (i) $ C^\p e_i < 0$ for $i=2,\ldots,X$.  (ii) $e_2,\ldots,e_X$ lie in polytope $\mathcal{P}_1$.
\end{itemize}
Assumption (PH)(i) says that  the optimal policy
$\mu^*(\pi)$ treats each  of the fictitious states $2,\ldots,X$ identically -- they all 
 lie outside the stopping set $\Stop$. 
In similar vein, (PH)(ii) requires that individual agents making local decisions
treat the fictitious states $i=2,\ldots,X$ identically, i.e., they lie to the left of each hyperplane $\eta_y$, $y=1,\ldots,Y$.\\
Obviously, (PH)  holds trivially for $X=2$ (geometric case) - otherwise the quickest  change problem would be degenerate.

\subsection{Case 1: Myopic Quickest Detection with Linear Hyperplane Threshold} \label{sec:explicit}
The main result of this subsection is Theorem \ref{thm:cpi} which shows that under suitable conditions, the optimal policy
$\mu^*(\pi)$ has a myopic  structure characterized by $C^\p \pi = 0$.
Recall that $C$ in (\ref{eq:costdef}) denotes  the transformed costs
of the global decision maker with elements $C_j$, $j=1,\ldots,X$.
Denote  the $X-1$ vertices of the intersection of the linear hyperplane  $\{\pi: C ^\p  \pi = 0\}$ with the facets of simplex $\I$ as
$\ver_j$, $j=1,\ldots, X-1$. Then it is straightforwardly seen that these vertices are
\beq  \label{eq:ver}
\ver_j = \frac{C_{j+1} \,e_1 - C_1 \, e_{j+1}} {C_{j+1} - C_1}, \quad j=1,\ldots, X-1. \eeq

Now introduce the following assumption:
\begin{itemize}
\item[(C1)] $C^\p  \Bsi_a^{\ver_j} P^\p \ver_j  \geq  0$
 for all $a \in \A$, $j=1,\ldots,X-1$.
\end{itemize}

The relevance of (C1)  is apparent from the following lemma (proof in Appendix~\ref{app:cpi}).
Define the set of belief states (polytope)
\beq S = \{\pi:  C^\p \pi \geq 0 \} \label{eq:ss} \eeq
\begin{lemma} \label{lem:mappoly}
(C1)  together with (A1), (A2), (A3), (PH) are sufficient for the set $S$ defined in (\ref{eq:ss})  to be closed under the social learning filter~(\ref{eq:piupdate}).
That is   $\pi \in S \implies \T(\pi,a) \in S $ for all $a\in \A$.
\end{lemma}

Recall (A1), (A2), (A3) were introduced in Sec.\ref{sec:assumptions} and (PH) at the beginning of Sec.\ref{sec:ph}.
The main result is as follows. The proof is in Appendix \ref{app:cpi}.

\begin{theorem} \label{thm:cpi}
Consider the  social learning based quickest time detection  model $ (P, B, c, C, \rho)$ in (\ref{eq:model}).
Assume  (A1), (A2), (A3), (S), (C1),  (PH).
Then the global decision maker's optimal  policy is  myopic and is of the form \beq\label{eq:explicit0}
\mu^*(\pi) = \begin{cases} 1 \text{ (stop) } & \text{ if } C^\p \pi \geq 0 \\
					2 \text{ (continue) } & \text{ otherwise } \end{cases}, \quad \Stop = \{\pi: C^\p \pi \geq 0 \}.
\eeq
For the special case $X=2$ (geometric change time), 
\begin{equation} \label{eq:explicitg}
\mu^*(\pi) = \begin{cases} 1 \text{ (stop) } & \text{ if } \pi(2) \leq \pi^*(2)\\
					2 \text{ (continue) } & \text{ if } \pi(2) > \pi^*(2) \end{cases}, \text{ where } \pi^*(2) = \frac{d}{d + f_2 ( 1- \rho P_{22})}
\end{equation}
\end{theorem}

The above result is similar to the entry fee optimal stopping problem having a myopic policy discussed in \cite[pp.389]{HS84}
and \cite[Theorem 2.2, pp.54]{Ros83}. It is important to note, however, that 
even though the optimal policy $\mu^*(\pi)$ in (\ref{eq:explicit0}) is characterized by
 a linear threshold, the value function $V(\pi)$ can still be discontinuous and non-concave (unlike classical
 stopping time problems). This will be illustrated in the numerical example below.

Let us illustrate  what Theorem \ref{thm:cpi} says.  Consider Fig.\ref{fig:theorem}. 
The shaded region in Fig.\ref{fig:theorem} denotes the set $S = \{\pi:C^\p \pi \geq 0\}$.
It is clear from Bellman's equation (\ref{eq:dp_alg})  that the stopping set $\Stop$ is a {\em subset} of this shaded region $S$.
What Theorem \ref{thm:cpi} says is that 
the stopping set  is {\em equal} to the shaded region, i.e., $\Stop = S$, if (C1) and (PH) hold. In terms of Fig.\ref{fig:theorem},
(C1) is sufficient for  $\T(\pi,a)$ to map the belief states $\nu_1$ and $\nu_{2}$ (which are the vertices of the line $C^\p \pi = 0$) to 
polytope $S$.   (PH)(i) implies that states $e_2, e_3$ lie to the left
of the line $C^\p \pi = 0$ (which corresponds to the region $C^\p \pi  < 0$). Similarly, (PH)(ii) means that  $e_2, e_3$ lie to the left
of each line segment $\eta_y$, $y=1,2$, i.e., $e_2,e_3 \in \mathcal{P}_1$.

\begin{figure}\centering
\epsfig{figure=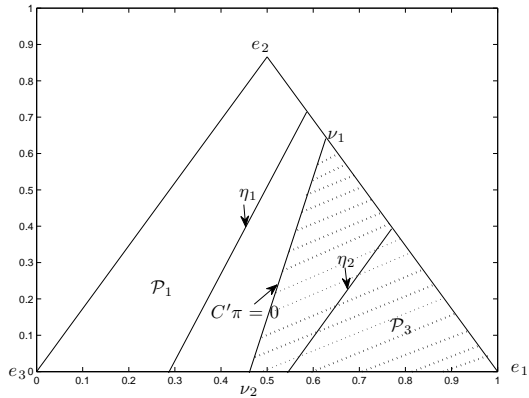,width=0.45\linewidth}
\caption{Illustration of Theorem \ref{thm:cpi}. The shaded region which depicts the polytope  $\{\pi: C^\p \pi \geq 0\}$ is equivalent to the stopping set $\S$ under the assumptions
of the theorem.}
\label{fig:theorem}
\end{figure}

{\em Numerical Example}:
To illustrate Theorem \ref{thm:cpi}, consider the geometric change time  model in (\ref{eq:smalleps}) except that  $P_{22} = 0.75$.
Even though the sufficient condition (C1)  does not hold, 
 the optimal policy is characterized
by a single threshold given by (\ref{eq:explicitg}). This is shown
in Fig.\ref{fig:cpi}. As can be seen in Fig.\ref{fig:cpi},  the value function is non-concave and discontinuous.

\begin{figure}[h]\centering
\mbox{\subfigure[Optimal global decision policy $\mu^*(\pi)$]
{\epsfig{figure=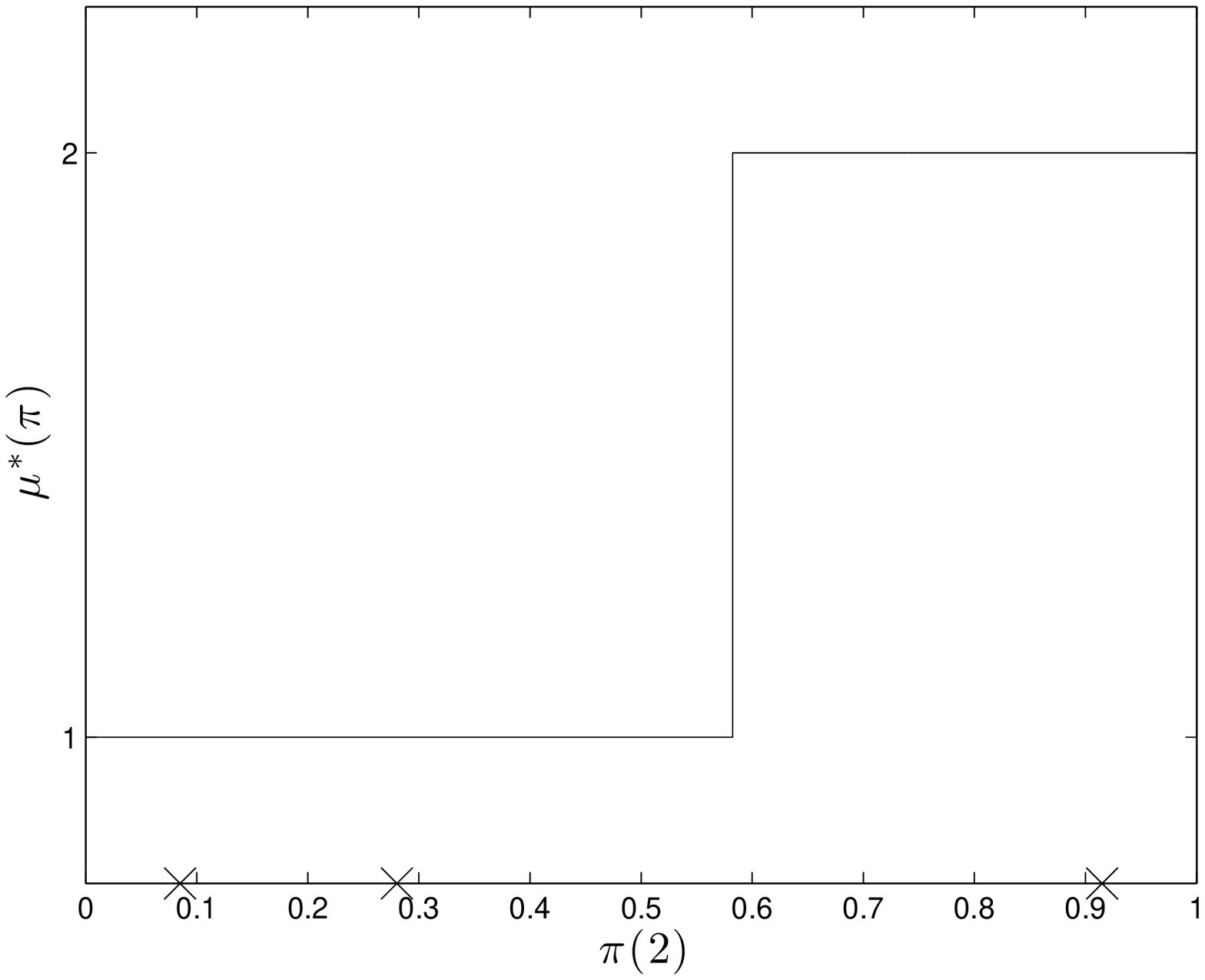,width=0.45\linewidth}} \quad
\subfigure[Value functions  for  global decision policy]
{\epsfig{figure=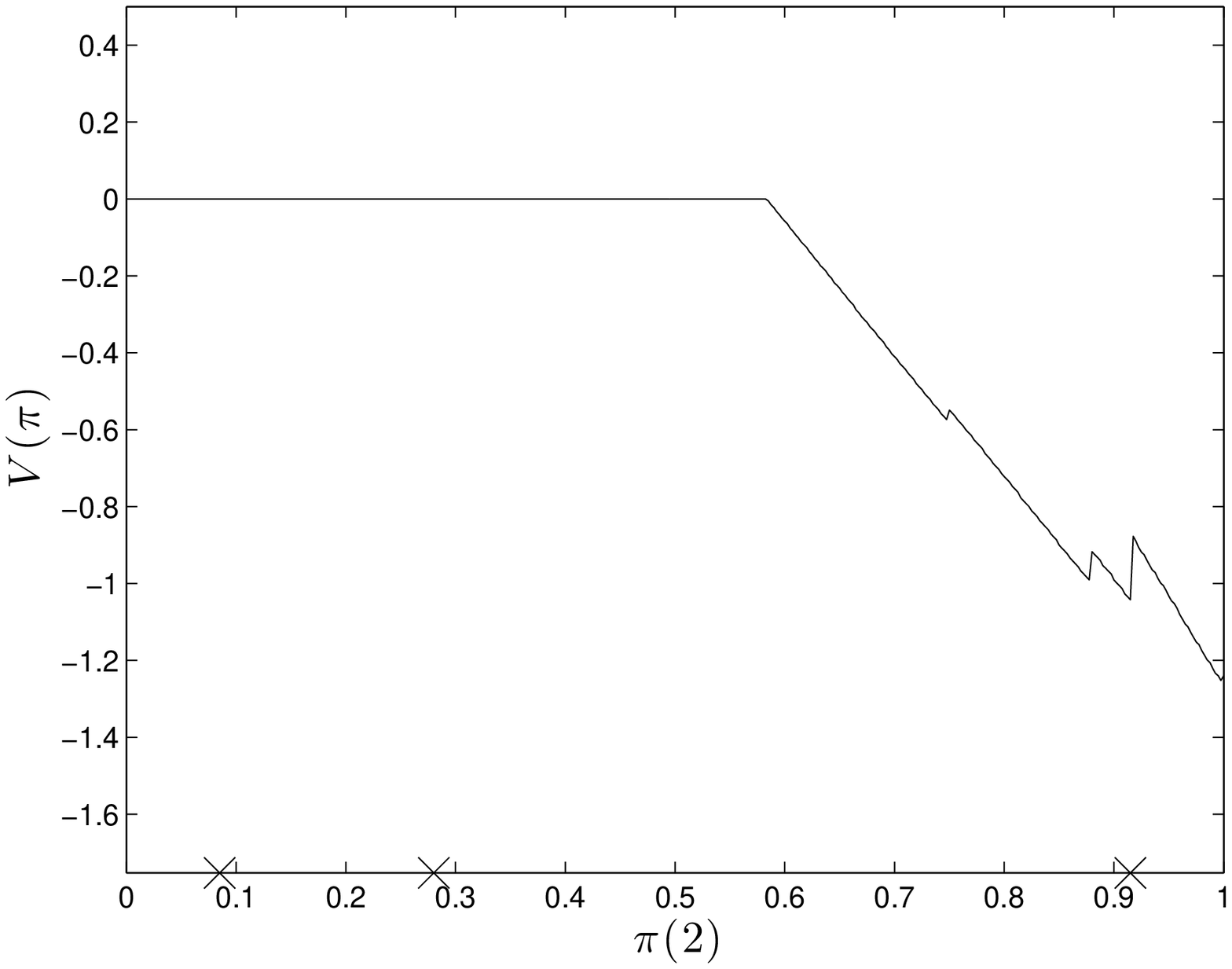,width=0.45\linewidth}}  }
\caption{Numerical example illustrating Theorem \ref{thm:cpi} which
characterizes the optimal decision policy for social learning based quickest detection.
The example is described in Sec.\ref{sec:explicit}.  Even though the value
function is non-concave and discontinuous, the optimal policy has a single threshold specified by (\ref{eq:explicitg}).}
\label{fig:cpi}
\end{figure}

\subsection{Case 2: Existence of a single threshold switching curve} \label{sec:tcurve}
In this subsection, we consider another special case of the social learning based quickest detection model
(\ref{eq:model}). Theorem \ref{thm:1} below shows that  the stopping set  is characterized by a single
threshold curve on the belief space.
The threshold coincides with the classical quickest time detection problem with non-informative observations.
For PH-distributed change times, unlike the previous subsection, the threshold curve  is not necessarily linear. We  give a stochastic gradient algorithm
to estimate this threshold curve in Sec.\ref{sec:linear}.

\subsubsection{Structural Result}
We make the following assumptions. Recall the global decision maker's cost vector $C$ is defined in (\ref{eq:costdef}).
 Let $\bver_j$, $j=1,\ldots,X-1$ denote the $X-1$ vertices of the intersection of hyperplane $\hy_Y$ (defined in (\ref{eq:hy}))
with $\I$. These vertices are computed as (\ref{eq:ver}) with $C$ replaced by $P B_Y (c_1-c_2)$.
\begin{itemize}
\item[(C2)]  
 $(c_1-c_2)^\p B_Y (P^\p)^2  \bver_j  \leq 0$
 for  $j=1,\ldots,X-1$.
\item[(C3)]
The linear hyperplane $\{\pi: C^\p \pi = 0\} $
lies in polytope $\mathcal{P}_{Y+1}$.
\end{itemize}

The following is the main result. The proof is  in Appendix \ref{app:thm1}.  

\begin{theorem} \label{thm:1}  Consider the social learning based quickest detection model $ (P, B, c, C, \rho)$ in (\ref{eq:model}).
Assume (A1), (A2), (S) and (PH) hold. The optimal policy $\mu^*(\pi)$     has the following structure  \\
(i)  Under (C3), $\mu^*(\pi) = 2$ for $\pi \notin \mathcal{P}_1$.
\\
(ii)  Under (C2) and (C3), the stopping set $\Stop $ is as convex subset of polytope $ \mathcal{P}_{Y+1}$. Therefore the boundary of $\Stop$ is differentiable almost everywhere. 
\\
(iii) For geometric-distributed change time ($X=2$),  under (C3),  the optimal policy is identical to that of the Kolmogorov--Shiryaev criterion
(\ref{eq:ksd}) with uniformly distributed observation probabilities.\\
(iv) Under (A3), (C2), (C3). 
on the polytope $\mathcal{P}_{Y+1}$, $\mu^*(\pi)$ has the following structure:
\beq \pi_1,\pi_2 \in \mathcal{P}_{Y+1} \text{   and }
\pi_1 \gr \pi_2  \text{ implies }
\mu^*(\pi_1 ) \geq \mu^*(\pi_2)  \label{eq:structure}\eeq 
(The MLR order $\gr$ is defined in  (\ref{eq:mlrorder}) in Appendix \ref{sec:mlrdef}).
Hence  the  boundary 
of the stopping set $\Stop$ within $\I$ intersects  any line segment $\l(e_1,\bp)$ or $\l(e_X,\bp)$ at most once (see geometric
interpretation below). 
 \qed
\end{theorem}

Even though the policy $\mu^*(\pi)$ in Theorem \ref{thm:1} coincides with that of classical quickest detection, the optimal
cost incurred is always larger as shown in Theorem \ref{thm:blackwell}.

\subsubsection{Discussion of Theorem \ref{thm:1} and assumptions}
Assumption (C3) localizes the decision threshold to polytope $\mathcal{P}_{Y+1}$. As a consequence of  (C3),  $C^\p \pi < 0$ on all polytopes except $\mathcal{P}_{Y+1}$. Therefore on these polytopes, $\mu^*(\pi) = 2$.
Thus statement (i) is  obvious.

Assumption (C2) together with  (A1), (A2), (S) and (PH)  ensures that the polytope  $\mathcal{P}_{Y+1}$ is closed under the belief state mapping $\T(\pi,a)$. That is, 
$\pi \in \mathcal{P}_{Y+1}$ implies  $\T(\pi,a) \in \mathcal{P}_{Y+1}$ for all $a$.  Note that  Assumption (C2) holds trivially for $X=2$ as shown in the footnote.\footnote{For $X=2$, the second element of $P^\p \pi$ is $P_{22} \pi_2$ which is always
smaller than $\pi_2$, So applying $P$ to any belief state keeps it within the interval $\mathcal{P}_{Y+1}$.}
(C2) is similar in spirit to (C1) of the Sec.\ref{sec:explicit}--the key difference is that  (C1) deals with
the global cost vector $C$ whereas (C2) deals with local costs $c_1,c_2$.

Assumptions (C2) and (C3)    allow us to show that the value function $ V(\pi)$ is concave on $\mathcal{P}_{Y+1}$.
 Then Statement (ii),
namely convexity of
the stopping set  $\Stop$,
follows 
from arguments in \cite{Lov87a}. 

Statement (iii) is straightforward to show. The
local decision likelihood probabilities on  $\mathcal{P}_{Y+1}$ are uniform since the local decision yields no information about the state. 
Thus under (C3) 
the threshold is identical
to the classical quickest detection threshold for the  Kolmogorov--Shiryaev criterion
(\ref{eq:ksd}) with uniformly distributed observation probabilities. 

The proof of Statement (iv) is more involved and is given in Appendix \ref{app:thm1}.  
The proof uses structural properties of the Bayesian social filter studied in Theorem~\ref{thm:key} of Appendix \ref{app:thm1},
 along with submodularity,  MLR stochastic order and
a version defined on line segments $\l(e_1,\bp)$ and $\l(e_X,\bp)$ in Appendix \ref{sec:mlrdef}.

\subsubsection{Geometric Interpretation of Statement (iv)}
Since Statement (iv) is non-trivial, let us explain what it says from a geometric point of view.
For PH-distributed change times with $X > 2$,  Statement~(iv) says a lot more than convexity of the stopping region $\Stop$.
 On the unit simplex $\I$
define $\l(e_1,\bp)$ as the line segment constructed from
 $e_1$ to any point $\bp \in \{e_2,\ldots,e_X)$ on the opposite facet of the simplex $\I$.
Similarly
denote $\l(e_X,\bp)$ as any line segment from $e_X$ to  any point $\bp$ on the opposite facet $(e_1,\ldots,e_{X-1})$. 
 Statement (iv) implies that the  boundary 
of the stopping set $\S$ within $\I$ intersects  any such line $\l(e_1,\bp)$ or $\l(e_X,\bp)$ at most once.  Fig.\ref{fig:valid} shows
examples of convex sets that violate this condition.
Also Statement (iv) leads to the 
 following nice geometrical interpretation. If a belief state $\pi \in \Stop$ lies on a line $\l(e_X,\bp)$, then
all belief states on this line closer to $\bp$ also lie in $\Stop$. Similarly if a belief state $\pi \in \l(e_1,\bp)$  lies outside the stopping set $\Stop$,
then all belief states on the line $\l(e_1,\bp)$ further away from $e_1$ also lie outside the stopping set.

Numerical examples are given in Sec.\ref{sec:numerical}.

\begin{figure}[h]\centering
\mbox{\subfigure[Invalid]
{\epsfig{figure=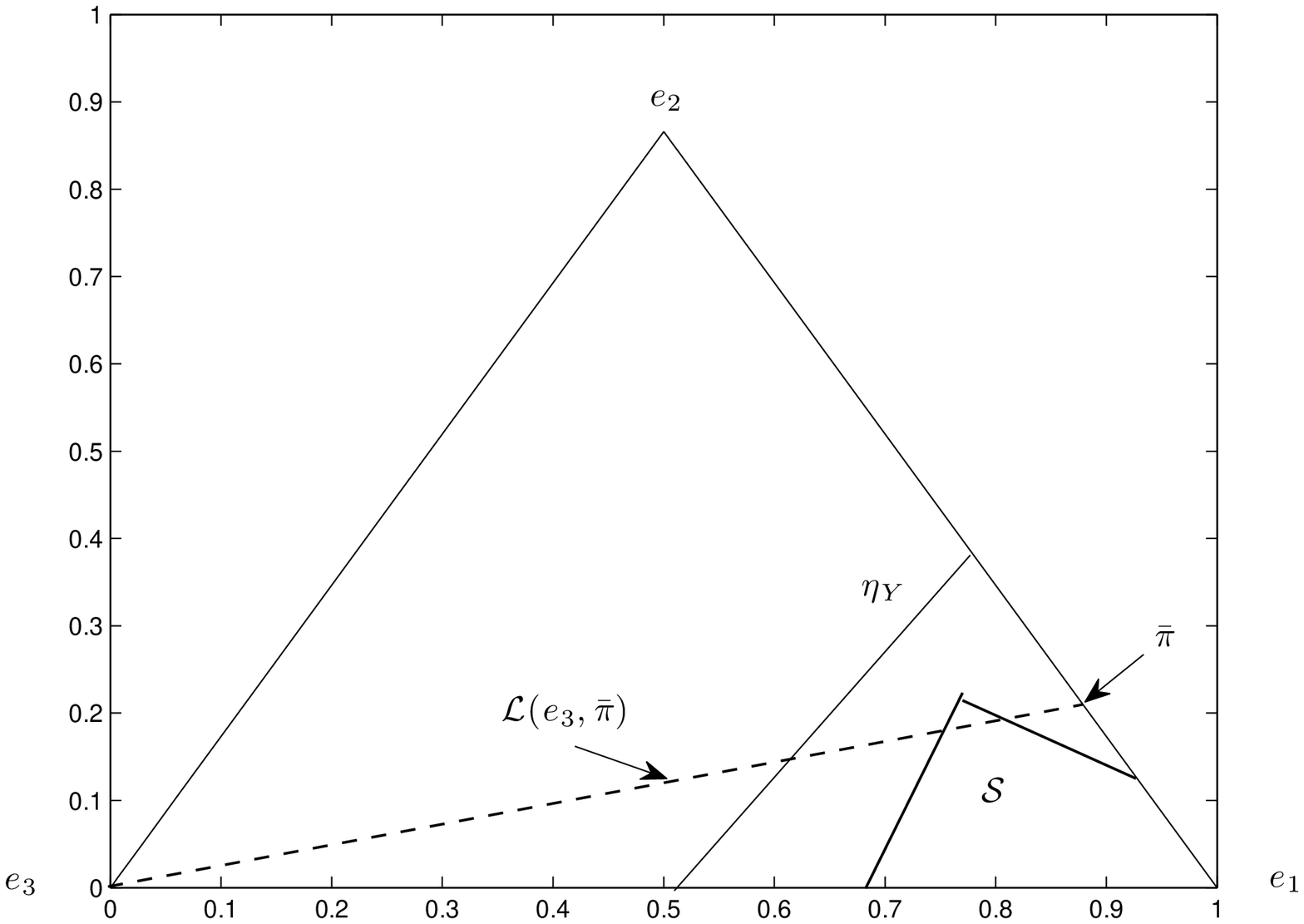,width=0.45\linewidth}} \quad
\subfigure[Valid]
{\epsfig{figure=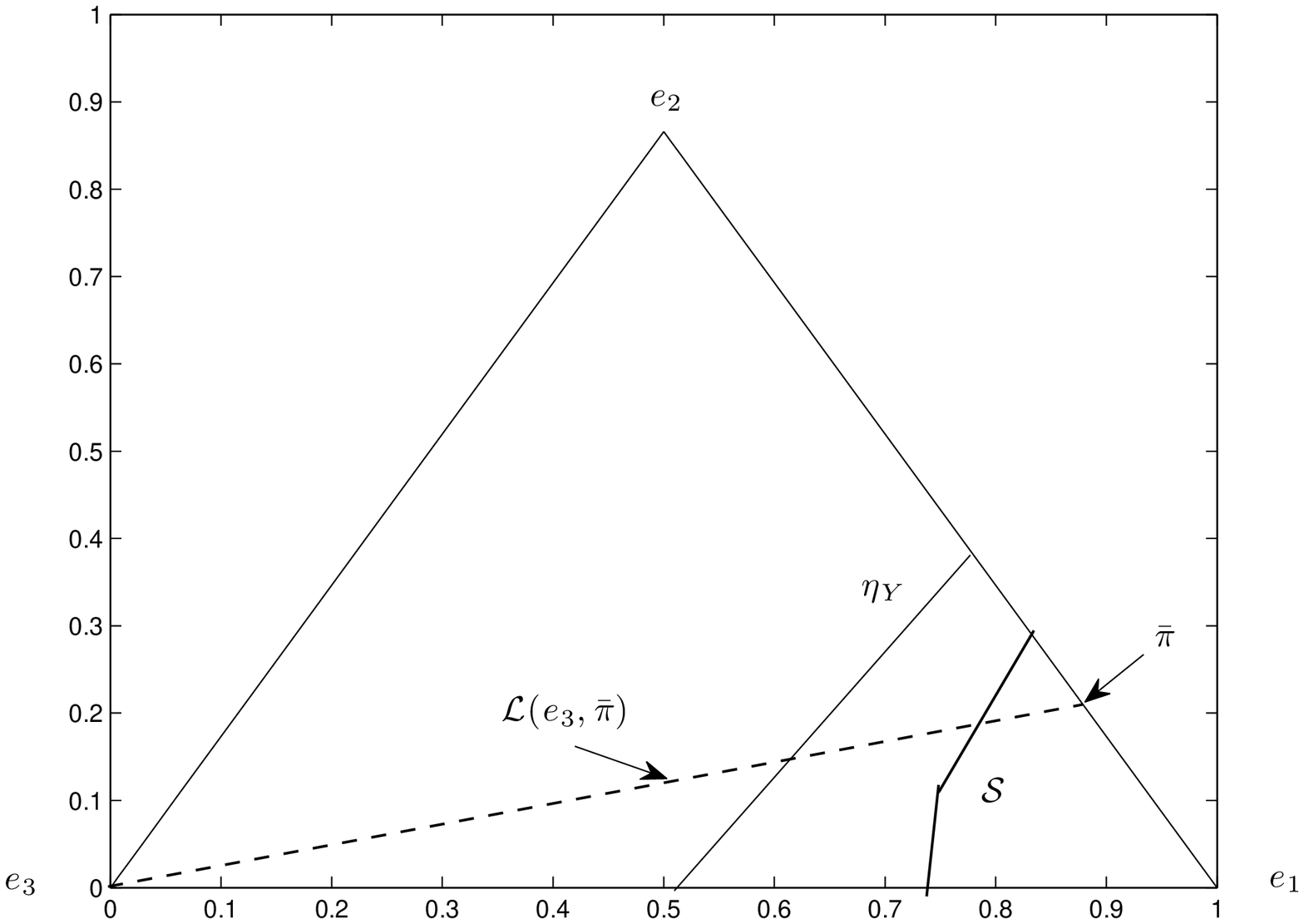,width=0.45\linewidth}}  }
\caption{The stopping set $\Stop$ in Fig (a) violates Statement (iv) of Theorem \ref{thm:1} since the boundary of stopping
set $\Stop$ within $\I$ 
intersects the line $\l(e_3,\bp)$ twice.
Fig(b) shows an example of a stopping set $\Stop$ that satisfies  Statement (iv) of Theorem \ref{thm:1}. In both figures, the region
to the right of line $\hy_Y$ is the polytope $\mathcal{P}_{Y+1}$.}
\label{fig:valid}
\end{figure}

\subsection{Extensions of Theorem \ref{thm:1}  and Multi-threshold Policies} \label{sec:extensions}

\subsubsection{ Local and Global Costs in Global Decision Making} 
Theorem \ref{thm:1} can be extended to consider
a more general global decision maker's cost function (instead of only false alarm and delay) which takes into account the cost of local decisions in social learning. For example,
suppose that the global decision  maker's cost for picking decision $u=2$ (continue) is the delay cost plus an ``operating cost''. That is,
\beq
\Cb(\pi,2) = d e_1^\p \pi + \beta C_{op}(\pi) 
\eeq
Here  $\beta \geq 0$ is a user defined constant and with $\sigma(\cdot)$, $T(\cdot)$, $\H$ defined in (\ref{eq:privu}), (\ref{eq:sigg}), 
\beq
C_{op}(\pi ) 
\ole  \E_{y}\{\min_{a\in \A}\E\{c(x,a) | \H_{k}\} \}
= \sum_{y \in \Y} \min_{a \in \A} \{c_a^\p T(\pi ,y)\} \sigma(\pi,y) = \sum_{y \in \Y} \min_{a \in \A} 
c_a^\p B_y P^\p\pi.
\label{eq:cop}
 \eeq 
 $C_{op}(\pi)$ is  the expected operating cost since it is incurred at each agent $k$ when it makes its local decision via social learning.
Note $C_{op}(\pi)$ 
 is the expected local cost  from choosing decision $u=2$,
 receiving signal $y$, picking recommendation $a$ and broadcasting
the information to the network: the probability of the event is $\sigp(\pi,y)$ and the 
cost is 
 $\min_a \cb^\p T(\pi,y)$. The last equality in (\ref{eq:cop}) follows since
  $\sigma(\pi,y)$ is a non-negative scalar independent of $a$. Actually, the above choice of $C_{op}(\pi)$
  is very similar to that used in constrained social learning in \cite[Chapter 4]{Cha04}.

Then using the same transformation as in (\ref{eq:costdef}), the optimal policy is given by the Bellman's equation (\ref{eq:dp_alg}) with $C(\pi,2) = C^\p \pi + C_{op}(\pi)$.
 Assumption (C3), namely, $\{\pi: C(\pi,2) = 0\} \in \mathcal{P}_{Y+1}$  is then equivalent to 
the linear hyperplane $(C+ P c_1)^\p \pi=0$ lying in polytope $\mathcal{P}_{Y+1}$. This is because on polytope $\mathcal{P}_{Y+1} $ the
optimal local decision $a = 1$, see (\ref{eq:ex31}), and so $C_{op}(\pi) = c_1^\p P^\p \pi$.
Suppose  Assumption (A3) is augmented with the condition that $c(i,a=1) $ is decreasing with $i$.
Then Theorem~\ref{thm:1} continues to hold.

\subsubsection{Multiple Thresholds}
Using a similar proof to Theorem \ref{thm:1}, sufficient conditions can be given for the optimal global policy $\mu^*(\pi)$ in social learning-based quickest detection to have 
multiple thresholds. We describe this below.

Suppose the hyperplane $C^\p \pi = 0$ lies in polytope $\mathcal{P}_{y^*}$ for some $y^* \in \{1,2,\ldots, Y+1\}$. Assume (C2) holds.
Also assume the following generalization of (C2) holds.
\begin{itemize}
\item[(C2')]  The social learning filter maps belief states in  polytope $ \mathcal{P}_y$ to polytope 
$\mathcal{P}_{y+1}$ for  $y = y^*, y^*+1,\ldots,Y$. That is,
  $\pi \in \mathcal{P}_y$ implies $\T(\pi,a) \in \mathcal{P}_{y+1}$.
\end{itemize}
  Then  similar
to the proof of Theorem \ref{thm:1}, the following result can be established (proof omitted).
\begin{theorem}\label{thm:multi}
Under (A1), (A2), (A3), (S), (PH), (C2), (C2'),
 the value function
 $V(\pi)$ is MLR decreasing and therefore optimal policy $\mu^*(\pi)$ is MLR increasing on each polytope $\mathcal{P}_y$, $y = y^*,
\ldots, Y+1$.
\end{theorem} 

 As a result, $\mu^*(\pi)$ is characterized by up to $Y+2 - y^*$ threshold curves, one on each of these polytopes. The reason
 is that even though $V(\pi)$ is decreasing in each polytope, there is no guarantee that is decreasing between polytopes.
Theorem \ref{thm:1} is a special case of the above result when $y^* = Y+1$ and therefore $\mu^*(\pi)$ is characterized by a single threshold
curve.

As an example, consider $X=2, Y=2,A=2$ and suppose $C^\p \pi = 0$ lies in $\mathcal{P}_2$, i.e., $y^* = 2$.  Since
$X=2$ (geometric change time), conditions
(A2), (A3), (PH) and  (C2) hold trivially.
 (C2') holds if the social learning filter $\T(\cdot)$ maps the belief states in $\mathcal{P}_2$ to $\mathcal{P}_3$. A
sufficient condition for this is  
$T^{\hy_1}(\hy_1,2) \in \mathcal{P}_3$, i.e., the transition matrix satisfies
\beq
P_{22} \geq  \frac{B_{12}}{B_{22}B_{11}} \, \frac {(c(2,1)-c(2,2)) B_{21} B_{12} -  (c(1,1)-c(2,1)) B_{11} B_{22} }
{(c(2,1)-c(2,2))  B_{22} - (c(1,1)-c(1,2))  B_{21} }. \label{eq:pmu}\eeq
If (A1) and (\ref{eq:pmu}) hold, then  according to the Theorem \ref{thm:multi},  the optimal policy $\mu^*(\pi)$ is monotone decreasing on each
interval  $\mathcal{P}_2$ and $\mathcal{P}_3$. So $\mu^*(\pi)$ is characterized by up to 2 thresholds, one in each of these intervals.

\subsection{Optimal Linear Decision Threshold and Algorithms} \label{sec:linear}
Theorem   \ref{thm:1} showed that under conditions (A1), (A2),  (A3), (S), (PH), (C2), (C3),  the optimal decision policy $\mu^*(\pi)$ was MLR increasing in
belief state  $\pi \in \mathcal{P}_{Y+1}$.
In this section, we characterize {\em linear} threshold hyperplanes that preserve this MLR structure. Such linear thresholds can then be computed via
a stochastic approximation algorithm. For geometric distributed change time $\tau^0$, since the thresholds are points, estimation is an obvious special
case.

Throughout this section we assume that the conditions of Theorem \ref{thm:1} hold.

\subsubsection{Characterization of MLR increasing linear threshold}
For  $\pi \in \mathcal{P}_{Y+1}$, 
define the $X-1$-dimensional parameter vector $\theta = (\theta(1),\ldots,\theta(X-1))^\p$.
Since $\I \subset \reals^{X-1}$, a linear hyperplane on $\I$ is parametrized by $X-1$ coefficients.
Define the linear threshold policy $\mu_{\theta}(\pi)$ parametrized by the vector $\theta$ as 
\beq \label{eq:linear}
\mu_{\theta}(\pi)  = \begin{cases}
\text{stop } = 1 & \text{ if } \pi(2) + \sum_{i=1}^{X-2} \theta(i) \pi(i+2)   \leq  \theta(X-1) \\
\text{continue} = 2 & \text{ otherwise}  . \end{cases}
 \eeq

  Assume conditions (A1), (A2), (A3),  (S), (PH), (C2), (C3)
hold  for the quickest detection  problem (\ref{eq:csdef}) so that from Theorem \ref{thm:1},
the optimal policy $\mu^*(\pi)$ is 
MLR increasing on lines $\l(e_X,\bp)$ and $\l(e_1,\bp)$. These are defined in Appendix \ref{sec:mlrdef}.
The 
requirement that state 1 lies in the stopping set, means $\mu_\theta(e_1) < 0$ which 
 implies
$\theta(X-1) >0$.

\begin{theorem}
\label{thm:dep}
For belief states $\pi \in \I$,
the   
linear threshold policy $\mu_{\theta}(\pi)$ defined
in (\ref{eq:linear}) is  \\
(i)  MLR increasing
on lines $\l(e_X,\bp)$  iff   $\theta(X-2) \geq 1 $ and $\theta(i) \leq \theta(X-2)$ for $i< X-2$. \\
(ii) MLR increasing
on lines  $\l(e_1,\bp)$  iff $\theta(i)\geq 0$,
for $i<X-2$.\qed
\end{theorem}

The proof of Theorem \ref{thm:dep} is in Appendix \ref{app:dep}. 
The constraints in the above theorem  are necessary {\em and} sufficient for the linear threshold policy
(\ref{eq:linear})
to be MLR increasing on lines $\l(e_X,\bp)$ and $\l(e_1,\bp)$. Under these constraints,
 (\ref{eq:linear}) 
defines the  set of all
MLR increasing linear threshold policies on $\l(e_X,\bp)$ and $\l(e_1,\bp)$ -- it does not leave out any MLR
increasing polices; nor does it  include
any non MLR increasing policies.
 In this sense, optimizing  over the space of MLR increasing linear 
threshold policies yields the optimal linear  approximation to
 threshold curve.

The conditions imposed on the
linear threshold parameters $\theta$ in Theorem \ref{thm:dep} have a nice interpretation  when $X=3$. Recall in this case
$\I$ is an equilateral triangle.
 Let $(\omega(1),\omega(2))$ denote Cartesian coordinates in the equilateral triangle. So  $\pi(2)=2\omega(2)/\sqrt{3}$, $\pi(1) = 
 \omega(1) - \omega(2)/\sqrt{3}$. Then  the linear threshold satisfies
$$
 \omega(2) = \frac{\sqrt{3}\theta(1)}{2-\theta(1)} \omega(1) + \bigl(\theta(2)-\theta(1)\bigr) \frac{\sqrt{3}}{2-\theta(1)}.  $$
So the conditions of Theorem \ref{thm:dep} require that  $\theta(1)\geq 1$, i.e., the 
  threshold has slope of $60^o$ or larger. When $\theta(1)>2$, slope 
becomes negative, i.e., more than $90^o$.

Fig.\ref{fig:validn} shows examples of a valid and invalid linear threshold. Fig.\ref{fig:validn}(a) illustrates a valid MLR increasing linear threshold policy. Fig.\ref{fig:validn}(b) is invalid since the threshold is less
than $60^o$ meaning that the resulting policy is not MLR increasing on lines. Also shown is the hyperplane
 $C^\p \pi = 0$ which by Assumption (C3) lies in polytope $\mathcal{P}_{Y+1}$.

\begin{figure}[h]\centering
\mbox{\subfigure[Valid]
{\epsfig{figure=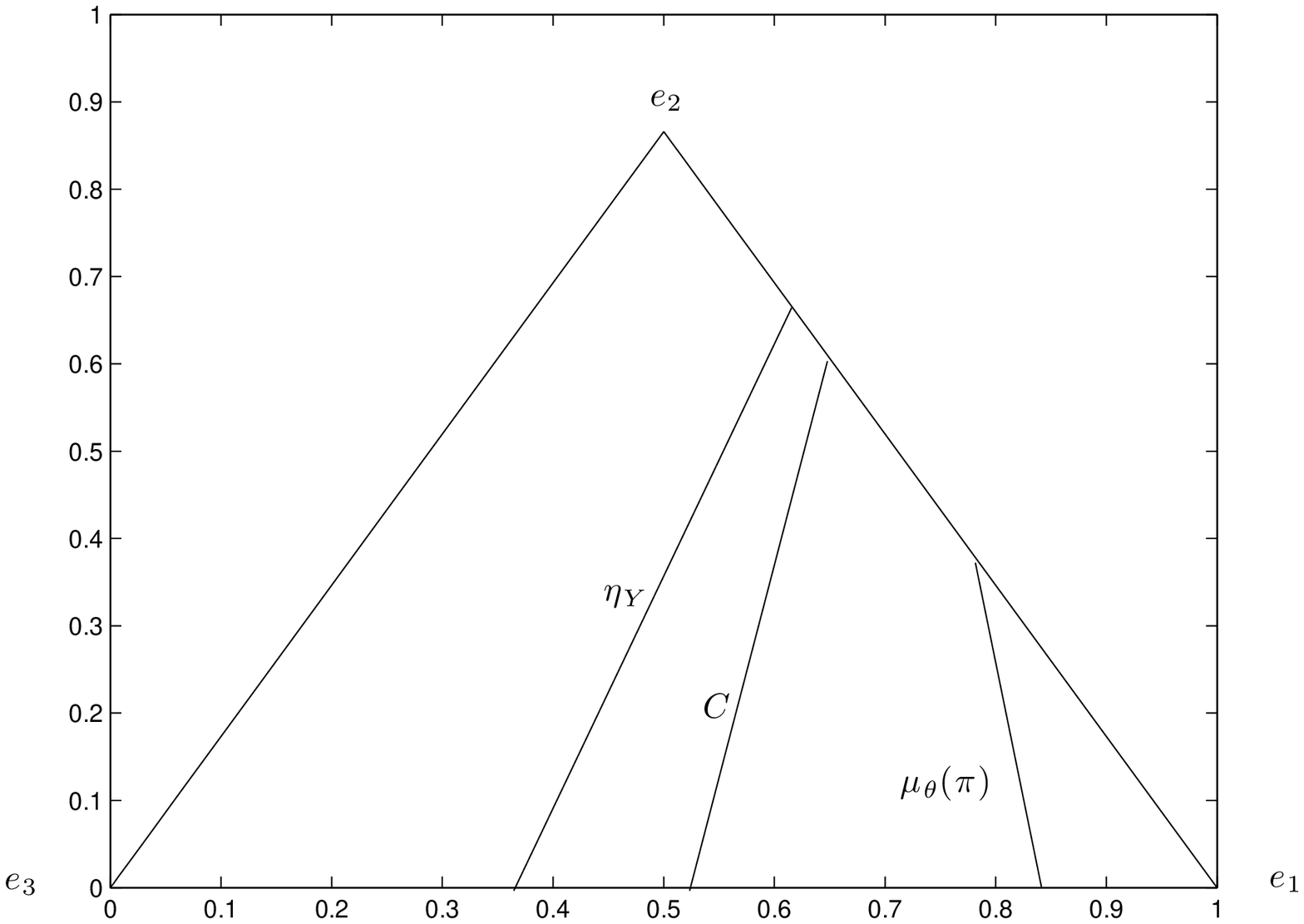,width=0.45\linewidth}} \quad
\subfigure[invalid]
{\epsfig{figure=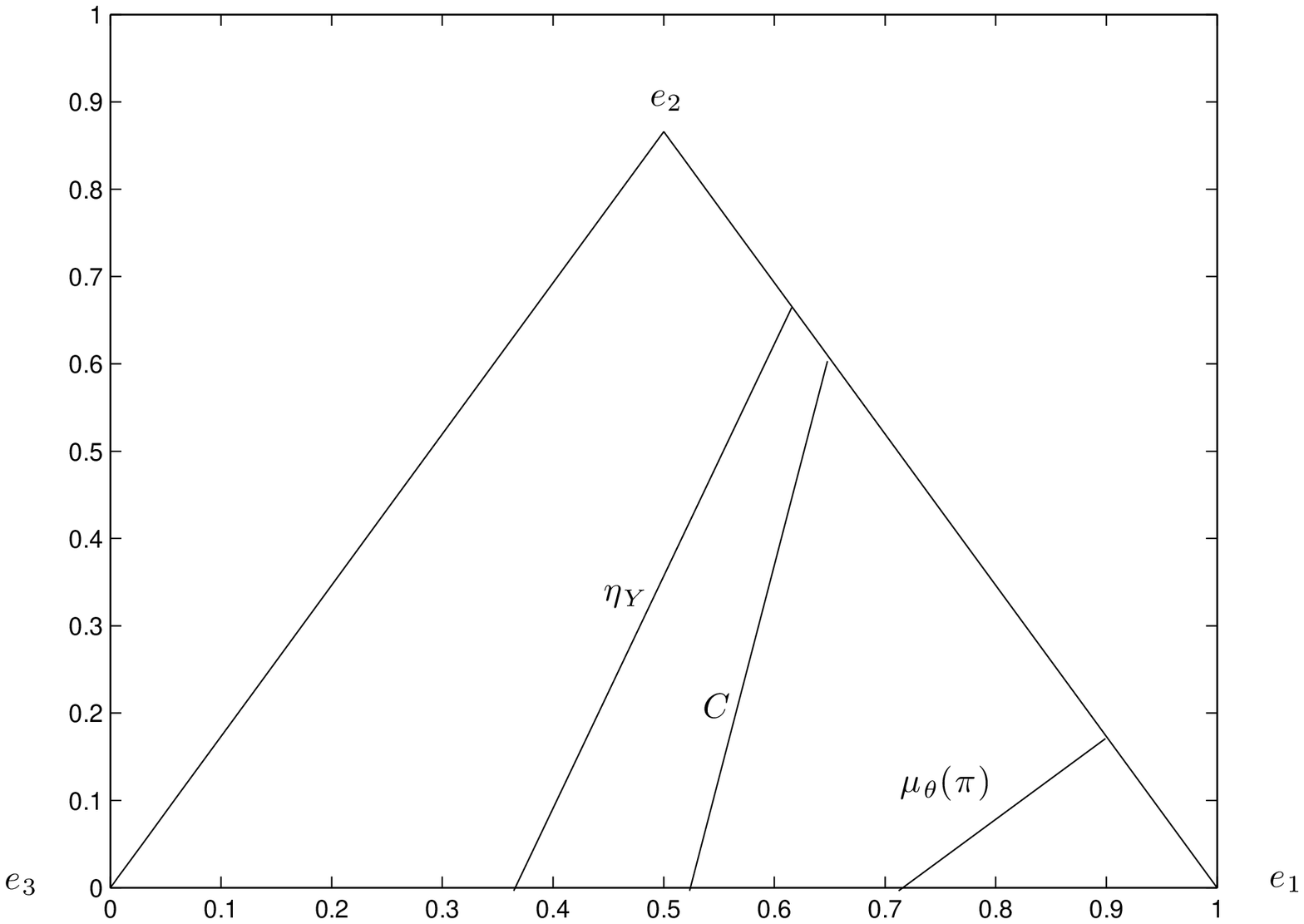,width=0.45\linewidth}}  }
\caption{Fig(a) illustrates a valid MLR increasing linear threshold policy. The linear threshold policy in Fig (b) violates the requirement that the policy $\mu_\theta(\pi)$ is MLR increasing since it has a slope less than $60^o$.
 In both figures, the region
to the right of line $\hy_Y$ is the polytope $\mathcal{P}_{Y+1}$. In the figures,
$C$ denotes  the hyperplane $C^\p \pi= 0$ which lies
in $\mathcal{P}_{Y+1}$ by Assumption (C3).}
\label{fig:validn}
\end{figure}

\subsubsection{Computation of Optimal Linear Threshold}
As a consequence of Theorem \ref{thm:dep}, the optimal
 linear threshold
approximation to  threshold curve $\Gamma$ of Theorem \ref{thm:1} is 
the solution of the following constrained optimization problem:
\beq
\theta^* = \arg \min_{{\theta} \in \reals^{X-1}} J_{\mu_{\theta}}(\pizero), \; \text{ subject to  $0 \leq \theta(i)\leq \theta(X-2)$, $\theta(X-2) \geq 1$ and $\theta(X-1)>0$
}
\label{eq:tc}\eeq
 where the cost $J_{\mu_{\theta}}(\pizero)$ is obtained as in  (\ref{eq:csdef})
by applying threshold policy
$\mu_{\theta}$ in (\ref{eq:linear}). 

  Because the cost $J_{\mu_\theta}(\pizero)$  in (\ref{eq:tc})  cannot be computed in closed form, we resort to simulation
 based stochastic optimization.  Let 
$n=1,2\ldots,$ denote iterations of the algorithm.  The aim is to  solve the following linearly constrained stochastic optimization problem:
\beq \text{  Compute } 
\theta^* = \arg\min_{\theta \in \Theta} \E\{ {J}_n(\mu_{\theta})  \} 
\;
\text{ subject to  $0 \leq \theta(i)\leq \theta(X-2)$, $\theta(X-2) \geq 1$ and $\theta(X-1)>0$.
}
 \label{eq:stochopt2}\eeq
Here, for each initial condition $\pi_0$, the 
sample path cost   $ {J}_n(\mu_{\theta},\pi_0) $ is evaluated as
\begin{align}
J_n(\mu_\theta,\pi_0)&=  \sum_{k=1}^\infty \discount^{k-1} C(\pi_k,u_k)  \quad
\text{ where } u_k = \mu_\theta(\pi_k) \text{ is computed via (\ref{eq:linear}) } \label{eq:jnmu}\\
J_n(\mu_\theta) &= \frac{1}{L} \sum_{l=1}^L J_n(\mu_\theta,\pi_0^{(l)}) 
\text { where prior $\pi_0^{(l)}$ is sampled uniformly from simplex $\I$.} \nonumber
\end{align}
A convenient way of sampling uniformly from $\I$ is to use 
the Dirichlet distribution (i.e., $\pi_0(i) = x_i/\sum_i x_i$, where $x_i \sim $ unit exponential distribution).

The above stochastic optimization problem is solved  by stochastic approximation
algorithms such as the
Simultaneous Perturbation Stochastic Approximation (SPSA) algorithm
\cite{Spa03}
which
converges to a local minimum; see \cite{Kri11} for a novel parametrization that deals with
the hypersphere constraints.
The stochastic gradient algorithm
converges to  local optima, so it is necessary
to try several initial conditions. The computational cost  at each
iteration is linear in the dimension of $\theta$ and is independent of 
 the observation alphabet size $Y$. Convergence  (w.p.1) can be established using
 techniques in~ \cite{KY02,KY03}.
More sophisticated methods than SPSA can also be used. For example,
\cite{BB02} uses the score function method to perform gradient-based reinforcement learning.
These algorithms are  applicable  to solve the constrained stochastic optimization
problem (\ref{eq:stochopt2}). Also, if the change time distribution (specified by $P$)
and the observation likelihoods (specified by  $B$) are not completely specified,
as long as the assumptions   Theorem \ref{thm:1} hold, then the  reinforcement learning
algorithms \cite{BB02} can be used to  solve (\ref{eq:stochopt2}).

\section{Multi-agent Quickest Time  Detection with Adaptive Sensing} \label{sec:sensor}
As mentioned in Sec.\ref{sec:intro}, the social learning protocol is very similar to multi-agent quickest time detection with a sensor manager (controller).
Motivated by sensor network applications,  this section describes the formulation and the main results.
The information patterns are similar to social learning and so the results developed in previous sections apply. The observations now can also belong to a continuum.

 Consider a countable number of agents indexed by $k=1,2,\ldots$.
Each agent acts once  in a predetermined sequential order indexed by $k=1,2,\ldots$ as follows: Based on the current
belief state $\pi_{k-1}$, agent $k$ acts as follows:
\begin{itemize} \item  Agent $k$ first chooses decision $u_k \in \{1 \text{ (stop) }, 2 \text{ (continue)}\}$.
If the agent decides to stop, then as in earlier sections, a  false alarm penalty is paid, and the problem terminates.
\item If agent $k$ chooses $u_k =2 $,   then it chooses its operating {\em mode}
$a_k \in \{1,2\}$ according to a built-in micro-manager. Agent $k$ then views the world according to this mode  -- that is, 
it  obtains  observation $y_k$ 
from a distribution that depends on mode $a_k$. 
It then communicates its belief state $\pi_k$ to the next agent.
\end{itemize}
{\em Remark}: An equivalent formulation is as follows: A single smart sensor  adapts its operating mode $a_k$ at each time $k$ based on the posterior
distribution of the underlying state at the previous time instant. How can quickest detection be achieved with this sensor?

  How can such a network of agents, where each agent makes autonomous micro-management decisions on its mode,  achieve quickest time detection? 
   The quickest time detector can be viewed as a macro-manager that operates on the belief states and micro-manager decisions. 
   Clearly the micro and macro-managers interact -- the  local decisions $a_k$ taken by the micro-manager determines $y_k$ which determines
   $\pi_k$ and hence determines decision $u_{k+1}$ of the quickest time macro-manager. 
  
\subsection{Micro-manager for Agent  Mode Selection}

\subsubsection{Costs and mode selection} As in  (\ref{eq:ca}), let $c_a$  denote the local cost of deploying sensor mode $a \in \A = \{1,2\}$. 
 To avoid trivial solutions,  as  in Sec.\ref{sec:assumptions}, we make the  submodular assumption (S).

Similar to the social learning formulation, the micro-manager picks local decision $a_k$ myopically as follows: Based on the belief state $\pi_{k-1}$ of the previous agent,
each  agent $k$  picks its  mode $a_k   \in \A=
\{ 1,2 \}$  of which sensor to deploy by minimizing its expected  predicted cost:  
\beq a_k = \arg\min_{a\in \{1,2\}} \E\{ c(x_k,a_k)|\F_{k-1}\} = \arg\min_{a\in \{1,2\}} c_a^\p P^\p \pi_{k-1} 
\label{eq:micro}
\eeq
where $\F_k$ denotes the filtration
$\sigma(y_l, l\leq k)$.
Define the convex polytopes $\mathcal{P}_1$
and $\mathcal{P}_2$ that partition $\I$ as
\beq \label{eq:polytopesched}
  \mathcal{P}_1 = \{\pi:  (c_1-c_2)^\p P^\p \pi \geq  0\}, \quad
   \mathcal{P}_2 = \{\pi:  (c_1-c_2)^\p P^\p \pi < 0\}
  \eeq
  Then  from (\ref{eq:micro}) it follows that for $\pi \in \mathcal{P}_1$, $a_k = 2$ and for $\pi \in \mathcal{P}_2$, $a_k = 1$.

\subsubsection{Mode dependent observations}
The agent then makes an 
observation $y_k$ 
depending on its choice of mode  $a_k$. 
Based on its mode $a_k$ in (\ref{eq:micro}), agent $k$ then obtains an observation from  conditional probability distribution  
\beq P(y_k \leq \bar{y} | x_k = e_x, a_k = a) = \sum_{y \leq \bar{y}} B^{(a)}_{xy}  , \quad x \in \X, a \in \{1,2\} . \label{eq:obsa}\eeq
Here $\sum_y$ denotes integration with respect to the Lebesgue measure
(in which case $\Y\subset \reals$ and $B_{xy}$ is the conditional probability density function)
or counting measure (in which case $\Y$ is a subset of the integers and $B_{xy}$ is the conditional probability
mass function $B_{xy}= P(y_k=y|x_k=x)$).
The key point  is that unlike classical quickest detection, each agent  now views the world based on its selected mode $a_k$.

Let $T^{(a)}(\pi,y)$ denote the belief state update if mode $a$ is chosen and measurement $y$ obtained. It is given by the HMM filter (\ref{eq:privu}) with mode dependent probabilities
$B^{(a)}_y = \text{diag}(P(y|x,a),x\in \X)$. That is,
\beq T^{(a)}(\pi,y) = B^{(a)}_y P^\p \pi/\sigma(\pi,y), \quad \sigma(\pi,y) = \ones^\p B_y^{(a)} P^\p \pi. \label{eq:pia}\eeq

\begin{figure}
\centering \epsfig{figure=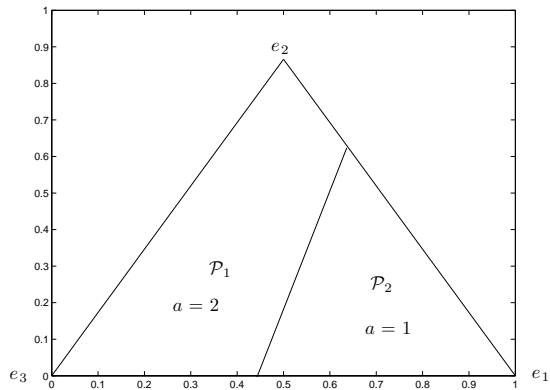,width=0.45\linewidth}
\caption{The figure illustrates the setup in Sec.\ref{sec:sensor}.
 The mode dependent  observation probabilities $B^{(a)}$, $a \in \{1,2\}$ are chosen depending on
the belief state $\pi$ in polytope  $\mathcal{P}_2$ or  $\mathcal{P}_1$ defined in (\ref{eq:polytopesched}). The aim is to
perform quickest detection given this mode dependent observation probability constraint.}
\end{figure}

\subsection{Macro-Manager for Quickest Time Detection}

Below we present the  assumptions and main result.
Based on the above micro-manager protocol, the aim is to perform quickest time change detection. So the quickest detection problem
can be viewed as optimizing the cost function (\ref{eq:costdef}) subject to the constraint that the belief state evolves according
to (\ref{eq:pia}).
The setup is identical to that in Sec.\ref{sec:qdform} and Sec.\ref{sec:dp}. For $k< \tau$, agents choose $u=2 \text{ (continue) }$ and
at $k = \tau$, agent $k$ picks $u_k = 1 \text{ ( declares a change and stop) } $.  The optimal policy
$\mu^*(\pi)$ of the macro-manager satisfies Bellman's equation (\ref{eq:dp_alg}).

The following theorems mimic the results for the social learning based quickest detection
problem,  and their proofs are identical.

\subsubsection{Blackwell Dominance} Suppose the mode dependent observation matrices are of the form $B^{(a)} =  B Q^{(a)}$ where $B$ and $Q^{(a)}$, $a = 1,2$ are stochastic kernels. Then an identical proof to Theorem \ref{thm:blackwell} shows that
 classical quickest detection with observation matrix $B$ always yields a lower cost than mode dependent quickest detection
with observation matrices $B^{(a)}$, where the mode $a$ is chosen according to any arbitrary strategy.

\subsubsection{Threshold Policies} 
Consider the following assumptions that are similar to (C1)  in Sec.\ref{sec:explicit} and  (C2) in Sec.\ref{sec:tcurve}.
Recall vertices $\ver_j$ are defined in (\ref{eq:ver}) and $\bver_j$ denote vertices of  hyperplane $(c_1-c_2)^\p P^\p \pi= 0$.
\begin{itemize}
\item[(C1)] If $\{ \pi: C^\p \pi = 0\}$ lies in one of the polytopes $\mathcal{P}_a$, then $C^\p B_y^{(a)} P^\p \ver_j \geq 0$,  $ j=1,\ldots, X-1$ for all $y \in \Y$.
\item[(C2)] $(c_1- c_2) ^\p  P^\p B_y^{(1)} P^\p \bver_j \leq  0$, $j=1,\ldots,X-1$, $y\in \Y$.
\end{itemize}

We have the following result regarding the structure of  $\mu^*(\pi)$ for quickest time
detection.

\begin{theorem}\label{thm:manga}
Theorems \ref{thm:cpi} and \ref{thm:1} hold for the optimal quickest time decision policy $\mu^*(\pi)$ of the macro-manager.
Also Theorem \ref{thm:dep} holds for MLR policies and computation of the optimal linear threshold can be formulated 
as the stochastic optimization problem (\ref{eq:jnmu}).
\end{theorem}

 (C1) and (C2) are relatively easy to check even if $y \in \Y$ is continuum as shown below.
For all $x$, let $y_{\text{max}}$ denote the maximum support of the distribution $B^{(1)}_{xy}$, i.e., 
$y_{max} = \sup \{y: B^{(1)}_{xy} > 0 \}$.

 \begin{lemma} (C1), (C2) hold if their inequalities hold  for $y=y_\text{max}$.
 \end{lemma}
Thus only a finite number of inequalities  need to be verified. 
In particular for a Gaussian distribution, since $y_\text{max} = \infty$, the
filter $T^{(1)}(\pi,\infty)$ becomes
the Bayesian predictor $P^\p \pi$. So it suffices to check
that $C^\p P^\p \ver_j \geq 0 $ for (C1) to hold.

\begin{proof} Consider (C1).
$ C^\p B^{(a)}_y P^\p \ver_j \geq 0 $ is equivalent to verifying 
$C^\p B^{(a)}_y P^\p \ver_j /\sigma(\ver_j,y) \geq 0$ since $\sigma(\pi,y)$ is non-negative for all $\pi \in \I$. So we need to
check that 
$C^\p T^{(a)}(\ver_j,y) \geq 0$ for all $y \in \Y$.
But since $P$ and $B$ are TP2 according to Assumptions (A1), (A2), from Theorem \ref{thm:key}(4) in Appendix \ref{app:thm1}, the belief state update $T^{(a)}(\pi,y)$ is MLR increasing in $y$.
Moreover by (A3) $C$ has  decreasing elements. Therefore from Result \ref{res1} in Appendix \ref{sec:mlrdef}, 
$C^\p T^{(a)}(\pi,y) $ is decreasing in $y$.
So it suffices
to check that 
$C^\p T^{(a)}(\ver_j,y_\text{max}) \geq 0$.
\end{proof}

\section{Numerical Results} \label{sec:numerical}

In addition to the numerical examples presented earlier, this section presents two numerical
examples. The first example illustrates the multiple threshold policies inherent in social learning (this example
was mentioned in Sec.\ref{sec:intro}). The second example illustrates the optimal threshold curve for
a PH-type distributed change time that was proved in Theorem \ref{thm:1}.

{\bf Example 1.  Geometric Distributed Change Time}: 
This examples illustrates the existence
of a triple threshold policy for quickest time change detection when the change time $\tau^0$ is geometrically distributed.
We chose the social learning model with parameters
$\X = \{1,2\}$ (so $\I=[0,1]$ is a one dimensional simplex), $\Y=\{1,2,3\}$, $\A=\{1 ,2 \}$, 
$$ B = \begin{bmatrix}  0.9  & 0.1  \\  0.1 & 0.9 \end{bmatrix},
\quad E\{\tau^0\} = 20 \implies P =  \begin{bmatrix} 1 &  0 \\  0.05 & 0.95 \end{bmatrix}, \;
c = \begin{bmatrix} 
1 & 2 \\ -1 & -3.57
 \end{bmatrix} .
$$
For the global quickest time detection parameters
we chose $\rho =0.99$,  delay $d=1.25$, false alarm vector $\f=3 e_2$ (i.e., $f_2 = 3$).
It is easily checked that (A1), (A2) and (S)  hold.

The optimal policy $\mu^*(\pi)$ is shown in Fig.\ref{fig:redgreen}(a) and comprises of a triple threshold
policy.
It was computed by constructing
a uniform grid of 500 points for $\pi(2)\in [0,1]$ and then implementing the value iteration algorithm (\ref{eq:vi}) for a horizon of $N=200$. The `x' in Fig.\ref{fig:redgreen}(a) and (b) are the values
of $\hy_2(2)$, $\q(2)$ and $\hy_1(2)$, respectively.

{\bf Example 2. Phase Distributed Change Time}: 
This examples illustrates Theorem \ref{thm:1} which proved the existence
of a single threshold curve  for social learning based quickest time change detection 
with PH-distributed change time.
We model the PH-distribution via a 3-state Markov chain.  So the belief space $\I$ is a two dimensional simplex (equilateral
triangle) and can
be visualized easily.

We chose the social learning model with parameters
 $\X = \{1,2,3\}$ $\Y=\{1,2,3,4,5\}$, $\A=\{1,2\}$. The
  observation probabilities and  local decision costs were chosen as 
 $$
\begin{matrix}  B_{1,y} \propto \exp(-(y-1)^2/6) \\  B_{2,y} = B_{3,y}\propto \exp(-(y-5)^2/6)
\end{matrix}  ,\quad 
c = (c(i,a)) = \begin{bmatrix} 4 & 50 \\ 2 & 0 \\ 2 & 0 \end{bmatrix} .$$
 
The global costs for quickest detection in (\ref{eq:cp1}) and (\ref{eq:cp2}) were chosen as $d = 1.5$, $\f=[0 \; 20\;  25]^\p$ and
discount factor $\rho = 0.9$.

The PH-distributed change times were modelled by the 3 state Markov chain with transition probability $P$.
 To illustrate the quickest time detection, we chose
4 candidate transition probability matrices, namely,
$$ P^{(1)} =  \begin{bmatrix} 1 & 0 & 0 \\  0.1 & 0.9 & 0 \\  0.1&  0.9&  0 \end{bmatrix},
P^{(2)} = \begin{bmatrix} 1 & 0 & 0 \\  0.1 & 0.5 & 0.4 \\  0&  0.1&  0.9 \end{bmatrix},
P^{(3)} = \begin{bmatrix} 1 & 0 & 0\\  0.1 & 0.7 &  0.2\\  0 & 0.4& 0.6 \end{bmatrix},
P^{(4)} =\begin{bmatrix} 1 &  0  & 0\\  0.1 & 0.45 & 0.45 \\ 0.05 & 0.40 & 0.55\end{bmatrix}.
$$
Note $P^{(1)}$ models the geometric distribution since states 2 and 3 are indistinguishable -- in fact
it is exactly lumpable \cite{KS60} into the 2 state Markov chain with transition matrix $\begin{bmatrix} 1 & 0 \\
0.1 & 0.9 \end{bmatrix}$. 

Fig.\ref{fig:phaseplot} plots the probability mass function $\nu_k$ (see (\ref{eq:nu})) of the  PH-distributed change time 
$\tau^0$ for these four
transition matrices for $\bar{\pi}_0 = [0.03,\;0.97]^\p$. Fig.\ref{fig:phaseplot} shows these
PH-distributions are quite different in behavior to a geometric distribution -- they are non-monotone and 
have heavier tails.

\begin{figure} \centering
{\epsfig{figure=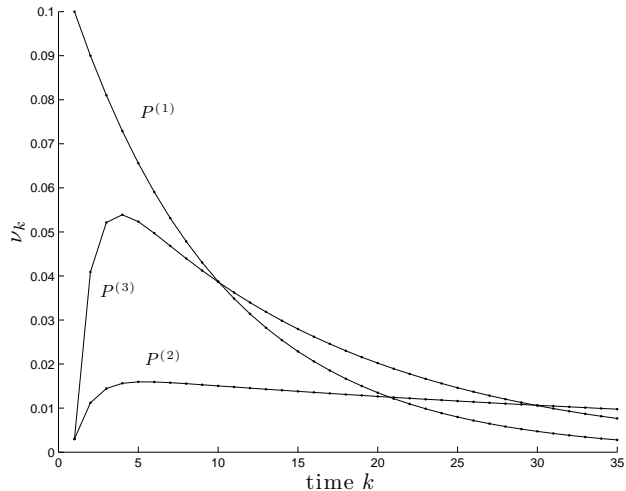,width=0.5\linewidth}} 
\caption{Plots of change time $\tau^0$ probability mass function $\nu_k$ in (\ref{eq:nu}) for $P^{(1)}$ (geometric distribution)
and  $P^{(2)}$, $P^{(3)}$  (phase-type distributions).
 } \label{fig:phaseplot}
\end{figure}

It is easily checked that (A1), (A2), (A3), (S), (PH), (C2) and (C3) of Theorem \ref{thm:1} hold.
Fig.\ref{fig:3stateplots} shows the optimal decision policies for these four cases with the stopping
set $\mathcal{S}$ shaded.
 The optimal policy was computed as follows.
A $50 \times 50$ grid of  $(\pi(1),\pi(2) )$ values was formed within the 2-dimensional unit simplex $\I$.Then the
  value iteration algorithm (\ref{eq:vi}) was solved for horizon $N=200$ (in all cases $\|V_{200}(\pi) - V_{199}(\pi)\|_\infty< 10^{-15}$ implying
  that the value iteration algorithm converged).
 In all 4 cases,  the optimal decision policy is characterized
by a single threshold curve in polytope $\mathcal{P}_6$. This is consistent with Theorem \ref{thm:1}.

 In each plot of Fig.\ref{fig:3stateplots} also shows the
hyperplanes $C^\p \pi = 0$ (defined in (\ref{eq:costdef})) and $\hy_5$ (defined in (\ref{eq:hy}).  The polytope $\mathcal{P}_6$ is to the right of hyperplane $\hy_5$.
The remaining  line segments from left to right are $\hy_1,\ldots,\hy_4$.
Note that  hyperplane $ C^\p \pi = 0$ 
lies in $\mathcal{P}_6$, thereby satisfying Assumption (PH) and (C3).

Actually cases $P^{(2)}$ and $P^{(3)}$ satisfy Assumptions (C1), and (PH) and so Theorem \ref{thm:cpi} holds.
Therefore, for these two cases, the optimal threshold curve is the linear hyperplane $C^\p \pi = 0$ as can be seen in Fig.\ref{fig:3stateplots}.

\begin{figure} \centering
\mbox{\subfigure[$P^{(1)}$]
{\epsfig{figure=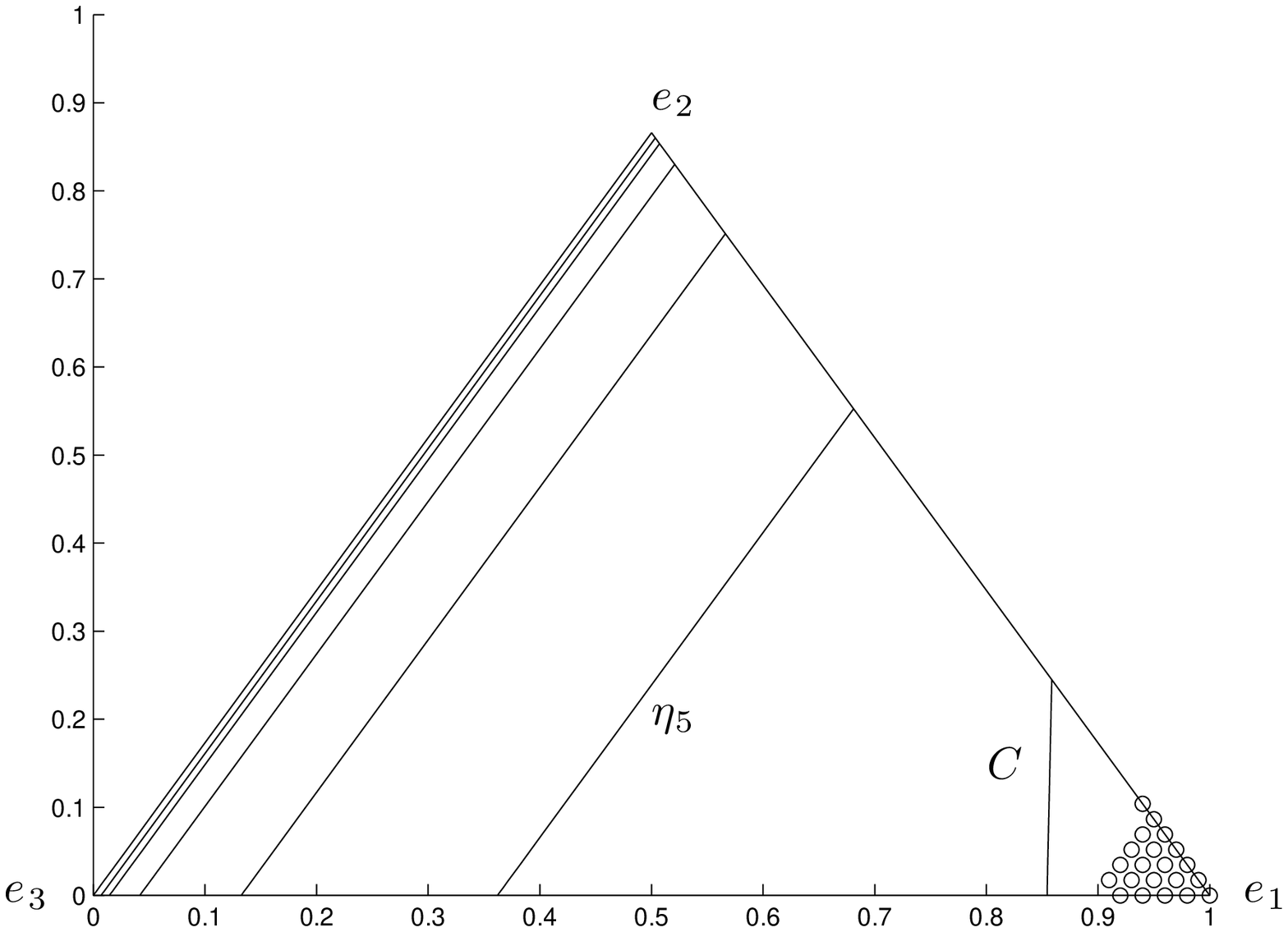,width=0.45\linewidth}}  \quad
\subfigure[$P^{(2)}$]{\epsfig{figure=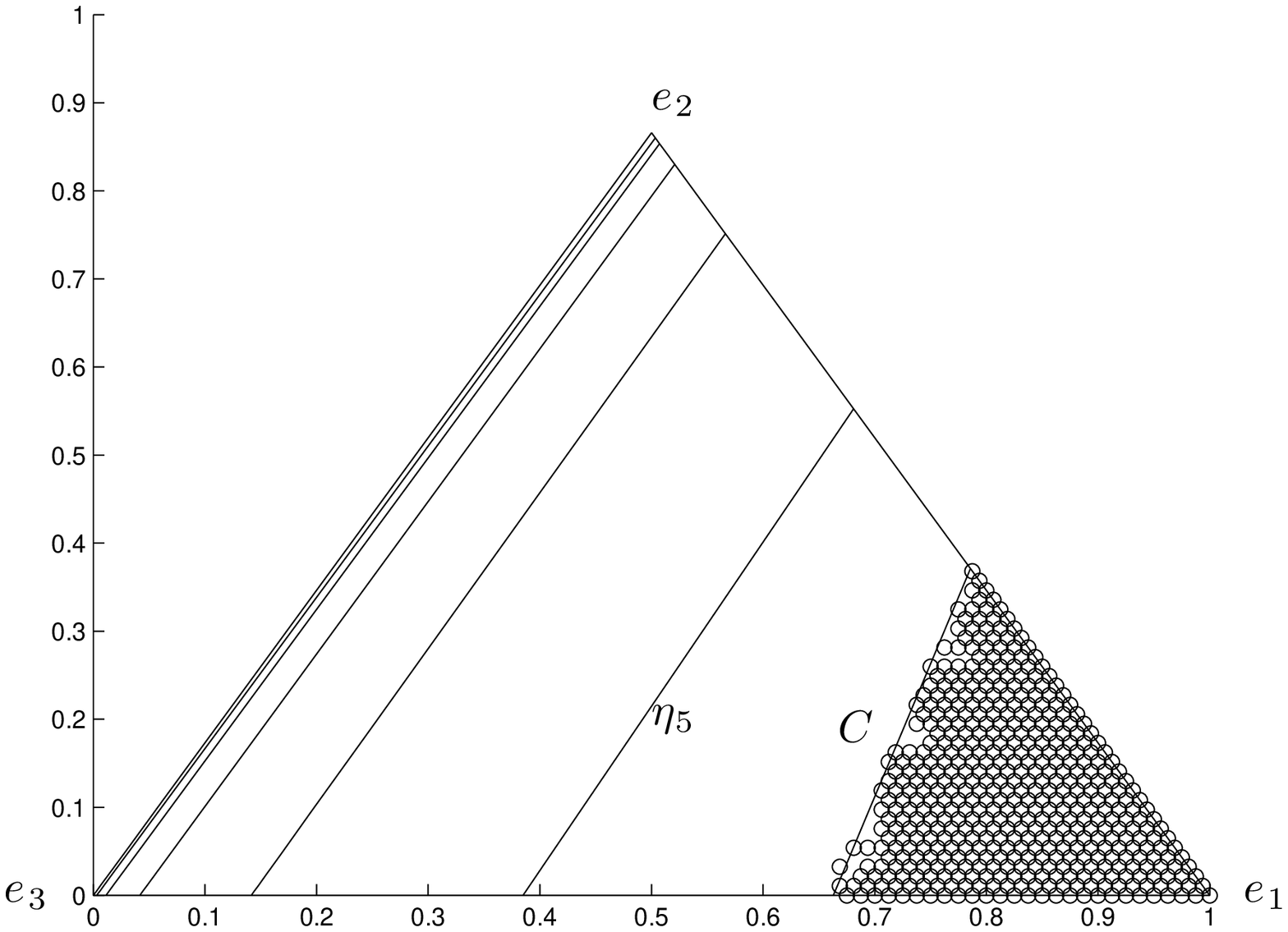,width=0.45\linewidth}}} \\
\mbox{\subfigure[$P^{(3)}$]
{\epsfig{figure=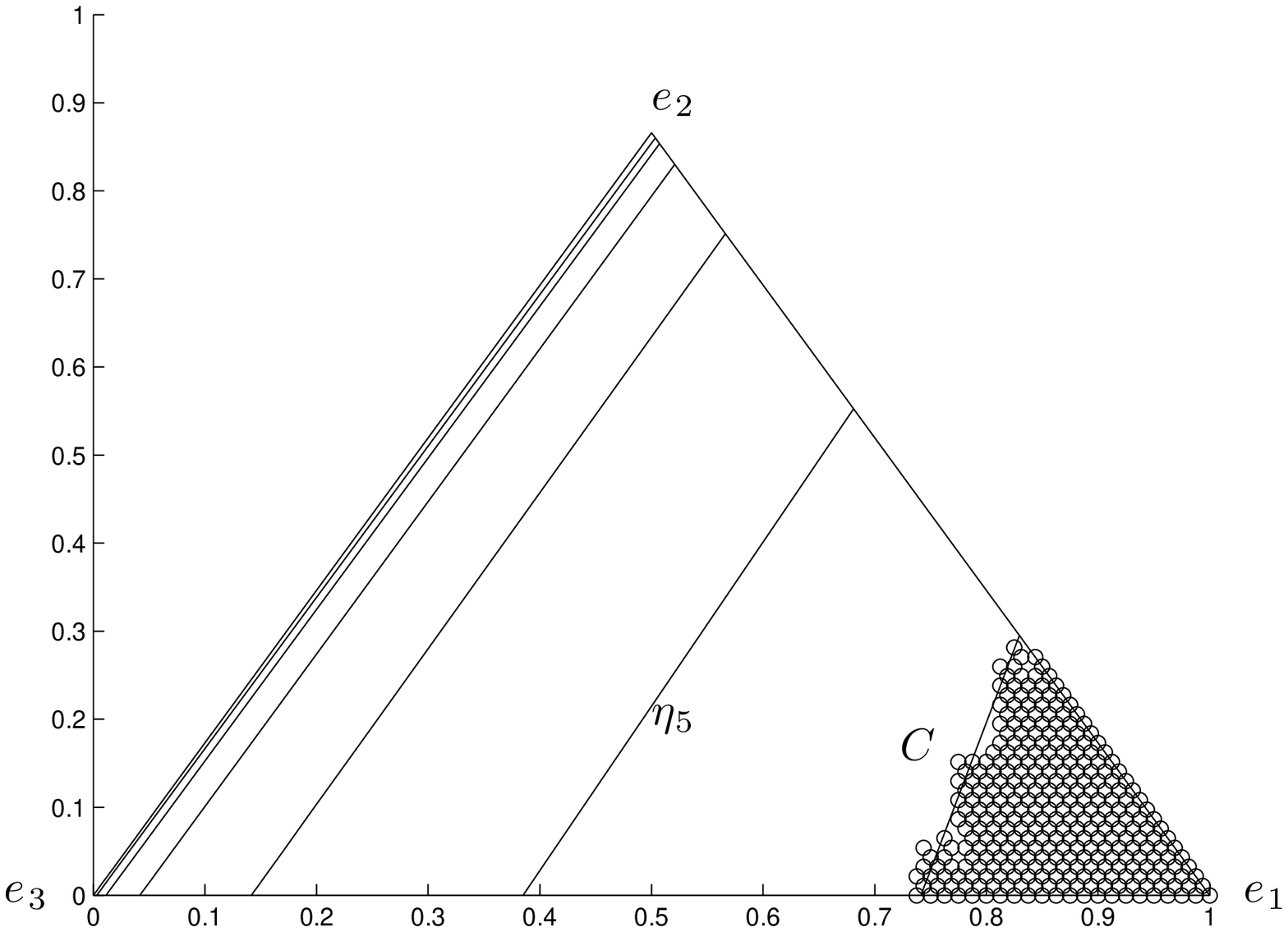,width=0.45\linewidth}}  \quad
\subfigure[$P^{(4)}$]{\epsfig{figure=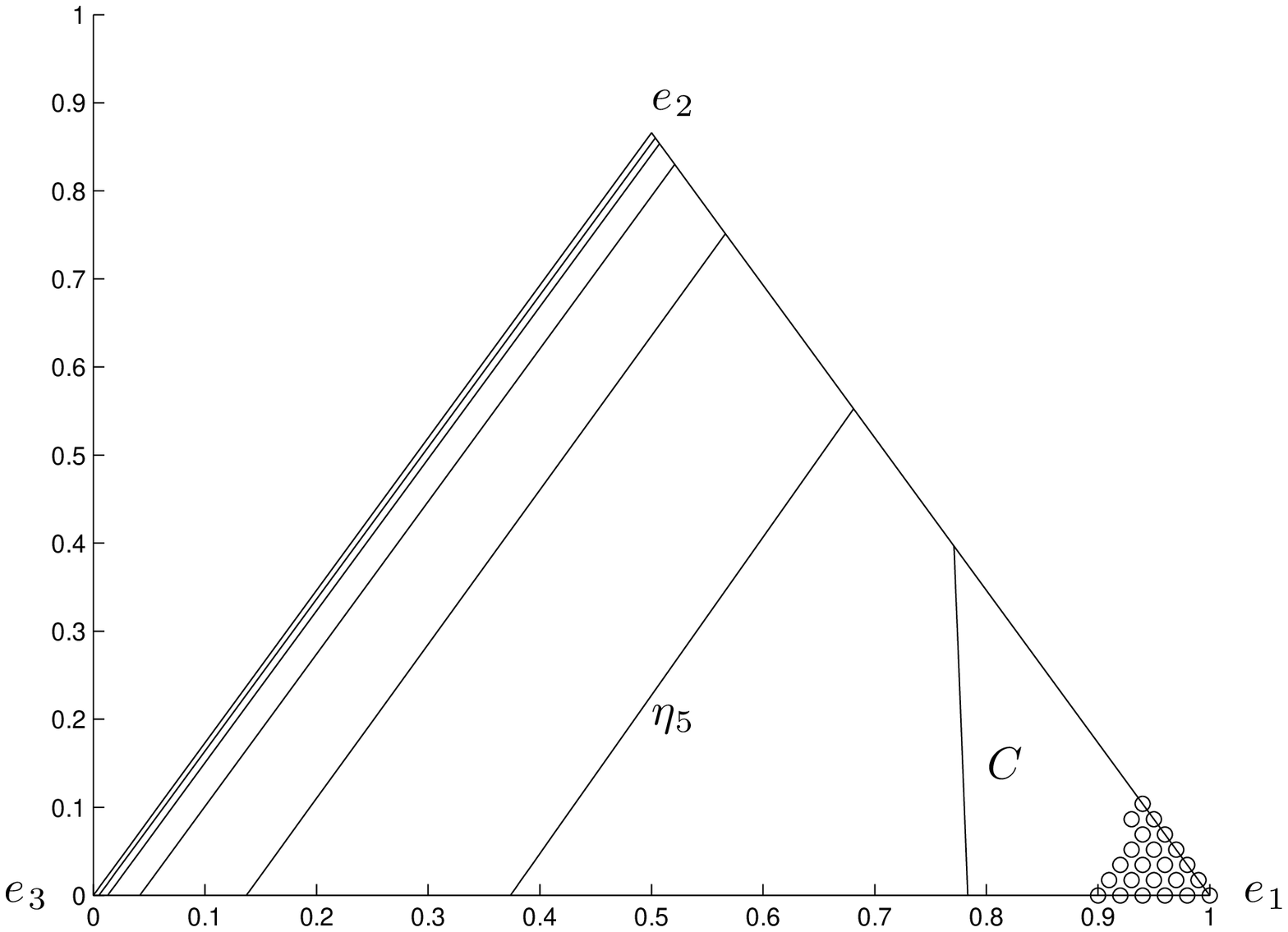,width=0.45\linewidth}}}
\caption{Optimal decision policy for quickest time change time with geometric probability mass function for geometric distribution (transition probability $P^{(1)}$), 
and  phase-type distributions (transition probabilities $P^{(2)}$, $P^{(3)}$ and $P^{(4)}$). The shaded region
depicts the stopping set $\Stop$ in (\ref{eq:stopset}). The parameters are specified in Example 2 of 
Sec.\ref{sec:numerical}.
 } \label{fig:3stateplots}
\end{figure}

\section{Conclusions} Motivated by understanding how local and global decision making interact, 
this paper has presented structural results for quickest time detection when agents
 perform social learning. Also a related model incorporating multi-agent sensor scheduling and quickest time detection
was considered. Unlike classical quickest detection, the optimal policy can have multiple thresholds.
Four main results were presented.
First, Theorem \ref{thm:blackwell}  showed using Blackwell dominance of measures  that social learning based quickest detection always results in more expensive cost
compared to classical quickest detection.
Second, for symmetric observation probabilities and  geometric change times, the explicit
multi-threshold behavior of social learning based quickest detection was characterized in Theorem \ref{thm:ep0} by approximating with a simpler detection problem.
Third, quickest time change detection for more general PH-type distributed change times was considered.
Theorem \ref{thm:cpi} gave sufficient conditions for the optimal policy to be characterized by a single linear hyperplane in the multi-dimensional
simplex of posterior distributions.
Finally, using lattice programming and likelihood ratio dominance Theorem \ref{thm:1} gave sufficient conditions for the optimal policy to be characterized by a single switching curve. The
optimal linear approximation to this curve (that preserves the MLR monotone nature of the policy) was characterized in Theorem \ref{thm:dep}.

The results of this paper are straightforwardly extended to more general stopping problems where the underlying Markov state does not
have an absorbing state, as long as the transition matrix satisfies assumption (A2). In current work, we are using similar social learning models
for ``order-book" trades in agent based models for algorithmic market making, see also \cite{PS11}.

\appendix

\section{Proofs of Theorems} \label{sec:proofs}

\subsection{Preliminaries: Stochastic Dominance,  Submodularity} \label{sec:mlrdef}
Excellent
background references for 
 stochastic dominance and lattice programming are \cite{Top98,Kij97,MS02,KR80}.
The proofs of Theorem \ref{lem:polytopes1} and 
Theorem \ref{thm:1}  require concepts in stochastic dominance.
In particular,  Statement (iv) of Theorem \ref{thm:1}  states
that the optimal social policy $\mu^*(\pi)$ is monotonically increasing in belief
state $\pi$. 
 In order to compare  belief states
$\pi$ and $\tpi$, 
we will use the monotone likelihood ratio (MLR)
stochastic ordering and a specialized version of the MLR order restricted to lines in
the simplex $\I$.
The MLR order is  useful for social
learning 
 since it is preserved
 after conditioning  \cite{Rie91,KR80,MS02}. 
 
\begin{definition}[MLR ordering, {\cite[pp.12--15]{MS02}}]
Let $\pi_1, \pi_2 \in \I$ be any two belief state vectors.
Then $\pi_1$ is greater than $\pi_2$ with respect to the MLR ordering -- denoted as
$\pi_1 \gr \pi_2$,
 if 
\beq \pi_1(i) \pi_2(j) \leq \pi_2(i) \pi_1(j), \quad i < j, i,j\in \{1,\ldots,X\}. 
\label{eq:mlrorder}\eeq
\end{definition}

\begin{definition}[First order stochastic dominance]
 Let $\pi_1 ,\pi_2 \in \I$.
Then $\pi_1$ first order stochastically dominates $\pi_2$  -- denoted as
$\pi_1 \gs \pi_2$ --
 if 
$\sum_{i=j}^X \pi_1(i) \geq \sum_{i=j}^X \pi_2(i)$  for $ j=1,\ldots,X$.
\end{definition}

\begin{result}[\cite{MS02}] \label{res1}
 (i)  
 $\pi_1 \gr \pi_2$ implies $\pi_1\gs \pi_2$. (For $X=2$,  $\gr$ and $\gs$ are equivalent) \\
(ii) Let $\mathcal{V}$ denote the set of all $S$ dimensional vectors
$v$ with 
 nondecreasing components, i.e., $v_1 \leq v_2 \leq \cdots
v_X$.
Then $\pi_1 \gs \pi_2$ iff for all $v \in \mathcal{V}$,
 $v^\p \pi_1 \geq v^\p \pi_2$. \\
 (iii) Suppose $f_i \geq g_i$, $i=1,\ldots,X$ and $f_i,g_i$ are increasing in $i$. Then $\pi \gs \bp$  implies $\sum_i f_i \pi_i \geq \sum_i g_i \bp_i$.
 (This follows since from (ii) $\sum_i g_i \pi_i \geq \sum_i g_i \bp_i$ and $\sum_i f_i \pi_i >\sum_i g_i \pi_i$ since $f_i > g_i$ $\forall i$).
\end{result}

For state-space dimension $X =2$, MLR is a complete order and coincides with
first order stochastic dominance. 
For state-space dimension $X >2$,
MLR is a  {\em partial order}, i.e., $[\I,\gr]$ is a partially ordered set (poset) since it is not always
possible to order any two belief states $\pi \in \I$. 

Finally, we define a modification of the MLR order on certain line
segments in the simplex
which yields   a total ordering.

Define the set of belief states $\mathcal{H}_i = \{\pi \in \I: \pi(i) = 0 \}$.
For each belief state $\bp \in \H_i$, denote the line segment
 $\l(e_{i},\bp)$ that connects $\bp$ to $e_{i}$. 
Thus
\beq \l(e_{i},\bp) = \{\pi \in \I: \pi = (1-\epsilon) \bp + \epsilon e_{i}, \;
0 \leq \epsilon \leq 1 \} ,
 \bp \in \H_i.  \label{eq:lines}
 \eeq

\begin{definition}[MLR ordering  ${\gl}$ and $\glx$  on  lines]  \label{def:tp2l}
 $\pi_1$ is greater than $\pi_2$ with respect to the MLR ordering on
the line $\l(e_{1},\bp)$  -- denoted as $\pi_1\gl \pi_2$ if 
$\pi_1,\pi_2 \in \l(e_1,\bp)$  for  some $\bp \in \H_1$   and
$\pi_1 \gr \pi_2$. Similarly, $\pi_1 \glx \pi_2$, if 
$\pi_1,\pi_2 \in \l(e_X,\bp)$) for  some $\bp \in \H_X$,  and
$\pi_1 \gr \pi_2$.
\end{definition}

Note that $[\I,\gl]$ is a  chain, i.e., all elements
$\pi,\tpi \in \l(e_{1},\bp)$ are comparable, i.e., either $\pi\gl \tpi$ or $\tpi \gl \pi$. 
Similarly  $[\I,\glx]$ is a chain.
In Lemma \ref{lem:convex}, we summarize useful properties of $[\I,\gl]$ that
will be used in our proofs.

\begin{lemma} \label{lem:convex} Consider
 $[\I,\gr]$, $[\l(e_X,\bp),\gl]$.
(i) On  $[\I,\gr]$, $e_1$ is the least and  $e_X$ is the greatest element.
On $[\l(e_X,\bp),\gl]$, $\bp$ is the least  and $e_X$ is the greatest. \\
(ii) Convex combinations of MLR comparable belief states form a chain. 
For any $\gamma \in [0,1]$,
$\pi \lr \tpi \implies \pi \lr \gamma\pi + (1-\gamma) \tpi \lr \tpi $.\\
(iii) All points on a line $\l(e_X,\bp)$ 
are MLR comparable. Consider  any two points
$\pi^{\gamma_1},\pi^{\gamma_2}\in \l(e_X,\bp)$ (\ref{eq:lines}) where
$\pi^\gamma = \gamma e_X + (1-\gamma) \bp$.
Then 
$\gamma_1 \geq \gamma_2$, implies $\pi^{\gamma_1}
\gl \pi^{\gamma_2}$. A similar result holds for $\l(e_1,\bp)$. \end{lemma}

\begin{definition}[Submodular function \cite{Top98}] \label{def:supermod}
 $f:\l(e_1,\bp)\times \u \rightarrow \reals$  is  submodular (antitone differences)
if 
$f(\pi,u) - f(\pi,\bar{u}) \leq f(\tilde{\pi},u)-f(\tilde{\pi},\bar{u})$, for
$\bar{u} \leq u$, $\pi \gl \tpi$.
\end{definition}


The following result says that for a submodular function  $Q(\pi,u)$, 
$u^*(\pi)=\argmin_u Q(\pi,u)$  is increasing in its argument $\pi$. This implies $\mu^*(\pi)$ is MLR increasing  on 
the line segments $\l(e_{x},\bp)$, which in turn will be used to prove the existence
of as threshold decision curve.

\begin{theorem}[\cite{Top98}] \label{res:monotone}
\label{res:supermod} If $f:\l(e_1,\bp)\times \u \rightarrow \reals$ is submodular, then
there exists a \\ $u^*(\pi) = \argmin_{u\in \u} f(\pi,u)$, that  is increasing on
 $[\l(e_1,\bp),\gl]$,
i.e., $\tpi \gl {\pi} \implies u^*(\pi) \leq u^*(\tpi)$.
\end{theorem}

\begin{definition}[Single Crossing Condition \cite{Top98,Ami05}] \label{def:scc}
$g:\Y\times \A\rightarrow \reals $ satisfies a single crossing condition in  $(y,a)$  if $g(y,a) - g(y,\bar{a}) \geq 0$ implies $g(\bar{y},a)- g(\bar{y},\bar{a}) \geq 0$ for $\bar{a}>a$ and $\bar{y} > y$.
Then $a^*(y) = \argmin_{a} g(y,a)$ is increasing in $y$.
\end{definition}

\begin{definition}[TP2 ordering and Reflexive TP2 distributions]
Let $P $ and $Q$ 
 denote any two multivariate probability mass functions.
Then: \\
 (i) $P \gtp Q$ if $P(\i) Q(\j) \leq P(\i \vee  \j) Q(\i \wedge \j)$.
 If $P$ and $Q$ are univariate, then this definition is equivalent to
 the MLR ordering  $P\gr Q$ defined above.\\
(ii)  A multivariate
distribution $P$ is said to be multivariate TP2 if $P \gtp P$  holds,
i.e.,  $P(\i) P(\j) \leq P(\i\vee  \j) P(\i \wedge \j)$.\\
(iii) If
 $\i,\j\in \{1,\ldots,X\}$ are scalar indices,
Statement (ii) is equivalent to saying that a $M \times N$ matrix $A$ is  TP2  if all 
second order minors are non-negative. 
 if $i \geq j$, then  the $i$-th row of $A$ MLR dominates the 
$j$-th row.
 \label{def:tp2}
\end{definition}

 \subsection{Proof of Theorem \ref{thm:blackwell}} \label{sec:pblackwell}

 Let $\uV_k(\pi)$ denote the value function at iteration $k$ of the value iteration algorithm
(\ref{eq:vi}) associated with the classical quickest detection Bellman equation (\ref{eq:dp_algc}).
Recall $V_k(\pi)$ is the value function associated with the social learning based quickest detection problem (\ref{eq:dp_alg}).

We start with the following lemma which is proved at the end of Appendix \ref{sec:pblackwell}

\begin{lemma} \label{lem:blackwell}  
$ \sum_a \uV_k(\T(\pi,a)) \sigma(\pi,a)  \geq \sum_y \uV_k(T(\pi,y)) \sigma(\pi,y) $.
\end{lemma}
 
The proof of Theorem~\ref{thm:blackwell} then follows by mathematical induction using the value iteration algorithm~(\ref{eq:vi}).
Assume $V_k(\pi) \geq \uV_k(\pi)$ for $\pi \in  \I$.
Then 
\begin{align*}C(\pi,2) +  \sum_a V_k(\T(\pi,a)) \sigma(\pi,a) & \geq
C(\pi,2) +  \sum_a \uV_k(\T(\pi,a)) \sigma(\pi,a) \\  & \geq C(\pi,2) + \sum_y \uV_k(T(\pi,y)) \sigma(\pi,y)  
 \end{align*}
where the second inequality follows from Lemma \ref{lem:blackwell}.
Thus  $V_{k+1}(\pi) \geq \uV_{k+1}(\pi)$. This completes the induction step. Since value iteration converges pointwise,
$V(\pi) \geq \uV(\pi)$ thus proving the theorem.

{\bf Proof of Lemma \ref{lem:blackwell}}. \\
{\em Step 1}: First, let us show that $\uV_k(\pi)$ is concave over $\I$ for any $k$ by induction.
Recall from (\ref{eq:vi}) that $\uV_0(\pi) = - \Cb(\pi,1)$ which is linear in $\pi \in \I$.
Assume $\uV_k(\pi)$ is concave at iteration $k$. Note that $\uV_k(k)$ is positively homogeneous, i.e.,
for any $c\geq 0$, $\uV_k(c \pi) = c \uV_k(\pi)$. So the value iteration algorithm (\ref{eq:vi}) associated with Bellman's 
equation (\ref{eq:dp_algc}) is 
$$
 \uV_{k+1}(\pi) = \min\{C^\p \pi + \rho \sum_y \uV_k(B_y P^\p \pi) , 0 \} 
$$
Since the composition of concave function with a linear function preserves concavity, therefore
$ \sum_y \uV_k(B_y P^\p \pi) $ is concave and so $\uV_{k+1}(\pi)$ is concave.

{\em Step 2}:
 We then use the Blackwell dominance condition (\ref{eq:aprob}). The social learning filter (\ref{eq:piupdate}) can be expressed in terms of the
 Hidden Markov Model filter (\ref{eq:privu}) as
$$ \T(\pi,a ) =   \sum_{y \in \Y} T(\pi,y) \frac{\sigp(\pi,y)}{\sigp(\pi,a)} P(a|y,\pi) 
\quad \text{ and } \sigp(\pi,a) = \sum_{y \in \Y} \sigp(\pi,y) P(a|y,\pi).$$
Therefore, $\frac{\sigp(\pi,y)}{\sigp(\pi,a)} P(a|y,\pi) $ is a probability measure wrt $y$.
Since from Step 1, $\uV_k(\cdot)$ is concave for $\pi \in \I$, using Jensen's inequality it follows that
\begin{align*}
\uV_k(\T(\pi,a) ) & = \uV_k \left(\sum_{y \in \Y} T(\pi,y) \frac{\sigp(\pi,y)}{\sigs(\pi,a)} P(a|y,\pi) \right)
\geq \sum_{y \in \Y}  \uV_k (T(\pi,y)) \frac{\sigp(\pi,y)}{\sigp(\pi,a)} P(a|y,\pi)\\
\text{ implying }&  \sum_{a}  \uV_k(\T(\pi,a) ) \sigp(\pi,a) \geq
\sum_{y} \uV_k(T(\pi,y)\sigp(\pi,y).
\end{align*}

\subsection{Proof of Theorem \ref{lem:polytopes1}} \label{sec:polytopes1}
Here we present  a detailed version of Theorem \ref{lem:polytopes1} that was presented in Sec.\ref{sec:assumptions}.
\begin{reptheorem}{lem:polytopes1}[Detailed version]  Under (A1),  (A2), (S), \\
(i) The local decision $a^*(\pi,y) = \arg\min_a c_a^\p B_y P^\p \pi$  (see (\ref{eq:step2})) 
 is increasing in $y$.\\
 (ii) $a^*(\pi,y) $ is 
MLR increasing in $\pi$,  i.e., $\pi \gr \bp \implies a^*(\pi,y) \geq a^*(\bp,y)$.
   \\
(iii) The $Y$ linear hyperplanes $ (c_1 - c_2)^\p B_y P^\p \pi = 0$, $y=1,\ldots,Y$ 
 do not intersect
within the interior of the belief space $\I$. Thus, out of the $2^Y$ polytopes in (\ref{eq:expopoly}),
there are a maximum of $Y+1$ non-empty polytopes in $\Pi$, namely (\ref{eq:reduced}). \\
(iv) Let $i^*_y = \max\{i: e_i \in \{\pi:  (c_1 - c_2)^\p B_y P^\p \pi < 0\}$. Then
each of the $Y$  hyperplanes $ (c_1 - c_2)^\p B_y P^\p \pi = 0$, $y=1,\ldots,Y$ partitions $\I$
such that  the vertices $e_1,e_2,\ldots,e_{{i^*_y}}$ lie in the convex polytope
$ (c_1 - c_2)^\p B_y P^\p \pi < 0$ and  the vertices $e_{{i^*_y+1}},\ldots, e_X$ lie in the convex polytope
$ (c_1 - c_2)^\p B_y P^\p \pi > 0$.\\
(v) 
 $i^*_y$ decreases with $y$.\\
 (vi) $M^\pi$ defined in (\ref{eq:aprob}) has the following structure:
\beq  M^\pi = \begin{bmatrix} 1_{Y-l+1} & 0_{Y-l+1} \\
  0_{l-1} &  1_{l-1}  \end{bmatrix}, \; \text{ for } \pi \in \mathcal{P}_l , l=1,\ldots,Y+1 \label{eq:Mstructure}\eeq
 \qed
\end{reptheorem}

{\bf Proof}:
(i)  From \cite[Lemma1.2(1)]{Lov87} if $B$ and $P$ are  TP2 (i.e., (A1), (A2) hold)  then 
$\frac{B_y P^\p \pi}{1^\p B_y P^\p \pi} \lr \frac{B_{y+1} P^\p \pi}{\ones^\p B_{y+1} P^\p \pi}$. Next MLR dominance implies first order stochastic dominance.
Then since from (S),  $c(i,1)-c(i,2)$ is increasing in $i$, it follows that 
$\frac{(c_1-c_2)^\p B_y P^\p \pi}{\ones^\p B_y P^\p \pi} <  \frac{(c_1-c_2)^\p B_{y+1} P^\p \pi}{\ones^\p B_{y+1} P^\p \pi}$. Since the denominators are non-negative,
this implies that 
$(c_1-c_2)^\p B_y P^\p \pi \geq 0 \implies (c_1-c_2)^\p B_{y+1} P^\p \pi\geq 0$. That is, 
the single crossing condition (\ref{eq:scmean}) holds, see Definition \ref{def:scc}. So $a^*(\pi,y) \uparrow y$.

(ii) To prove $a^*(\pi,y) \uparrow \pi$ wrt $\gr$, we use a similar approach to Part (i). From \cite[Lemma 1.2(2)]{Lov87}, assuming (A3), $\pi \lr \bp$ implies
$\frac{B_y P^\p \pi}{\ones^\p B_y P^\p \pi} \lr \frac{B_{y} P^\p \bp}{\ones^\p B_{y} P^\p \bp}$. As in the proof above, using (S) this implies
$\frac{(c_1-c_2)^\p B_y P^\p \pi}{\ones^\p B_y P^\p \pi} <  \frac{(c_1-c_2)^\p B_{y} P^\p \bp}{\ones^\p B_{y} P^\p \bp}$. Since the denominators are non-negative,
this implies that 
 \beq\pi \lr \bp \text{ and }
 (c_1-c_2)^\p B_y P^\p \pi \geq 0 \implies (c_1-c_2)^\p B_{y} P^\p \bp\geq 0. \label{eq:scpi}
 \eeq
  That is a single crossing condition (see Definition \ref{def:scc}) holds wrt $(\pi,a)$ and
the partial order $\gr$. So $a^*(\pi,y) \uparrow \pi$.

(iii) follows immediately from  (\ref{eq:scmean}).

 (iv) Since $i^*_y = \max\{i: e_i \in \{\pi:  (c_1 - c_2)^\p B_y P^\p \pi < 0\}$, clearly
 $e_{i^*_y+1} \in \{\pi: (c_1 - c_2)^\p B_y P^\p \pi > 0\}$.
Next since $e_{i^*_y +1} \lr e_{i^*_y +2}
 \cdots \lr e_X$,  the single crossing condition (\ref{eq:scpi}) yields $ (c_1-c_2)^\p B_y P^\p e_{i^*_y+2}\geq 0,
 \ldots, (c_1-c_2)^\p B_y P^\p e_{X}\geq 0$. 

(v) Start with the single crossing condition (\ref{eq:scmean}) repeated below for clarity:
\begin{align*}
 \{\pi: (c_1-c_2)^\p B_{y+1} P^\p \pi \leq 0 \} &\subseteq \{\pi: (c_1-c_2)^\p B_y P^\p \pi \leq 0 \} \\
\text{ Therefore } \max\{i: e_i \in \{\pi:  (c_1 - c_2)^\p B_{y+1} P^\p \pi \leq 0\} & \leq
\max\{i: e_i \in \{\pi:  (c_1 - c_2)^\p B_{y} P^\p \pi \leq 0\} \end{align*}
 
 (vi) follows by enumerating all matrices $M^\pi$ that satisfy (i) and (ii); see (\ref{eq:Mexample})
 for an example.

\subsection{Proof of Theorem \ref{thm:ep0}}
Similar to the example given below Theorem \ref{lem:polytopes1},
it can be verified from (\ref{eq:aprob}) that  there are only 3 possible values for $\Bs$, namely,
\beq  \label{eq:bsvalues}\Bs=\begin{bmatrix} 0 & 1 \\ 0 & 1 \end{bmatrix},
\pi \in \mathcal{P}_1, \quad
\Bs =  B, \pi \in
 \mathcal{P}_2 , \quad \Bs = \begin{bmatrix} 1 & 0 \\ 1 & 0 \end{bmatrix}, \text{ and }
\pi \in \mathcal{P}_3.\eeq
Thus Bellman's equation 
 (\ref{eq:dp_alg}),
reads
\begin{multline}
V(\pi) = \min\{ C(\pi,2) + \rho V(\pi) I(\pi \in \mathcal{P}_1) + \rho \sum_{a\in \A} V(\T(\pi,a)) \sigma(\pi,a) \, I(\pi \in \mathcal{P}_2) 
 \\   +   \rho V(\pi) I(\pi \in \mathcal{P}_3)  , 0 \}  \label{eq:visocialstop}\end{multline}
 {\em Claim (i)}:
 For $\pi \in \mathcal{P}_1 \cup \mathcal{P}_3$,
 $V(\pi) = \min\{ C(\pi,2) + \rho V(\pi), 0\}$. This can be solved explicitly
as
 $$ V(\pi) = \min \{C(\pi,2)/(1-\rho), 0\} \text{ implying } 
\mu^*(\pi) = \begin{cases} 1 & C(\pi,2) < 0 \\
 						2 & C(\pi,2) \geq 0 \end{cases} .$$
Since $C(\pi,2)$ is MLR decreasing in $\pi$, the optimal policy for $\pi \in \mathcal{P}_1 \cup \mathcal{P}_3$ is a threshold policy
with threshold at $C(\pi^*,2) = 0$. This proves the first claim of the theorem.
 \\
 {\em Claim (ii)} Since $\T(\pi,a) = \pi$ for $\pi \in \mathcal{P}_1 \cup \mathcal{P}_3$, the private belief state update 
 (\ref{eq:privu})
 freezes in these regions,
 i.e., $\pi_{k-1} \in \mathcal{P}_1 \cup \mathcal{P}_3$ implies
that , $\hy_k = \pi_{k-1}$.
 Therefore all agents take the same local decision $a$ according to (\ref{eq:step2}) implying an information cascade.

{\em Claim (iii}) The proof of this 
 is more involved.

We need the following property of the social learning Bayesian filter which is a detailed
version of Lemma \ref{lem:fixed} in Sec.\ref{sec:convex}.
Since we are going to partition $\I$ into four intervals, namely $[0,\hy_2(2))$, $[\hy_2(2),\q(2))$, $[q(2), \hy_1(2))$ and $[\hy_1(2),1]$, it is convenient to introduce the following notation:
Denote these intervals as $\D^1,\D^2,\D^3,\D^4$, respectively.
Note $\mathcal{P}_1 = \D_1$, $\mathcal{P}_2 = \D_2 \cup \D_3$, $\mathcal{P}_3 = \D_4$.

\begin{replemma}{lem:fixed}[Detailed version]
Consider the  social learning Bayesian filter  (\ref{eq:piupdate}). 
Then
$T^{\hy_1}(\hy_1,1) = q$,  $T^{\hy_2}(\hy_2,2) = q$.  Furthermore if $B$ is symmetric TP2, then 
$T^{q}(q,2) = \hy_1$,  $ T^{q}(q,1) = \hy_2$ and $\hy_2 \lr q\lr \hy_1$.
So \\ (i)  $\pi \in \D_2$  implies   $\T(\pi,2) \in \D_1$ and 
   $\T(\pi,1)  \in \D_3$. \\
(ii)  $\pi \in \D_3$  implies   $\T(\pi,2)  \in \D_2$ and
  $\T(\pi,1)  \in \D_4$. \qed
\end{replemma}
The proof of Lemma \ref{lem:fixed} is as follows. Recall from (\ref{eq:bsvalues}) that
on interval $\mathcal{P}_2$, $\Bs = B$.
Then it is straightforwardly verified from 
(\ref{eq:piupdate}) that $T^{\hy_1}(\hy_1,1) = T^{\hy_2}(\hy_2,2) = q$.
Next, using (\ref{eq:piupdate}) it follows that $B_{12}B_{11} = B_{22} B_{21}$ is a sufficient
condition for $T^{q}(q,2) = \hy_1$ and $T^{q}(q,1) = \hy_2$. 
Also,  since by (A1) $B$ is TP2,  applying Theorem \ref{thm:key}(2), implies $\hy_2 \lr 
q \lr  \hy_1$.
So $B$ symmetric TP2 is sufficient for the claims of the lemma to hold.
Statements (i) and (ii) then follow straightforwardly. In particular, from Theorem \ref{thm:key}(1), 
$\hy_1 \gr \pi \gr q$ implies $T^{\hy_1}(\hy_1,1) = q \gr \T(\pi,1) \gr T^{q}(q,1) = \hy_2$,
which implies Statement (i) of the lemma. Statement (ii) follows similarly.

Returning to the proof of Theorem \ref{thm:ep0}, we use mathematical induction on the value iteration algorithm (\ref{eq:vi}).
Clearly $V_0(\pi) = - \Cb(\pi,1)$ is linear.
Assume now that $V_k(\pi)$ is  piecewise linear and concave on each 
of the four intervals $\D_1,\ldots\D_4$.
That is, for two dimensional vectors $\gamma_{m_l}$ in the set $\Gamma_l$,
$$
V_k(\pi) = \sum_l \min_{m_l \in \Gamma_l}  \gamma_{m_l}^\p \pi\, I(\pi \in \D_l) $$
Consider $\pi \in \D_2$. 
From (\ref{eq:bsvalues}), since 
 $\Bs_a = B_a$, $a = 1,2$,
Lemma \ref{lem:fixed} (i) together with the value iteration algorithm (\ref{eq:visocialstop}) yields
$$ V_{n+1}(\pi) =  \min\{ C(\pi,2) + \rho \left[\min_{m_3 \in \Gamma_3} \gamma_{m_3}^\p B_1 \pi + \min_{m_1 \in \Gamma_1}
\gamma_{m_1}^\p B_2 \pi\right], 0 \}. $$
Note the crucial point in the above equation: as a result of Lemma \ref{lem:fixed} (i) -- the
social learning filter maps $\D_2$ 
to only $\D_1$ (for $a=1$) and $\D_3$ (for $a=2$).
Since each of the terms in the above equation  are piecewise linear and concave, it follows that 
 $V_{k+1}(\pi)$ is piecewise linear and concave on $\D_2$. A similar proof holds for $\D_3$ and this involves using Lemma \ref{lem:fixed}(ii).
As a result the stopping set on each interval $\D_l$, $l=1,\ldots,4$ is a convex region, i.e., an interval. This proves claim (ii).

\subsection{Proof of Lemma \ref{lem:mappoly} and Theorem \ref{thm:cpi}} \label{app:cpi}

\subsubsection*{Proof of Lemma \ref{lem:mappoly}}
Let us introduce the following notation. Define
\beq S^+ = \{\pi: C^\p \pi > 0\} \; \text{  and } S^= = \{\pi: C^\p \pi = 0\}.  \label{eq:spe} \eeq
The proof comprises of three parts.\\
{\em Statement (i)}: Under (PH), for every $\pi \in S^+$, there exists a $\bpi \in S^=$ such that
$\bpi \gr \pi$.\\
{\em Proof}:
Consider any belief state $\pi \in S^+$.
Construct a line segment from $e_1$ through the belief state $\pi$ and let this line segment intersect the hyperplane $S^=$.
Denote $\bpi$ as this point of intersection.
Clearly $\bpi = \alpha e_1 + (1-\alpha) \pi$ where $\alpha = C^\p \pi/(C^\p \pi - C_1)$. It is straightforwardly established
that $\bpi \gr \pi$ if $\alpha \geq 1$ which is clearly true since $C_1 > 0$ and $C^\p \pi > 0$ for $\pi \in S^+$.
\\
{\em Statement (ii)}: Under (A1), (A2), (A3),  if $\bpi \gr \pi$, then 
$C^\p  T^{\bpi}(\bpi,a) > 0 \implies C^\p  \T(\pi,a) > 0$.
\\ {\em Proof}:
Under (A1), (A2), (A3),  it follows from Theorem \ref{thm:key}(2) in  Appendix \ref{app:thm1} that $\T(\pi,a)$ is MLR increasing, that is, $\bpi \gr \pi$ implies $T^{\bpi}(\bpi,a) \gr \T(\pi,a)$.
Under (A3), the elements of $C$  are decreasing.
So  from Result \ref{res1} in Appendix \ref{sec:mlrdef}, it follows that
 $C^\p  T^{\bpi}(\bpi,a) \leq C^\p \T(\pi,a)  $.
So 
$C^\p T^{\bpi}(\bpi,a) > 0$ implies
$C^\p \T(\pi,a) > 0$.
\\
Statements (i) and (ii) imply that if the social learning filter~(\ref{eq:piupdate})  maps  belief states in $S^=$  to $S$, then all belief states in $\{\pi:C^\p \pi > 0\}$ are also mapped to $S$.
 Since the hyperplane  $S^= = \{\pi:C^\p \pi = 0\}$ has infinite points, how can we formulate a sufficient
  condition for belief states $\{\pi:C^\p \pi = 0\}$ to be mapped to the polytope $\mathcal{P}_{Y+1}$?
(C1) serves as a sufficient condition as proved in Statement (iii) below.\\
{\em Statement (iii)}: A sufficient condition for $C^\p  T^{\bpi}(\bpi,a) < 0 $ to hold for all $\bpi \in S^=$ is that 
$C^\p  T^{\ver_i}(\ver_i,a) < 0 $ for all  $X-1$ vertices $\ver_j$ of (\ref{eq:ver}).\\
{\em Proof}: Clearly every belief state $\pi \in S^=$ is a convex combination of the vertices, i.e.,
$\pi = \sum_i \alpha_i \ver_i$, for some  $\alpha_i \geq 0$ and $\sum_i \alpha_i = 1$. Now
$C^\p T^{\ver_i}(\ver_i,a) < 0$ is equivalent to $ C^\p \Bsi^{\ver_i} P^\p \ver_i  < 0$ since the normalization term in $\T(\cdot)$ is 
non-negative.
This implies $C^\p \sum_i \alpha_i  \Bsi^{\ver_i} P^\p \ver_i < 0$, and this is equivalent
to $C^\p  \Bs P^\p \pi < 0$.

\subsubsection*{Proof of Theorem \ref{thm:cpi}}
Define  $S = \{\pi:C^\p \pi \geq 0\}$.\\
{\em Step 1}: We first prove that $V(\pi) = 0$ for $\pi \in S$. This is equivalent to saying 
that for $\{\pi:C^\p \pi \geq 0\}$, the optimal policy $\mu^*(\pi) = 1$.

The proof of Step 1 is 
 by induction on the value iteration algorithm
(\ref{eq:vi}).
Suppose $V_0(\pi) = 0$. Then it trivially satisfies $V_0(\pi) = 0$ for $\pi \in S$. Next suppose
$V_k(\pi) = 0$ for $\pi \in S$. Then for $\pi \in S$, Assumption (C1) implies that $\T(\pi,a) $ belongs to $S$
implying that $V(\T(\pi,a)) = 0$. So from (\ref{eq:vi}), it follows that $V_{k+1}(\pi) = \min \{C^\p \pi, 0 \} = 0$ since $C^\p \pi \geq 0$
for $\pi \in S$.
Since $V_k(\pi)$ converges pointwise to $V(\pi)$, Step 1 follows.
For initial condition $V_0(\pi) = -\Cb(\pi,1)$ (see (\ref{eq:vi})), $V(\pi)$ obtained as the limit
of the value iteration algorithm is identical to that
with initial condition $V_0(\pi) = 0$.

{\em Step 2}: From Bellman's equation it follows trivially that for $\{\pi: C^\p \pi < 0\}$, $\mu^*(\pi ) =2$.

From  Steps 1 and 2, we have $C^\p \pi \geq  0$  iff $\mu^*(\pi) = 1$.

\subsection{Proof of Theorem \ref{thm:1}} \label{app:thm1}

This section is in two parts. 
We start with several preliminary results that are
similar to the results in \cite{Lov87}.  Then the proof of  Theorem \ref{thm:1} is presented.

\subsubsection{Structural Properties of Social Learning Filter}
\begin{theorem}\label{thm:key} The following structural properties hold for the public belief  update evaluated by 
the social learning Bayesian filter defined in (\ref{eq:piupdate}):
\begin{enumerate}
\item Under (S), $M^\pi$ is TP2 for $\pi \in \I$, see Definition \ref{def:tp2}.
\item Under (A1), (A2), (S) if $\pi_1,\pi_2 \in \mathcal{P}_l$, then
$\pi_1 \gr \pi_2$ implies 
 $T^{\pi_1}(\pi_1,a)\gr T^{\pi_2}(\pi_2,a)$ \label{part3}
\item Under (A1), (A2), (S),
if  $\pi_1,\pi_2 \in \I$,  then  $\pi_1 \gr \pi_2 \implies \sigp(\pi_1,\cdot) \gs \sigp(\pi_2,\cdot)$.
\item Under (A1), (A2), if $\pi \in \I$, then
  $a > \bar{a}$  implies $\T(\pi,a)\gr \T(\pi,\bar{a})$.
\end{enumerate}
\end{theorem}

\begin{proof}
1). We need to show that for fixed $\pi \in \I$,
\beq  M^\pi_{ya} M^\pi_{y'a'} \leq M^\pi_{y \wedge y', a \wedge a'} M^\pi_{y \vee y', a \vee a'}  
\label{eq:tpcheck} \eeq
Recall from (\ref{eq:aprob}) that $M^\pi$ is a matrix with a single 1 in each row at $a^*(\pi,y)$ and all other elements
zero. 
So the only non trivial  case to prove is when both terms on the LHS are 1, i.e.,  $M^\pi_{y,a^*(\pi,y)}=1$ and 
$ M^\pi_{y',a^*(\pi,y')}=1$.
Assuming  (A2), (S), Theorem \ref{lem:polytopes1}(i) says that $a^*(\pi,y) \uparrow y$.
This means that $y < y' \implies a^*(\pi,y) < a^*(\pi,y')$ and $y \geq  y' \implies a^*(\pi,y) \geq a^*(\pi,y')$.
In either case (\ref{eq:tpcheck}) holds with equality since the RHS is identical to the LHS.

2). Since $P$ is TP2 (A2), we have $P^\p  \pi \gr P^\p  \bp$ for $\pi \gr \bp$, see \cite{KR80}.
So it suffices to show that 
$ \frac{\Bs_a \pi }{\ones^\p B^\pi_a \pi} \gr \frac{B^\bp_a \bp }{\ones^\p B^\bp_a \bp} $ for $\pi \gr \bp$.
Moreover, since $\pi,\bp$ belong to the same polytope $\mathcal{P}_l$, $\Bs_a = R^\bp_a = \Bsl_a$ (say), see (\ref{eq:expopoly}).
From \cite{KR80}, a sufficient condition for $ \frac{\Bsl_a \pi }{\ones^\p \Bsl_a \pi} \gr \frac{\Bsl_a \bp }{\ones^\p \Bsl_a \bp} $ is that $\Bsl$ is TP2. Of course
we need this to hold on  each of the $Y+1$  polytopes, i..e, for $l=1,\ldots,Y+1$.

So under what conditions is $\Bs$ TP2 in each of the $Y+1$ polytopes? Note (A2) says $B$ is TP2. Since $\Bs = B M^\pi$ (see (\ref{eq:aprob})) and the product of TP2 matrices
is TP2  \cite[pp.471]{KR80},  it only remains to prove that 
$M^\pi$ is TP2. This follows from (S) as proved in (i) above.

3). Since $P$ is TP2 (A2), it suffices to prove that  $\pi \gr \bp$ implies $\ones^\p B_a^\pi \pi \gs \ones^\p B_a^\bp \bp$, i.e.,
$\sum_i \sum_{a>\bar{a}} B_{ia}^\pi \pi_i \geq \sum_i \sum_{a \geq \bar{a}} B_{ia}^\bp \bp_i $. 
From Statement 1, $M^\pi$ and $M^\bp$ are TP2 and from (A2) $B$ is TP2. So
$B_\pi = B M^\pi$  and $B^\bp = B M^\bp$ are TP2.
Therefore, from Definition \ref{def:tp2}(iii), the rows of $\Bs$ and $R^\bp$ are MLR increasing.
Since MLR dominance implies first order stochastic dominance, this means that  both $\sum_{a>\bar{a}} \Bs_{ia} $ and $\sum_{a>\bar{a}} R^\bp_{ia}$ are increasing with $i$.
Since $\pi \gr \bp$, 
Result~\ref{res1}(i),(ii) and (iii),  imply that a sufficient condition for $\pi \gr \bp \implies\ones^\p \Bs_a \pi \gs \ones^\p R^\bp_a \bp$
 is that $\sum_{a>\bar{a}}\Bs_{ia} > \sum_{a>\bar{a}} R^\bp_{ia} $ or equivalently,
 $\sum_y B_{iy} \sum_{a>\bar{a}} M_{ya}^\pi \geq \sum_y B_{iy} \sum_{a> \bar{a}}  M_{ya}^\bp $. A sufficient condition for this
 is 
 $\pi \gr \bp \implies \sum_{a>\bar{a}} M_{ya}^\pi \geq \sum_{a> \bar{a}}  M_{ya}^\bp$. But this condition holds from the structure of $M$ in (\ref{eq:Mstructure}) and the fact that $a^*(\pi,y)$
 is MLR increasing wrt $\pi$ (Statement (ii) of Theorem \ref{lem:polytopes1} in Appendix \ref{sec:polytopes1}).

4). Since $P$ is TP2 (A2), it suffices to prove that $\pi \gr \bp \implies
\frac{\Bs_a P^\p \pi }{\ones^\p \Bs_a P^\pi \pi} \gr \frac{\Bs_{a'} P^\p \pi }{\ones^\p \Bs_{a'} P^\p \pi}$ for $a \geq a'$. Since $\Bs$ is TP2 (A1),
this result follows straightforwardly from \cite[Theorem 4]{Whi79}.
\end{proof}

\subsubsection{Proof of Theorem \ref{thm:1}}
Here we prove 
Theorem \ref{thm:1}.  The update of belief state in $\mathcal{P}_{Y+1}$ is simple, since
$\Bs_{ia}  = 1/X$ (uniformly distributed) for each $i$, see (\ref{eq:ex31}) for example. In comparison, the  sensor management case of Theorem~\ref{thm:manga} on $\mathcal{P}_2$   with update
given by (\ref{eq:pia}) requires an arbitrary TP2 matrix $\Bs$.    To allow for this generality,
in the proof below, we  assume
$\Bs$ is an arbitrary TP2 matrix on $\mathcal{P}_{Y+1}$.


{\bf Part 1}:  Under (A1), (A2), (A3), (S), (C3), (C2), (PH),   $V(\pi) $ is MLR decreasing  on   polytope $\mathcal{P}_{Y+1}$:\\
The proof of Part 1 is by mathematical induction on the value iteration algorithm (\ref{eq:vi}). Start
with  $V_0(\pi) = -\Cb(\pi,1)$
 in (\ref{eq:vi}). Clearly this is MLR decreasing on $\I$ and therefore on polytope  $\mathcal{P}_{Y+1}$ since
 $\f$ is chosen with increasing elements, see (\ref{eq:cp1}).
Now for the inductive step:  Assume at iteration $k$, $V_k(\pi)$ is MLR decreasing on  polytope  $\mathcal{P}_{Y+1}$.
Then since $\T(\pi,a) $ is MLR increasing in $a$ (Theorem \ref{thm:key}(4)) and $T(\pi,a) \in \mathcal{P}_{Y+1}$ by (C2), it follows that 
$V_k(\T(\pi,1)) \geq V_k(\T(\pi,2))$.  

Consider any $\pi\gr \bp \in  \mathcal{P}_{Y+1}$. 
Since $\sigma(\pi,.) \gs \sigma(\bp,.) $ (see Theorem \ref{thm:key}(3)), 
\beq \sum_a V_k(\T(\pi,a)) \sigma(\pi,a) \leq  \sum_a V_k(\T(\pi,a)) \sigma(\bp,a)  \label{eq:qbar}\eeq
Next since $\pi \gr \bp \implies \T(\pi,a) \gr \T(\bp,a)$ (Theorem \ref{thm:key}(2)),
so $V_k(\pi)$ MLR decreasing in $\pi$ implies $V_k(\T(\pi,a)) \leq V_k( \T(\bp,a))$. So from (\ref{eq:qbar}),  $\pi \gr \bp$ implies
\beq\sum_a V_k(\T(\pi,a)) \sigma(\pi,a) \leq  \sum_a V_k(\T(\pi,a)) \sigma(\bp,a) \leq  \sum_a V_k(\T(\bp,a)) \sigma(\bp,a)  
\label{eq:qbar2}
\eeq

From (A3),  $C(\pi,2)$ is MLR decreasing. So  $\pi \gr \bp$ implies $C(\pi,2)\leq C(\bp,2)$.
Therefore $\pi \gr \bp$ implies   $Q_{k+1}(\pi,2) \leq Q_{k+1}(\bp,2) $. 
Thus $\min_u Q_{k+1}(\pi,u) \leq \min_u Q_{k+1}(\bp,u)$, i.e., $V_{k+1}(\pi) \leq V_{k+1}(\bp)$. This completes the induction step.
Finally, since $V_k \rightarrow V$ as $k\rightarrow \infty$ pointwise (see discussion below (\ref{eq:vi})), $V$ is  MLR decreasing on 
polytope  $\mathcal{P}_{Y+1}$.

{\bf Part 2}: Under the above conditions,  $\mu^*(\pi)$ is MLR increasing on polytope $\mathcal{P}_{Y+1}$.
It suffices to show that
 $Q(\pi,u)$ is submodular (see Definition \ref{def:supermod}) on $\mathcal{P}_{Y+1}$ wrt the MLR ordering since then 
Theorem \ref{res:monotone} applies implying that $\mu^*(\pi)$ is MLR decreasing in $\pi \in \mathcal{P}_{Y+1}$.
To show that $Q(\pi,u)$ in (\ref{eq:dp_alg}) is submodular, we need to show that 
 $Q(\pi,2) $ is MLR decreasing in $\pi$. But this follows from (A3) and Part 1.
 Thus from Theorem \ref{res:monotone}, (\ref{eq:structure}) holds.

\subsection{Proof of Theorem \ref{thm:dep}}  \label{app:dep}
Given any $\pi_1,\pi_2 \in \l(e_X,\bp)$ with $\pi_2\glX \pi_1$, we need to prove:
$\mu_\theta(\pi_1) \leq \mu_\theta(\pi_2)$ iff  $\theta(X-2) \geq 1 $, $\theta(i) \leq \theta(X-2)$ for $i< X-2$.
But from the structure of (\ref{eq:linear}), obviously $\mu_\theta(\pi_1) \leq \mu_\theta(\pi_2)$ 
is equivalent to 
$ \begin{bmatrix} 0 & 1 & \theta^\p\end{bmatrix}^\p \begin{bmatrix}\pi_1 \\ -1 \end{bmatrix} 
\leq \begin{bmatrix} 0 & 1 & \theta^\p\end{bmatrix}^\p \begin{bmatrix}\pi_2 \\ -1 \end{bmatrix}$,
or equivalently, 
$\begin{bmatrix} 0 & 1 & \theta(1) & \cdots & \theta(X-2)\end{bmatrix} (\pi_1 - \pi_2 ) \leq 0$.

Now from Lemma \ref{lem:convex}(iii),  $\pi_2 \glX\pi_1 $ implies that
 $\pi_1 = \epsilon_1 e_{X} + (1-\epsilon_1) \bp$,
$\pi_2 = \epsilon_2 e_{X} + (1-\epsilon_2) \bp$ and $\epsilon_1 \leq \epsilon_2$.
Substituting these into the above expression, we need to prove
$$
( \epsilon_1 -  \epsilon_2)\bigl(\theta(X-2) -  \begin{bmatrix} 0 & 1 & \theta(1) & \cdots & \theta(X-2)\end{bmatrix}^\p \bp \bigr) \leq 0, \quad \forall \bp \in \H_X$$  iff  $\theta(X-2) \geq 1 $, $\theta(i) \leq \theta(X-2)$, $i < X-2$. This is obviously true.

A similar proof shows that  on lines $\l(e_1,\bp)$ the linear
threshold policy satisfies
$\mu_\theta(\pi_1) \leq \mu_\theta(\pi_2)$ iff  $\theta(i) \geq 0$ for $i < X-2$.

\begin{biography}[]{Vikram Krishnamurthy}
(S'90-M'91-SM'99-F'05) was born in 1966.  He received his bachelor's
degree from the University of Auckland, New Zealand in 1988, and
Ph.D.\ from the Australian National University in 1992.
He currently is  a professor and Canada Research Chair at the
Department of Electrical Engineering, University of British Columbia,
Vancouver, Canada. 

Dr Krishnamurthy's  current research interests include computational game theory,
stochastic control in sensor networks, and stochastic
dynamical systems for
modeling of biological ion channels and biosensors.
Dr. Krishnamurthy currently serves as Editor in Chief of IEEE Journal Selected Topics
in Signal Processing. He has served as associate editor for several journals
including IEEE Transactions Automatic Control and 
IEEE Transactions on Signal Processing. In 2009-2010, he served as Distinguished lecturer for the IEEE signal
processing society.
\end{biography}


\begin{thebibliography}{10}

\bibitem{AO10}
D.~Acemoglu and A.~Ozdaglar.
\newblock Opinion dynamics and learning in social networks.
\newblock {\em Dynamic Games and Applications}, pages 1--47, 2010.

\bibitem{AOT10}
D.~Acemoglu, A.~Ozdaglar, and A.~Tahbaz-Salehi.
\newblock Cascades in networks and aggregate volatility.
\newblock Technical report, National Bureau of Economic Research, 2010.

\bibitem{ASS02}
I.~F. Akyildiz, W.~Su, Y.~Sankarasubramaniam, and E.~Cayirci.
\newblock Wireless sensor networks: A survey.
\newblock {\em Computer Networks}, 38(4):393--422, 2002.

\bibitem{Ami05}
R.~Amir.
\newblock Supermodularity and complementarity in economics: An elementary
  survey.
\newblock {\em Southern Economic Journal}, 71(3):636--660, 2005.

\bibitem{AT96}
M.S. Andersland and D.~Teneketzis.
\newblock Measurement scheduling for recursive team estimation.
\newblock {\em Journal of Optimization Theory and Applications},
  89(3):615--636, June 1996.

\bibitem{Ban92}
A.~Banerjee.
\newblock A simple model of herd behavior.
\newblock {\em Quaterly Journal of Economics}, 107:797--817, 1992.

\bibitem{BB02}
P.~Bartlett and J.~Baxter.
\newblock Estimation and approximation bounds for gradient-based reinforcement
  learning.
\newblock {\em J. Comput. Syst. Sci.}, 64(1):133--150, 2002.

\bibitem{BN93}
M.~Basseville and I.V. Nikiforov.
\newblock {\em Detection of Abrupt Changes --- Theory and Applications}.
\newblock Information and System Sciences Series. Prentice Hall, New Jersey,
  USA, 1993.

\bibitem{Ber00b}
D.P. Bertsekas.
\newblock {\em Dynamic Programming and Optimal Control}, volume 1 and 2.
\newblock Athena Scientific, Belmont, Massachusetts, 2000.

\bibitem{BHW92}
S.~Bikchandani, D.~Hirshleifer, and I.~Welch.
\newblock A theory of fads, fashion, custom, and cultural change as information
  cascades.
\newblock {\em Journal of Political Economy}, 100:992--1026, 1992.

\bibitem{BV02}
L.~Bru and X.~Vives.
\newblock Informational externalities, herding, and incentives.
\newblock {\em Journal of Institutional and Theoretical Economics JITE},
  158(1):91--105, 2002.

\bibitem{Cha04}
C.~Chamley.
\newblock {\em Rational herds: Economic models of social learning}.
\newblock Cambridge, 2004.

\bibitem{CG94}
C.~Chamley and D.~Gale.
\newblock Information revelation and strategic delay in a model of investment.
\newblock {\em Econometrica}, 62(5):1065--1085, 1994.

\bibitem{CZK07}
Y.~Chen, Q.~Zhao, V.~Krishnamurthy, and D.~Djonin.
\newblock Transmission scheduling for optimizing sensor network lifetime: A
  stochastic shortest path approach.
\newblock {\em IEEE Trans.\ Signal Proc.}, 55(5):2294--2309, May 2007.

\bibitem{CH70}
T.~Cover and M.~Hellman.
\newblock The two-armed-bandit problem with time-invariant finite memory.
\newblock {\em Information Theory, IEEE Transactions on}, 16(2):185--195, 1970.

\bibitem{Gan60}
F.R. Gantmacher.
\newblock {\em Matrix Theory}, volume~2.
\newblock Chelsea Publishing Company, New York, 1960.

\bibitem{Har05}
S.~Hart.
\newblock {Adaptive heuristics}.
\newblock {\em Econometrica}, 73(5):1401--1430, 2005.

\bibitem{HM00}
S.~Hart and A.~Mas-Colell.
\newblock A simple adaptive procedure leading to correlated equilibrium.
\newblock {\em Econometrica}, 68(5):1127--1150, 2000.

\bibitem{HC70}
M.E. Hellman and T.M. Cover.
\newblock Learning with finite memory.
\newblock {\em The Annals of Mathematical Statistics}, pages 765--782, 1970.

\bibitem{HL96}
O.~Hern\'{a}ndez-Lerma and J.~Bernard Laserre.
\newblock {\em Discrete-Time {M}arkov Control Processes: Basic Optimality
  Criteria}.
\newblock Springer-Verlag, New York, 1996.

\bibitem{HS84}
D.P. Heyman and M.J. Sobel.
\newblock {\em Stochastic Models in Operations Research}, volume~2.
\newblock McGraw-Hill, 1984.

\bibitem{Kar68}
S.~Karlin.
\newblock {\em Total Positivity}, volume~1.
\newblock Stanford Univ., 1968.

\bibitem{KR80}
S.~Karlin and Y.~Rinott.
\newblock Classes of orderings of measures and related correlation
  inequalities. {I}. {M}ultivariate totally positive distributions.
\newblock {\em Journal of Multivariate Analysis}, 10:467--498, 1980.

\bibitem{KS60}
J.G. Kemeny and J.L. Snell.
\newblock {\em Finite {M}arkov Chains}.
\newblock Van Nostrand, New York, 1960.

\bibitem{Kij97}
M.~Kijima.
\newblock {\em Markov Processes for Stochastic Modelling}.
\newblock Chapman and Hall, 1997.

\bibitem{Kri11}
V.~Krishnamurthy.
\newblock Bayesian sequential detection with phase-distributed change time and
  nonlinear penalty -- a lattice programming pomdp approach.
\newblock {\em IEEE Trans.\ Inform.\ Theory}, 57(3), Oct. 2011.
\newblock http://arxiv.org/abs/1011.5298.

\bibitem{KMY08}
V.~Krishnamurthy, M.~Maskery, and G.~Yin.
\newblock Decentralized activation in a {Z}ig{B}ee-enabled unattended ground
  sensor network: A correlated equilibrium game theoretic analysis.
\newblock {\em IEEE Trans.\ Signal Proc.}, 56(12):6086--6101, December 2008.

\bibitem{KW09}
V.~Krishnamurthy and B.~Wahlberg.
\newblock {POMDP} multiarmed bandits -- structural results.
\newblock {\em Mathematics of Operations Research}, 34(2):287--302, May 2009.

\bibitem{KY02}
V.~Krishnamurthy and G.~Yin.
\newblock Recursive algorithms for estimation of hidden {M}arkov models and
  autoregressive models with {M}arkov regime.
\newblock {\em IEEE Trans.\ Inform.\ Theory}, 48(2):458--476, February 2002.

\bibitem{KY03}
H.J. Kushner and G.~Yin.
\newblock {\em Stochastic Approximation Algorithms and Recursive Algorithms and
  Applications}.
\newblock Springer-Verlag, 2nd edition, 2003.

\bibitem{LADO07}
I.~Lobel, D.~Acemoglu, M.~Dahleh, and A.E. Ozdaglar.
\newblock Preliminary results on social learning with partial observations.
\newblock In {\em Proceedings of the 2nd International Conference on
  Performance Evaluation Methodolgies and Tools}, Nantes, France, 2007. ACM.

\bibitem{Lov87a}
W.S. Lovejoy.
\newblock On the convexity of policy regions in partially observed systems.
\newblock {\em Operations Research}, 35(4):619--621, July-August 1987.

\bibitem{Lov87}
W.S. Lovejoy.
\newblock Some monotonicity results for partially observed {M}arkov decision
  processes.
\newblock {\em Operations Research}, 35(5):736--743, Sept.-Oct. 1987.

\bibitem{Mou86}
G.B. Moustakides.
\newblock Optimal stopping times for detecting changes in distributions.
\newblock {\em Annals of Statistics}, 14:1379--1387, 1986.

\bibitem{MS02}
A.~Muller and D.~Stoyan.
\newblock {\em Comparison Methods for Stochastic Models and Risk}.
\newblock Wiley, 2002.

\bibitem{Neu89}
M.F. Neuts.
\newblock {\em Structured stochastic matrices of M/G/1 type and their
  applications}.
\newblock Marcel Dekker, N.Y., 1989.

\bibitem{PPA89}
J.~Papastravrou, J.~Pothiawala, and M.~Athans.
\newblock Designing an organization in a hypothesis testing framework.
\newblock Technical report, Laboratory for Information and Decision Systems,
  MIT, 1989.

\bibitem{PS11}
A.~Park and H.~Sabourian.
\newblock Herding and contrarian behavior in financial markets.
\newblock {\em Econometrica}, 79(4):973--1026, 2011.

\bibitem{PH08}
H.V. Poor and O.~Hadjiliadis.
\newblock {\em Quickest Detection}.
\newblock Cambridge, 2008.

\bibitem{Rie91}
U.~Rieder.
\newblock Structural results for partially observed control models.
\newblock {\em Methods and Models of Operations Research}, 35:473--490, 1991.

\bibitem{Ros83}
S.~Ross.
\newblock {\em Introduction to Stochastic Dynamic Programming}.
\newblock Academic Press, San Diego, California., 1983.

\bibitem{RIM09}
S.~Ross, M.~Izadi, M.~Mercer, and D.~Buckeridge.
\newblock Sensitivity analysis of {POMDP} value functions.
\newblock In {\em IEEE International Conference on Machine Learning and
  Applications ICMLA'09}, 2009.

\bibitem{Shi63}
AN~Shiryaev.
\newblock {On optimum methods in quickest detection problems}.
\newblock {\em Theory of Probability and its Applications}, 8:22, 1963.

\bibitem{Shi78}
A.N. Shiryayev.
\newblock {\em Optimal stopping rules}.
\newblock Springer-Verlag, 1978.

\bibitem{SS97}
L.~Smith and P.~Sorensen.
\newblock Informational herding and optimal experimentation.
\newblock Economics Papers 139, Economics Group, Nuffield College, University
  of Oxford, 1997.

\bibitem{SS00}
L.~Smith and P.~Sorenson.
\newblock Pathological outcomes of observational learning.
\newblock {\em Econometrica}, 68(2):371--398, 2000.

\bibitem{Spa03}
J.~Spall.
\newblock {\em Introduction to Stochastic Search and Optimization}.
\newblock Wiley, 2003.

\bibitem{TV05}
A.G. Tartakovsky and V.V. Veeravalli.
\newblock General asymptotic {B}ayesian theory of quickest change detection.
\newblock {\em Theory of Probability and its Applications}, 49(3):458--497,
  2005.

\bibitem{Top98}
D.M. Topkis.
\newblock {\em Supermodularity and Complementarity}.
\newblock Princeton University Press, 1998.

\bibitem{WD80}
CC~White and DP~Harrington.
\newblock {Application of Jensen's inequality to adaptive suboptimal design}.
\newblock {\em Journal of Optimization Theory and Applications}, 32(1):89--99,
  1980.

\bibitem{Whi79}
W.~Whitt.
\newblock A note on the influence of the sample on the posterior distribution.
\newblock {\em Journal American Statistical Association}, 74:424--426, 1979.

\bibitem{Yak97}
B.~Yakir.
\newblock A note on optimal detection of a change point in distribution.
\newblock {\em Annals of Statistics}, 25:2117--2126, 1997.

\bibitem{YKP99}
B.~Yakir, A.M. Krieger, and M.~Pollak.
\newblock {Detecting a change in regression: First-order optimality}.
\newblock {\em Annals of Statistics}, 27(6):1896--1913, 1999.

\end{thebibliography}
\end{document}